\documentclass{amsart}

\usepackage{amsmath,amssymb,amsfonts,amsthm}
\usepackage[all]{xy}
\usepackage{color}

\usepackage{rotating}
\usepackage{fullpage}

\def\endofproof {\hfill{$\Box$}\\}

\renewcommand{\(}{\begin{equation}}
\renewcommand{\)}{\end{equation}}
\newcommand{\bea}{\begin{eqnarray}}
\newcommand{\eea}{\end{eqnarray}}

\newcommand{\Z}{{\mathbb Z}}

\theoremstyle{plain}
\newtheorem{theorem}{Theorem}[subsection]

\newtheorem{proposition}[theorem]{Proposition}
\newtheorem{corollary}[theorem]{Corollary}

\theoremstyle{definition}
\newtheorem{definition}[theorem]{Definition}
\newtheorem{example}[theorem]{Example}

\newtheorem{remark}[theorem]{Remark}

\numberwithin{equation}{subsection}

\begin{document}

\title{Extended higher cup-product Chern-Simons theories}

\author{Domenico Fiorenza}
\address{Dipartimento di Matematica ``Guido Castelnuovo'' ---
Universit\`a degli Studi di Roma ``la Sapienza'' --- P.le Aldo Moro, 2 --
00185 -- Roma, Italy -- {\tt fiorenza@mat.uniroma1.it}}

\author{Hisham Sati}
\address{Math Department --- University of Pittsburgh ---
301 Thackeray Hall -- Pittsburgh -- PA  15260 -- {\tt hsati@pitt.edu}}

\author{Urs Schreiber}
\address{Mathematical Institute ---
Utrecht University --- 80125 --
3508 TC-- Utrecht, The Netherlands -- {\tt urs.schreiber@googlemail.com}}

\begin{abstract}
 The
  proper action functional of $(4k+3)$-dimensional $U(1)$-Chern-Simons theory
  including the instanton sectors has a well known description: it is   
  given on the moduli space of fields by the fiber integration 
  of the cup product square of classes in degree-$(2k+2)$ differential
  cohomology.
  We first refine this statement from the moduli space to 
  the full higher smooth moduli stack of fields,  to which the higher-order-ghost 
  BRST complex is the infinitesimal approximation.
  Then we generalize the refined formulation to cup product Chern-Simons theories of 
  nonabelian and higher nonabelian gauge fields,
  such as the nonabelian $\mathrm{String}^{\mathbf{c}}$-2-connections appearing in quantum-corrected
  11-dimensional supergravity and M-branes.  
  We discuss aspects of the 
  off-shell extended geometric pre-quantization
 (in the sense of extended or multi-tiered QFT) 
 of these theories, 
 where there is a prequantum 
 $U(1)$-$k$-bundle (equivalently: a $U(1)$-$(k\!-\!1)$-bundle gerbe) in each codimension $k$.
 Examples we find include moduli stacks for differential T-duality structures 
 as well as the anomaly line bundles of 
 higher electric/magnetic charges, such as the 5-brane charges 
 appearing in heterotic supergravity, appearing
 as line bundles with connection on the smooth higher moduli stacks of field configurations.
\end{abstract}

\keywords{
Chern-Simons theory; topological field theories;
differential cohomology; M-theory; M-branes; geometric quantization.}
\subjclass[2010]{81T45 (Primary); 53C08, 53C80, 81T50, 53D50 (Secondary).}

\maketitle


\tableofcontents


\section{Introduction and Overview}
  \label{Introduction}
  
It has become a familiar fact, known from various examples, that there should be an $n$-dimensional
topological quantum field theory $Z_{[c]}$ associated to the following data:
\begin{enumerate}
  \item a \emph{gauge group} $G$: a Lie group such as $U(N)$; or more generally 
  a higher smooth group, such as the smooth \emph{circle $m$-group} $\mathbf{B}^{m-1}U(1)$
  or the \emph{String 2-group} or the smooth \emph{Fivebrane 6-group}
  \cite{SSSIII, FSS};
  \item a universal characteristic class $[c]\in H^{n+1}(BG,\mathbb{Z})$
  and/or its image $\omega$ in real/de Rham cohomology.
\end{enumerate}  
The tqft $Z_{[c]}$ is a 
$G$-gauge theory defined naturally over all closed oriented
$n$-dimensional smooth manifolds $\Sigma_n$, and its topological nature is expressed by the fact that whenever $\Sigma_n$
happens to be the boundary of some manifold $\Sigma_{n+1}$ then the action fuctional 
on a field configuration $\phi$ can be written as the integral over $\Sigma_{n+1}$ of the pullback $\hat \phi^* \omega$ of a certain universal $(n+1)$-form $\omega$, for some extension $\hat \phi$ of $\phi$.
These are \emph{Chern-Simons type} gauge theories. 
See 
\cite{DJT} for the emergence of Chern-Simons theories in physics and 
also \cite{Za} for a recent introduction to the general idea of 
such theories. 

\medskip
Notably for 
$G$ a compact connected and simply connected simple Lie group, for
$c \in H^4(B G, \mathbb{Z})\simeq \mathbb{Z}$ any integer -- the ``level'' --,
this quantum field theory is the original and standard Chern-Simons theory
introduced in \cite{WittenCS}. See \cite{FreedCS} for a comprehensive review. 
Familiar as this theory is, there is an interesting aspect of it that has maybe not yet found the attention it deserves, 
and which is an example of our constructions here. 

\medskip
To motivate this, it is helpful to 
look at the 3d Chern-Simons action functional as follows:
if we write $H(\Sigma_3, \mathbf{B}G_{\mathrm{conn}})$ for the set of gauge equivalence classes of $G$-principal connections on $\Sigma_3$, then the (exponentiated) action functional of 3d Chern-Simons theory over $\Sigma_3$ is a function of sets
$$
  \exp(i S(-)) : H(\Sigma_3, \mathbf{B}G_{\mathrm{conn}}) \to U(1)
  \,.
$$
This function acts as follows: since the classifying space $BG$ is homotopically trivial in degree less than or equal to $3$, any principal $G$-bundle on $\Sigma_3$ can be trivialized, and so for any gauge 
equivalence class of $G$-connections there exists a representative given by a smooth $\mathfrak{g}$-valued 1-form $A$ on $\Sigma_3$; the  action functional
sends $A$ to the element $\exp(2 \pi i k \int_{\Sigma_3} \mathrm{CS}(A))$ in $U(1)$, 
where $\mathrm{CS}(A) \in \Omega^3(\Sigma_3)$ is the Chern-Simons 3-form of $A$ \cite{CS},
which gives the whole theory its name. 
That this is indeed well defined is the fact that for every gauge 
transformation $g : A \stackrel{}{\to} A^g$, for $g \in C^\infty(\Sigma_3, G)$, both $A$ 
as well as its gauge transform $A^g$, are sent to the same element of $U(1)$. 
A natural formal way to express this is to consider the \emph{groupoid}
$\mathbf{H}(\Sigma_3, \mathbf{B}G_{\mathrm{conn}})$ whose objects are gauge fields $A$ and whose 
morphisms are gauge transformations $g$ as above. Then the fact that the Chern-Simons action is
defined on individual gauge field configurations while being invariant under gauge transformations
is equivalent the statement that it is a \emph{functor}, hence a morphism of groupoids, 
$$
  \exp(i S(-)) : \mathbf{H}(\Sigma_3, \mathbf{B}G_{\mathrm{conn}}) \to U(1)
  \,,
$$
where the set underlying $U(1)$ is regarded as a groupoid with only identity morphisms.
Hence the fact that $\exp(i S(-))$ has to send every morphism on the left to a
morphism on the right is the gauge invariance of the action.

\medskip
Furthermore, the action functional has the property of being \emph{smooth}. It takes any \emph{smooth family} of gauge fields, over some parameter space $U$, to a corresponding smooth family of elements of $U(1)$, and these assignments are compatible with pullback along smooth functions $U_1 \to U_2$ between parameter spaces. 
The formal language that expresses this concept is that of 
\emph{stacks on the site of smooth manifolds}
(see section \ref{GeneralTheory} below for a review and pointers to the literature): 
to say that for every $U$ there is a groupoid, as above, of smooth $U$-families of gauge fields and smooth $U$-families of gauge transformations between them, in a consistent way, is to say that there is a \emph{smooth moduli stack}, denoted $[\Sigma_3,\mathbf{B}G_{\mathrm{conn}}]$, of gauge fields on $\Sigma_3$ -- 
the mapping stack\footnote{A more detailed discusison of such mapping stacks and their
role as moduli stacks is in section 2.3 of \cite{hgp}.}. 
Finally, the fact that the Chern-Simons action functional is not only gauge invariant but also smooth is the fact that it refines to a morphism of smooth stacks
$$
  \exp(i \mathbf{S}(-)) : [\Sigma_3, \mathbf{B}G_{\mathrm{conn}}] \to U(1)
  \,,
$$
where now $U(1)$ is regarded as a smooth stack by declaring that a smooth family of elements is a smooth function with values in $U(1)$. 
This is the refined perspective on Chern-Simons theory which we make use of here.
A pedagogical introduction that walks through the example of 1- dimensional 
and 3-dimensional Chern-Simons theory in this fashion is in the first part of 
\cite{FiorenzaSatiSchreiberCS}.

\medskip
To see what this means, 
it is useful to think of a smooth stack simply as being a \emph{smooth groupoid}.
Lie groups and Lie groupoids are examples (and are called ``differentiable stacks''
when regarded as special cases of smooth stacks)
but there are important smooth groupoids which are not Lie groupoids, since they do not have an
underlying smooth \emph{manifold} but a more general smooth space of objects and of morphisms.
Just as Lie groups have an infinitesimal approximation given by Lie algebras, so 
smooth stacks/smooth groupoids have an infinitesimal approximation given by
\emph{Lie algebroids}. For instance, the smooth moduli stack 
$[\Sigma_3, \mathbf{B}G_{\mathrm{conn}}]$ of gauge field configuration on $\Sigma_3$
(or rather its concretification, see \cite{hgp}) 
is best known in the physics literature in the guise of its 
underlying Lie algebroid: this is the formal dual of the
(off-shell) \emph{BRST complex} of the $G$-gauge theory on $\Sigma_3$: in degree 0
this consists of the functions on the space of gauge fields on $\Sigma_3$, and in 
degree 1 it consists of functions on infinitesimal gauge transformations between these:
the ``ghost fields''.

\medskip
The smooth structure on the action functional is of course crucial in 
field theory: in particular 
it allows to define the \emph{differential} $d \exp(i \mathbf{S}(-))$ of the action functional and hence its critical locus, characterized by   
the Euler-Lagrange equations of motion. In particular, if $\Sigma_3=\Sigma_2\times [0,1]$ is a 3-dimensional worldvolume swept out by a 2-dimensional surface $\Sigma_2$, then the groupoid of critical field configurations on $\Sigma_3$
is equivalent to the groupoid of their initial values (i.e., of their
restrictions to $\Sigma_2 \times\{0\}$), since the equation of motion uniquely
determine the extension of these to the whole cylinder $\Sigma_3$. This gives the \emph{phase space}
of the theory, which is the substack
$$
  [\Sigma_2, \flat \mathbf{B}G] \hookrightarrow [\Sigma_2, \mathbf{B}G_{\mathrm{conn}}]
$$
consisting of flat $G$-connections on $\Sigma_2$. The phase space has a natural structure of symplectic manifold; moreover, its canonical symplectic form is the restriction to $ [\Sigma_2, \flat \mathbf{B}G]$ of a presymplectic 2-form defined on the whole of $[\Sigma_2, \mathbf{B}G_{\mathrm{conn}}]$. 
To formalize 
this, write $\Omega^2_{\mathrm{cl}}(-)$ for the 
smooth stack of closed 2-forms (without gauge transformations), sending
a parameter manifold $U$ to the set $\Omega^2_{\mathrm{cl}}(U)$ of smooth closed 2-forms on $U$.
This may be regarded as the \emph{smooth moduli 0-stack} of closed 2-forms: for
every smooth manifold $X$ (considered as a smooth groupoid with only trivial morphisms) the set of morphisms $X \to \Omega^2_{\mathrm{cl}}(-)$
is in natural bijection to the set $\Omega^2_{\mathrm{cl}}(X)$ of closed 2-forms on $X$.
This is an instance of the \emph{Yoneda lemma}.
Similarly, a smooth 2-form on the moduli stack of field configurations
is a morphism of smooth stacks of the form
$$
  [\Sigma_2, \mathbf{B}G_{\mathrm{conn}}] \to \Omega^2_{\mathrm{cl}}(-)
  \,.
$$
Explicitly, for Chern-Simons theory this morphism sends, for each smooth parameter space $U$, a
given smooth $U$-family of gauge fields 
$A  \in \Omega^1(\Sigma_2 \times U, \mathfrak{g})$
to the 2-form
$$
  \int_{\Sigma_2} \langle d_U A \wedge d_U A\rangle \in \Omega^2_{\mathrm{cl}}(U)
  \,.
$$
Notice that if we restrict to \emph{genuine} families $A$, i.e., to $\mathfrak{g}$-valued 1-forms on $\Sigma_2\times U$ vanishing on tangent vectors to $U$ (to be thought of as $\mathfrak{g}$-valued 1-forms on $\Sigma_2$ parametrized by $U$; technically these are elements in the
\emph{concretification} of the moduli stack) then this 2-form is the 
\emph{fiber integral}
of the Poincar{\'e} 2-form $\langle F_A \wedge F_A\rangle$ along the projection
$\Sigma_2 \times U \to U$, where $F_A := d A + \tfrac{1}{2}[A \wedge A]$ is the 
curvature 2-form of $A$. This is the first sign of a general pattern, 
which we highlight in a moment.

\medskip
There is a more fundamental smooth moduli stack equipped with a closed 2-form:
the moduli stack $\mathbf{B}U(1)_{\mathrm{conn}}$ of $U(1)$-gauge fields,
hence of smooth circle bundles with connection. This is the rule
that sends a smooth parameter manifold $U$ to the groupoid 
$\mathbf{H}(U, \mathbf{B}U(1)_{\mathrm{conn}})$ of $U(1)$-gauge fields $\nabla$ on $U$,
which we have already seen above. Since the  curvature 2-form 
$F_\nabla \in \Omega^2_{\mathrm{cl}}(U)$ of a $U(1)$-principal connection on $U$
 is gauge invariant, the assignment
$\nabla \mapsto F_\nabla$ gives rise to a morphism of smooth stacks of the form
$$
  F_{(-)} : \mathbf{B}U(1)_{\mathrm{conn}} \to \Omega^2_{\mathrm{cl}}(-)
  \;.
$$
In terms of this morphism the fact that every $U(1)$-gauge field $\nabla$
on some space $X$
has an underlying field strength 2-form $\omega$ is expressed by the
existence of a commuting diagram of smooth stacks of the form
$$
  \raisebox{20pt}{
  \xymatrix{
    & \mathbf{B}U(1)_{\mathrm{conn}}
	\ar[d]^{F_{(-)}} & \mbox{gauge field / differential cocycle}
    \\
    X
	\ar[r]^\omega 
	\ar[ur]^{\nabla}
	&
	\Omega^2_{\mathrm{cl}}(-) & \mbox{ field strength / curvature}\;.
  }
  }
$$
Conversely, if we regard the bottom morphism $\omega$ as given, and regard this 
closed 2-form as a (pre)symplectic form, then a \emph{choice of lift} $\nabla$
in this diagram is a choice of refinement of the 2-form by a circle bundle with connection,
hence the choice of a \emph{prequantum circle bundle} in the language of 
geometric quantization (see for instance section II in \cite{Brylinski} for a review).

\medskip
Applied to the case of Chern-Simons theory this means that a smooth (off-shell)
prequantization of the theory is a choice of dashed morphism in a diagram
of smooth stacks of the form
$$
  \raisebox{20pt}{
  \xymatrix{
    && \mathbf{B}U(1)_{\mathrm{conn}}
	\ar[d]^{F_{(-)}}
    \\
    [\Sigma_2, \mathbf{B}G_{\mathrm{conn}}]
	\ar[rr]_-{\int_{\Sigma_2}\langle F_{(-)},F_{(-)}\rangle}
	\ar@{-->}[urr]
	&&
	\Omega^2_{\mathrm{cl}}(-)~~.
  }
  }
$$
Similar statements apply to on-shell geometric (pre)quantization of Chern-Simons theory,
which has been so successfully applied in the original article
\cite{WittenCS}; see also \cite{GomiCS}.
In summary, this means that in the context of smooth stacks 
the Chern-Simons action functional and its prequantization are as in the following table:

\medskip
\medskip
\hspace{1cm}
\begin{tabular}{|c|c|c|}
  \hline
  {\bf dimension} & {\bf description} & {\bf moduli stack description }
  \\
  \hline
  $k = 3$ & action functional (0-bundle) 
  & $ \exp(i \mathbf{S} (-)) : [\Sigma_3, \mathbf{B}G_{\mathrm{conn}}] 
   \to \;\;\;U(1)_{\;\;\;\;\;\;\;}$
  \\
  \hline
  $k = 2$ & prequantum circle 1-bundle 
  & \hspace{2.0cm}$
  [\Sigma_2, \mathbf{B}G_{\mathrm{conn}}]  
  \to
  \mathbf{B} U(1)_{\mathrm{conn}}
  $
  \\
  \hline
\end{tabular}

\medskip
\medskip
There is a precise sense, discussed in section \ref{CircleBundleModuli} below,
in which a $U(1)$-valued function is a \emph{circle $k$-bundle with connection}
for $k = 0$. If we furthermore regard an ordinary $U(1)$-principal bundle
as a \emph{circle 1-bundle} then this table says that in dimension $k$
Chern-Simons theory appears as a \emph{circle $(3-k)$-bundle with connection}
-- at least for $k = 3$ and $k = 2$.

\medskip
Formulated this way, this should remind one of what is called an
\emph{extended} or \emph{multi-tiered} topological quantum field theory
(formalized and classified in \cite{LurieTQFT}) 
which is the full formalization of \emph{locality}
in the Schr{\"o}dinger picture of quantum field theory. This roughly says that, \emph{after quantization}, an 
$n$-dimensional topological field theory should be a rule that 
to a closed manifold of dimension $k$ assigns an $(n-k)$-categorical
analog of a vector space of quantum states. 
Since ordinary geometric quantization of Chern-Simons theory
assigns to a closed $\Sigma_2$ the vector space of \emph{polarized sections}
(holomorphic sections) of the line bundle associated with the above circle 1-bundle, 
this suggests that there should be an \emph{extended} or \emph{multi-tiered}
refinement of geometric (pre)quantization of Chern-Simons theory,
which to a closed oriented manifold of dimension $0 \leq k \leq n$ assigns
a \emph{prequantum circle $(n-k)$-bundle} 
(bundle $(n-k-1)$-gerbe) on the moduli stack of 
field configurations over $\Sigma_k$. In other words, one expects  a natural morphism 
$[\Sigma_k, \mathbf{B}G_{\mathrm{conn}}] \to \mathbf{B}^{(n-k)} U(1)_{\mathrm{conn}}$
to the moduli $(n-k)$-stack of circle $(n-k)$-bundles with connection
(details on this are given below in section \ref{CircleBundleModuli}).

\medskip
In particular for $k = 0$ and $\Sigma_0$ connected, hence for $\Sigma_0$ the space $*$ consisting of a single point,
we have that the moduli stack of fields on $\Sigma_0$ is the 
\emph{universal} moduli stack itself:
$[*, \mathbf{B}G_{\mathrm{conn}}] \simeq \mathbf{B}G_{\mathrm{conn}}$.
Therefore, a  \emph{fully extended prequantization} of 
3-dimensional $G$-Chern-Simons theory would have to involve a 
\emph{universal characteristic} morphism
$$
  \mathbf{c}_{\mathrm{conn}} : \mathbf{B}G_{\mathrm{conn}}
  \to 
  \mathbf{B}^3 U(1)_{\mathrm{conn}}
$$
of smooth moduli stacks, hence a smooth circle 3-bundle with connection
on the universal moduli stack of $G$-gauge fields. This indeed naturally
exists: an explicit construction is given in \cite{FSS}.
This morphism of smooth higher stacks is a differential refinement of a smooth refinement
of the level itself: forgetting the connections and only remembering the 
underlying (higher) gauge bundles, we still have a morphism of smooth higher stacks
$$
  \mathbf{c}: \mathbf{B}G \to \mathbf{B}^3 U(1)
  \,.
$$
This expression should remind one of the continuous map of topological 
spaces
$$
  c : B G \to B^3 U(1) \simeq K(\mathbb{Z},4)
$$
from the classifying space $B G$ to the Eilenberg-MacLane space $K(\mathbb{Z},4)$, which 
represents the level as a class in integral cohomology $H^4(B G, \mathbb{Z}) \simeq \mathbb{Z}$.
Indeed, there is a canonical topological realization \emph{derived functor} or \emph{$\infty$-functor}
$$
  {\vert-\vert} : \mathbf{H} \to \mathrm{Top}
$$
from smooth higher stacks to topological spaces \cite{cohesive}, 
derived left adjoint to the operation of forming \emph{locally constant higher stacks},
and under this map we have
$$
  \vert \mathbf{c}\vert \simeq c
  \,.
$$
In this sense $\mathbf{c}$ is a \emph{smooth refinement} of $[c] \in H^4(B G, \mathbb{Z})$
and then $\mathbf{c}_{\mathrm{conn}}$ is a further \emph{differential refinement}
of $\mathbf{c}$.

\medskip
However, more is true. Not only  there is an extension of the prequantization of 
3d $G$-Chern-Simons theory to the point, but this also induces 
the extended prequantization in every other dimension by \emph{tracing}:
for $0 \leq k \leq n$ and $\Sigma_k$ a closed and oriented smooth manifold,
there is a canonical morphism of smooth higher stacks of the form
$$
  \exp(2 \pi i \int_{\Sigma_k}(-)) : [\Sigma_k, \mathbf{B}^n U(1)_{\mathrm{conn}}]
  \to 
  \mathbf{B}^{n-k}U(1)_{\mathrm{conn}}\;,
$$
which refines the fiber integration of differential forms from curvature $(n+1)$-forms to their entire prequantum 
circle $n$-bundles (we discuss this below in section \ref{FibInt}).
Since, furthermore, the formation of mapping stacks $[\Sigma_k,-]$ is functorial, this means
that from a morphism $\mathbf{c}_{\mathrm{conn}}$ as above we get for every
$\Sigma_k$ a composite morphism as such:
$$
  \exp(2 \pi i \int_{\Sigma_k} [\Sigma_k, \mathbf{c}_{\mathrm{conn}}])
  :
  \xymatrix{
    [\Sigma_k, \mathbf{B}G_{\mathrm{conn}}]
	\ar[rr]^-{[\Sigma_k, \mathbf{c}_{\mathrm{conn}}]}
	&&
	[\Sigma_k, \mathbf{B}^n U(1)_{\mathrm{conn}}]
	\ar[rrr]^-{\exp(2 \pi i\int_{\Sigma_k}(-))}
	&&&
	\mathbf{B}^{n-k}U(1)_{\mathrm{conn}}
  }
  \,.
$$ 
For 3d $G$-Chern-Simons theory and $k = n = 3$ this composite 
\emph{is} the action functional of the theory. 
This is effectively
the perspective on ordinary Chern-Simons theory amplified in \cite{CJMSW}).
Therefore, for general $k$ we may speak of this as the 
\emph{extended action functional}, with values not in $U(1)$
but in $\mathbf{B}^{n-k}U(1)_{\mathrm{conn}}$.

\medskip
This way we find that the above table, containing the
Chern-Simons action functional together with its prequantum circle 1-bundle,
extends to the following table that goes all the way from dimension 3
down to dimension 0.\\

\vspace{3mm}
\hspace{-1.2cm}
{\small
\begin{tabular}{|c|c|c|c|}
  \hline
  {\bf dim.} & & {\bf prequantum $(3-k)$-bundle} & 
  \\
  \hline
  $k = 0$ & \begin{tabular}{c}differential fractional \\ first Pontrjagin\end{tabular} & 
  $\mathbf{c}_{\mathrm{conn}} : \mathbf{B}G_{\mathrm{conn}} 
  \to \mathbf{B}^3 U(1)_{\mathrm{conn}}$ &
  \cite{FSS}
  \\
  \hline
  $k = 1$ & \begin{tabular}{c} WZW \\ background B-field \end{tabular} &
  $ \xymatrix{ [S^1, \mathbf{B}G_{\mathrm{conn}}]
   \ar[rr]^-{[S^1,\mathbf{c}_{\mathrm{conn}}]} &&
    [S^1, \mathbf{B}^3 U(1)_{\mathrm{conn}}] 
	 \ar[rr]^-{\exp(2 \pi i \int_{S^1}(-))} && \mathbf{B}^2 U(1)_{\mathrm{conn}}
	 } $
	 &
	 \begin{tabular}{c}
	   \cite{InfinityWZW}
	 \end{tabular}
  \\
  \hline
  $k = 2$ & \begin{tabular}{c} off-shell CS \\ prequantum bundle \end{tabular} & 
  $ \xymatrix{ [\Sigma_2, \mathbf{B}G_{\mathrm{conn}}]
   \ar[rr]^-{[\Sigma_2,\mathbf{c}_{\mathrm{conn}}]} &&
    [\Sigma_2, \mathbf{B}^3 U(1)_{\mathrm{conn}}] 
	 \ar[rr]^-{\exp(2 \pi i \int_{\Sigma_2}(-))} && \mathbf{B}U(1)_{\mathrm{conn}}
	 }
  $
   & \cite{cohesive}
  \\
  \hline
  $k = 3$ & \begin{tabular}{c} 3d CS \\ action functional \end{tabular} & 
  $ \xymatrix{[\Sigma_3, \mathbf{B}G_{\mathrm{conn}}]
   \ar[rr]^-{[\Sigma_3,\mathbf{c}_{\mathrm{conn}}]}
     &&
    [\Sigma_3, \mathbf{B}^3 U(1)_{\mathrm{conn}}] 
	 \ar[rr]^-{\exp(2 \pi i \int_{\Sigma_3}(-))} && U(1)
	 }
  $  &
  \cite{FSS}
  \\
  \hline
\end{tabular}  
}

\vspace{5mm}
\noindent For each entry of this table one may compute the \emph{total space} object of the
corresponding prequantum $k$-bundle. This is now in general itself
a higher moduli stack. In full codimension $k = 0$ one finds \cite{ExtPrequant} that 
this is the moduli 2-stack of $\mathrm{String}(G)$-2-connections described in 
\cite{SSSIII, FiorenzaSatiSchreiberII}. This we discuss in section \ref{3dCSSpin} below.

\medskip

It is clear now that this is just the first example
of a general class of theories which we may call
\emph{higher extended prequantum Chern-Simons theories} or just
\emph{$\infty$-Chern-Simons theories}, for short \cite{FS}.
These are defined by a choice of
\begin{enumerate}
  \item a smooth higher group $G$;
  \item a smooth universal characteristic map $\mathbf{c} : \mathbf{B}G \to \mathbf{B}^n U(1)$;
  \item a differential refinement $\mathbf{c}_{\mathrm{conn}} : \mathbf{B}G_{\mathrm{conn}} \to
    \mathbf{B}^n U(1)_{\mathrm{conn}}$.
\end{enumerate}
An example of a 7-dimensional such theory on $\mathrm{String}$-2-form gauge fields
is discussed in \cite{FiorenzaSatiSchreiberI}, given by a differential refinement
of the second fractional Pontrjagin class to a morphism of smooth moduli 7-stacks
$$
  \tfrac{1}{6}(\mathbf{p}_2)_{\mathrm{conn}}
  :
  \mathbf{B}\mathrm{String}_{\mathrm{conn}}
  \to
  \mathbf{B}^7 U(1)_{\mathrm{conn}}
  \,.
$$
We expect that these $\infty$-Chern-Simons theories are part of a general procedure of 
\emph{extended geometric quantization} (\emph{multi-tiered} geometric quantization)
which proceeds in two steps, as indicated in the following table.

\vspace{5mm}
\hspace{-1cm}
{\small
\begin{tabular}{|c|c|c|}
  \hline
  {\bf classical system}
  &
  {\bf geometric prequantization}
  &
  {\bf quantization}
  \\
  \hline
  \begin{tabular}{c}
    char. class $c$ of deg. $(n+1)$
	\\
	with de Rham image $\omega$:
	\\
	invariant polynomial/
	\\
	$n$-plectic form 
  \end{tabular}
  &
  \begin{tabular}{c}
    prequantum circle $n$-bundle
	\\
	on moduli $\infty$-stack of fields
	\\
	$\mathbf{c}_{\mathrm{conn}} : \mathbf{B}G_{\mathrm{conn}} \to \mathbf{B}^n U(1)_{\mathrm{conn}}$
  \end{tabular}
  &
  \begin{tabular}{c}
    extended quantum field theory 
    \\
    $Z_{\mathbf{c}}: \Sigma_k \mapsto
	 \left\{
	   \begin{array}{l}
	    \mbox{polarized sections of }
	   \\
	    \mbox{prequantum $(n-k)$-bundle}
		\\
	   \exp(2 \pi i \int_{\Sigma_k}[\Sigma_k, \mathbf{c}_{\mathrm{conn}}])
	   \end{array}
	 \right\}
	$
  \end{tabular}
  \\
  \hline
\end{tabular}
}

\vspace{5mm}
\noindent Here we are concerned with the first step, the 
discussion of $n$-dimensional Chern-Simons gauge theories
(higher gauge theories)
in their incarnation as prequantum circle $n$-bundles on their
universal moduli $\infty$-stack of fields. 
A dedicated discussion of 
higher geometric prequantization, including the discussion of 
higher Heisenberg groups, higher quantomorphism groups, 
higher symplectomorphisms and higher
Hamiltonian vector fields,
and their action on higher prequantum spaces of states by higher
Heisenberg operators, is in 
\cite{ExtPrequant}, see also \cite{Goettingen}. As shown there, plenty of interesting
physical information turns out to be captured by extended
prequantum $n$-bundles. For instance, if one regards the
B-field in type II superstring backgrounds as a prequantum 2-bundle, then 
its extended prequantization knows all about
twisted Chan-Paton bundles, the Freed-Witten anomaly 
cancellation condition for type II superstrings on D-branes and the associated anomaly line bundle
on the string configuration space.

\vspace{3mm}
Generally, all higher Chern-Simons theories that arise from extended action functionals this way 
enjoy a collection of  very good formal properties. Effectively, they may be understood as constituting examples of a fairly extensive generalization of the \emph{refined} Chern-Weil homomorphism with coefficients in \emph{secondary characteristic cocycles}. Moreover, we have shown previously  that the class of theories arising this way is large and contains not only several  familiar theories, some of which are not traditionally recognized to be of this good form, but also contains various new QFTs that turn out to be of interest  within known contexts, e.g. \cite{FiorenzaSatiSchreiberI,FiorenzaSatiSchreiberII}.  
Here we further enlarge the pool of such examples. 

\medskip

Notably, here we are concerned with examples arising from \emph{cup product}
characteristic classes, hence of $\infty$-Chern-Simons theories which are decomposable
or non-primitive secondary characteristic cocyles, obtained by cup-ing more
elementary characteristic cocycles. 
The most familiar example of these is again ordinary 3-dimensional Chern-Simons theory,
but now for the non-simply connected gauge group $U(1)$. In this case 
a gauge field configuration in $\mathbf{H}(\Sigma_3, \mathbf{B}U(1)_{\mathrm{conn}})$
is not necessarily given by a globally defined 1-form $A \in \Omega^1(\Sigma_3)$.
Instead it may have a non-vanishing ``instanton number'', the first Chern class
of the underlying circle bundle.
Only if that happens to vanish the value of the action functional is  given again by the 
simple expression $\exp(2 \pi i k \int_{\Sigma_3} A \wedge d A)$ as before.
Yet, in view of the above discussion, we are naturally led to conjecture there should be a circle 3-bundle (bundle 2-gerbe) with connection over $\Sigma_3$, depending naturally on the $U(1)$-gauge field, whose connection 3-form reduces to $A \wedge d A$ in  the special case of vanishing first Chern class. This would imply that the expression $\exp(2 \pi i k \int_{\Sigma_3} A \wedge d A)$ is actually the special instance of an
action functional which is naturally defined in the general situation: the \emph{volume holonomy} of a 3-bundle with connection (see section \ref{VolumeHolonomy} below). And it is indeed so: this circle 3-bundle with connection is the one induced by  the \emph{differential cup square}
of the gauge field with itself. As a fully extended action functional this is 
a natural morphism of higher moduli stacks of the form
$$
  (-)^{\cup_{\mathrm{conn}}^2} 
    : 
  \mathbf{B}U(1)_{\mathrm{conn}}
  \to
  \mathbf{B}^3 U(1)_{\mathrm{conn}}
  \,.
$$
We explain this below in section \ref{CircleBundleModuli}.
This morphism of higher stacks is characterized by the
fact that under forgetting the differential refinement and then taking geometric realization as
before, it is exhibited as a differential refinement of the ordinary 
cup square on Eilenberg-MacLane spaces 
$$
  (-)^{\cup^2} : K(\mathbb{Z},2) \to K(\mathbb{Z}, 4)
$$
and hence on ordinary integral cohomology.
By the above general procedure, we obtain a well-defined action functional
for $3d$ $U(1)$-Chern-Simons theory by the expression
$$
  \exp(2 \pi i \int_{\Sigma_3} [\Sigma_3, (-)^{\cup_{\mathrm{conn}}^2}])
  :
  [\Sigma_3, \mathbf{B}U(1)_{\mathrm{conn}}]
  \to
  U(1)
$$
and this is indeed the action functional of the familiar $3d$ $U(1)$-Chern-Simons theory,
also on non-trivial instanton sectors, see section \ref{3dU1CS} below.

\medskip
In terms of this general construction, there is nothing particular to the
low degrees here, and we have generally a differential cup square
/ extended action functional for a $(4k+3)$-dimensional Chern-Simons theory
$$
 (-)^{\cup_{\mathrm{conn}}^2} : 
 \mathbf{B}^{2k+1}U(1)_{\mathrm{conn}}
 \to
 \mathbf{B}^{4k+3}U(1)_{\mathrm{conn}}
$$
for all $k \in \mathbb{N}$, which induces an ordinary action functional
$$
  \exp(2 \pi i \int_{\Sigma_3} [\Sigma_{4k+3}, (-)^{\cup_{\mathrm{conn}}^2}])
  :
  [\Sigma_{4k+3}, \mathbf{B}^{4k+3}U(1)_{\mathrm{conn}}]
  \to
  U(1)
$$
on the moduli $(2k+1)$-stack of $U(1)$-$(2k+1)$-form gauge fields, 
given by the fiber integration on differential cocycles over the 
differential cup product of the fields.
This is discussed in section \ref{4k+3} below.

\vspace{3mm}
Forgetting the smooth structure on $[\Sigma_{4k+3}, \mathbf{B}^{2k+1}U(1)_{\mathrm{conn}}]$ 
and passing to gauge equivalence classes of fields yields to the cohomology group
$H^{2k+2}_{\mathrm{conn}}(\Sigma_{4k+3})$. This is what is known as 
\emph{ordinary differential cohomology} and is equivalent to the group of
\emph{Cheeger-Simons differential characters}. An excellent review on the subject with further pointers to the literature is 
in \cite{HopkinsSinger}.
That gauge equivalence classes of 
higher degree $U(1)$-gauge fields are to be regarded as
differential characters and that the $(4k+3)$-dimensional 
$U(1)$-Chern-Simons action functional on these is given by the 
fiber integration of the cup product is discussed in 
detail in \cite{FP}, also mentioned notably in \cite{Witten96,Witten98}
and expanded on in \cite{Freed}. Effectively this observation led to the general
development of differential cohomology in \cite{HopkinsSinger}. 
Or rather, the main 
theorem there concerns a shifted version of the functional of
$(4k+3)$-dimensional $U(1)$-Chern-Simons theory which allows to further divide it
by 2. We have discussed the refinement of this to smooth moduli stacks of fields in 
\cite{FiorenzaSatiSchreiberII}. 
These developments were largely motivated by the 
relation of $(4k+3)$-dimensional $U(1)$-Chern-Simons theories 
 to theories of self-dual
forms in dimension $(4k+2)$ via holography (see \cite{BM} for a survey and references):
a choice of conformal structure on a $\Sigma_{4k+2}$ naturally induces
a polarization of the prequantum 1-bundle of the $(4k+3)$-dimensional theory,
and for every such a choice the resulting space of quantum states is naturally
identified with the corresponding space of conformal blocks (correlators) for the
$(4k+2)$-dimensional theory. 

\medskip
Therefore we see that looking at the differential cup square on smooth higher moduli 
stacks as an extended action functional yields the following table of familiar notions
under extended geometric prequantization.

\vspace{5mm}
\hspace{-2cm}
{\small
\begin{tabular}{|c|c|c|}
  \hline
  {\bf dim.} & & {\bf prequantum $(4k+3-d)$-bundle} 
  \\
  \hline
  $d = 0$ & \begin{tabular}{c}differential cup square\end{tabular} & 
  $(-)^{\cup_{\mathrm{conn}}^2} : \mathbf{B}^{2k+1} U(1)_{\mathrm{conn}} 
  \to \mathbf{B}^{4k+3} U(1)_{\mathrm{conn}}$ 
  \\
  \hline
   $\vdots$ & $\vdots$ & $\vdots$
   \\
   \hline
   $d = 4k+2 $ & \begin{tabular}{c} ``pre-conformal blocks'' of \\ self-dual $2k$-form field \end{tabular}
   & 
   $ 
    \xymatrix{[\Sigma_{4k+2}, \mathbf{B}^{2k+1}U(1)_{\mathrm{conn}}]
    \ar[rr]^-{[\Sigma_{4k+2},(-)^{\cup_{\mathrm{conn}}^2}]}
      &&
      [\Sigma_{4k+2}, \mathbf{B}^{2k+1} U(1)_{\mathrm{conn}}]  
 	  \ar[rr]^-{\exp(2 \pi i \int_{\Sigma_{4k+2}}(-))} && \mathbf{B}U(1)_{\mathrm{conn}}
	  }
    $   
   \\
  \hline
  $d = 4k+3$ & \begin{tabular}{c} CS \\ action functional \end{tabular} & 
  $ 
   \xymatrix{[\Sigma_{4k+3}, \mathbf{B}^{2k+1}U(1)_{\mathrm{conn}}]
   \ar[rr]^-{[\Sigma_{4k+3},(-)^{\cup_{\mathrm{conn}}^2}]}
     &&
     [\Sigma_{4k+3}, \mathbf{B}^{2k+1} U(1)_{\mathrm{conn}}] 
	 \ar[rr]^-{\exp(2 \pi i \int_{\Sigma_{4k+3}}(-))} && U(1)
	 }
  $ 
  \\
  \hline
\end{tabular}  
}

\vspace{5mm}
This fully extended prequantization of $(4k+3)$-dimensional $U(1)$-Chern-Simons theory
allows for instance to ask for and
compute the total space of the prequantum circle $(4k+3)$-bundle. 
This is now itself
a higher smooth moduli stack. For $k = 0$, hence in $3d$-Chern-Simons theory 
it turns out to be the moduli 
2-stack of \emph{differential T-duality structures}.
This we discuss in section \ref{3dU1CS} below. 

\medskip

More generally, as the name suggests, the \emph{differential cup square} is a specialization
of a general \emph{differential cup product}. As a morphism of bare homotopy types
this is the familiar cup product of Eilenberg-MacLane spaces
$$
  (-)\cup (-) : K(\mathbb{Z},p+1) \times K(\mathbb{Z}, q+1) \to K(\mathbb{Z}, p+q+2)
$$
for all $p,q \in \mathbb{N}$.
Its smooth and then its further differential refinement is a morphism of smooth higher stacks of the
form
$$
  (-)\cup_{\mathrm{conn}} (-)
  :
  \mathbf{B}^p U(1)_{\mathrm{conn}}
  \times
  \mathbf{B}^1 U(1)_{\mathrm{conn}}
  \to
  \mathbf{B}^{p+q+1}U(1)_{\mathrm{conn}}
  \,,
$$
which, as before, we describe below in section \ref{CircleBundleModuli}.

\medskip
By the above discussion this now defines a higher extended gauge theory
in dimension $p+q+1$ of \emph{two different} species of higher 
$U(1)$-gauge fields.
One example of this is the \emph{higher electric-magnetic coupling anomaly} in higher
(Euclidean) $U(1)$-Yang-Mills theory, as explained in section 2 of \cite{Freed}. 
In this example one considers on an oriented smooth manifold $X$ 
(here assumed to be closed, for simplicity) 
an \emph{electric current}
$(p+1)$-form $J_{\mathrm{el}} \in \Omega^{p+1}_{\mathrm{cl}}(X)$
and a \emph{magnetic current} $(q+1)$-form 
$J_{\mathrm{mag}} \in \Omega^{q+1}_{\mathrm{cl}}(X)$, such that
$p + q = \mathrm{dim}(X)$ is the dimension of $X$.
A \emph{prequantization}
of these current forms in our sense of higher geometric quantization 
\cite{ExtPrequant}
is a lift to differential cocycles 
$$
  \raisebox{20pt}{
  \xymatrix{
    && \mathbf{B}^p U(1)_{\mathrm{conn}}
	\ar[d]^{F_{(-)}}
    \\
    X
	\ar@{-->}[urr]^{\widehat J_{\mathrm{el}}}
	\ar[rr]^{J_{\mathrm{el}}} 
	&&
	\Omega^{p+1}_{\mathrm{cl}}(-)~,
  }
  }
  \;\;\;~~~~~~~~~
    \raisebox{20pt}{
  \xymatrix{
    && \mathbf{B}^q U(1)_{\mathrm{conn}}
	\ar[d]^{F_{(-)}}
    \\
    X
	\ar@{-->}[urr]^{\widehat J_{\mathrm{mag}}}
	\ar[rr]^{J_{\mathrm{mag}}} 
	&&
	\Omega^{q+1}_{\mathrm{cl}}(-)
  }
  }
$$
and here this amounts to electric and magnetic \emph{charge quantization}, respectively:
the electric charge is the universal integral cohomology class of 
the circle $p$-bundle underlying the electric charge cocycle: its
\emph{higher Dixmier-Doudy class}
$[\widehat J_{\mathrm{el}}] \in H^{p+1}_{\mathrm{cpt}}(X, \mathbb{Z})$
(see section \ref{VolumeHolonomy} below); and similarly
for the magnetic charge.
Accordingly, the higher mapping stack 
$[X, \mathbf{B}^p U(1)_{\mathrm{comm}}\times \mathbf{B}^q U(1)_{\mathrm{conn}}]$
is the smooth higher moduli stack of charge-quantized electric and magnetic currents on 
$X$. Recall that this assigns to a smooth test manifold $U$ the higher groupoid
whose objects are $U$-families of pairs of charge-quantized electric and magnetic 
currents.
As \cite{Freed} explains in terms of such families of fields, the 
$U(1)$-principal bundle with connection which in the present formulation is 
the one described by the morphism
$$
 \nabla_{\mathrm{an}}
 :=
  \exp(2 \pi i \int_X [X, (-)\cup_{\mathrm{conn}}(-)])
  :
  [X, \mathbf{B}^p U(1)_{\mathrm{comm}}\times \mathbf{B}^q U(1)_{\mathrm{conn}}]
  \to
  \mathbf{B}U(1)_{\mathrm{conn}}
$$
is the \emph{anomaly line bundle} of $(p-1)$-form electromagnetism on $X$
in the presence of electric and magnetic currents subject to charge quantization.
In the language of $\infty$-Chern-Simons theory as above, this is
equivalently the off-shell prequantum 1-bundle of the higher cup
product Chern-Simons theories on pairs of $U(1)$-gauge $p$-form and $q$-form fields.

\medskip
It is customary to call \emph{local anomaly} the curvature of the anomaly bundle, and \emph{global anomaly} its 
 holonomy. In our contex the holonomy of $\nabla_{\mathrm{an}}$ is
(discussed again in section \ref{VolumeHolonomy} below) the morphism
$$
 \mathrm{hol}(\nabla_{\mathrm{an}})
 =
  \exp(2 \pi i \int_{S^1} [S^1,  \nabla_{\mathrm{an}}])
  :
  [S^1, [X, \mathbf{B}^p U(1)_{\mathrm{comm}}\times \mathbf{B}^q U(1)_{\mathrm{conn}}]]
  \to 
  U(1)
$$
from the loop space of the moduli stack of fields to $U(1)$.
By the characteristic universal propery of higher mapping stacks, 
together with the ``Fubini-theorem''-property of fiber integration, this is equivalently
the morphism
$$
  \exp(2 \pi i \int_{X \times S^1} [X \times S^1, (-)\cup_{\mathrm{conn}}(-)])
  :
  [X\times S^1, \mathbf{B}^p U(1)_{\mathrm{comm}}\times \mathbf{B}^q U(1)_{\mathrm{conn}}]
  \to 
  U(1)
  \,.
$$
But from the point of view of $\infty$-Chern-Simons theory  this is the \emph{action functional}
of the higher cup product Chern-Simons field theory induced by $\cup_{\mathrm{conn}}$. 
The situation is now summarized in the following table.

\vspace{3mm}
{\small
\hspace{-2cm}
\begin{tabular}{|c|c|c|}
  \hline
  {\bf dim.} & {\bf description} & {\bf prequantum $(\mathrm{dim}(X)+1-k)$-bundle} 
  \\
  \hline
  $k = 0$ & \begin{tabular}{c}differential cup product\end{tabular} & 
  $(-)^{\cup_{\mathrm{conn}}^2} : 
  \mathbf{B}^{p} U(1)_{\mathrm{conn}} 
  \mathbf{B}^{q} U(1)_{\mathrm{conn}}
  \to \mathbf{B}^{d+2} U(1)_{\mathrm{conn}}$ 
  \\
  \hline
   $\vdots$ & $\vdots$ & $\vdots$
    \\
   \hline
   $k = \mathrm{dim}(X) $ & \begin{tabular}{c} higher E/M-charge \\ anomaly line bundle \end{tabular}
   & 
   $ 
    \xymatrix{
	  \exp(2 \pi i \int_{X} [X,(-)\cup_{\mathrm{conn}} (-)])
	  :
	  [X, \mathbf{B}^p U(1)_{\mathrm{conn}} \times \mathbf{B}^q U(1)_{\mathrm{conn}}]
	  \ar[r]
	  &
	  \mathbf{B}U(1)_{\mathrm{conn}}
	}
    $   
   \\
  \hline
  $k = \mathrm{dim}(X)+1$ & \begin{tabular}{c} global anomaly \end{tabular} & 
   $ 
    \xymatrix{
	  \exp(2 \pi i \int_{X\times S^1} [X\times S^1,(-)\cup_{\mathrm{conn}} (-)])
	  :
	  [X \times S^1, \mathbf{B}^p U(1)_{\mathrm{conn}} \times \mathbf{B}^q U(1)_{\mathrm{conn}}]
	  \to
	  &
	  \hspace{-1.2cm}
	  U(1)
	}
    $   
  \\
  \hline
\end{tabular}  
}

\vspace{3mm}
These higher electric-magnetic anomaly Chern-Simons theories
are of particular interest when the higher electric/magnetic currents
are themselves induced by other gauge fields. Namely if we have any two 
$\infty$-Chern-Simons theories given by extended action functionals
$\mathbf{c}^{(1)}_{\mathrm{conn}} : \mathbf{B}G_{1,\mathrm{conn}} \to \mathbf{B}^{p}U(1)_{\mathrm{conn}}$
and 
$\mathbf{c}^{(2)}_{\mathrm{conn}} : \mathbf{B}G_{2,\mathrm{conn}} \to \mathbf{B}^{q}U(1)_{\mathrm{conn}}$,
respectively, then composition of these with the differential cup product yields an
extended action functional of the form
$$
  \mathbf{c}_{\mathrm{conn}}^{(1)}
  \cup_{\mathrm{conn}}
  \mathbf{c}_{\mathrm{conn}}^{(2)}
  :
  \xymatrix{
    \mathbf{B}(G_1 \times G_2)_{\mathrm{conn}}
	\ar[rr]^-{(\mathbf{c}^{(1)}_{\mathrm{conn}}, \mathbf{c}^{(2)}_{\mathrm{conn}})}
	&&
	\mathbf{B}^p U(1)_{\mathrm{conn}}
	\times
	\mathbf{B}^q U(1)_{\mathrm{conn}}
	\ar[rr]^-{(-)\cup_{\mathrm{conn} }(-)}
	&&
    \mathbf{B}^{p+q+1} U(1)_{\mathrm{conn}}
  }
  \,,
$$
which describes extended topological field theories in dimension $p+q+1$ on
two species of (possibly non-abelian, possibly higher) gauge fields, 
or equivalently describes the higher electric/magnetic anomaly for
higher electric fields induced by $\mathbf{c}^{(1)}$ and higher magnetic
fields induced by $\mathbf{c}^{(2)}$.

\medskip
For instance for heterotic string backgrounds $\mathbf{c}^{(2)}_{\mathrm{conn}}$
is the differential refinement of the first fractional Pontrjagin class
$\tfrac{1}{2}p_1 \in H^4(B \mathrm{Spin}, \mathbb{Z})$
\cite{SSSIII, FSS}, i.e., one has
$$
  \mathbf{c}^{(2)}_{\mathrm{conn}}
  =
  \widehat J^{\mathrm{NS5}}_{\mathrm{mag}}
  =
  \tfrac{1}{2}(\mathbf{p}_1)_{\mathrm{conn}}
  :
  \mathbf{B}\mathrm{Spin}_{\mathrm{conn}}
  \to
  \mathbf{B}^3 U(1)_{\mathrm{conn}}
  \,,
$$
formalizing the \emph{magnetic NS5-brane charge} needed to cancel the 
fermionic anomaly of the heterotic string by the Green-Schwarz mechanism.
It is curious to observe, going back to the very first 
example of this introduction, that this $\widehat J_{\mathrm{mag}}^{\mathrm{NS5}}$
is at the same time the extended action functional for 3d $\mathrm{Spin}$-Chern-Simons theory.

\medskip

Still more generally, we may differentially cup this way more than two 
factors. Examples for such \emph{higher order cup product theories}
appear in 11-dimensional supergravity. We discuss this in section \ref{Cubic examples}.
Notably plain classical 11d supergravity contains an 11-dimensional cubic
Chern-Simons term whose extended action functional in our sense is
$$
  (-)^{\cup_{\mathrm{conn}}^3}
  :
  \mathbf{B}^3 U(1)_{\mathrm{conn}}
  \to
  \mathbf{B}^{11} U(1)_{\mathrm{conn}}
  \,.
$$
Here for $X$ the 11-dimensional spacetime, a field in $[X, \mathbf{B}^3 U(1)]$
is a first approximation to a model for the \emph{supergravity $C$-field}.
If the differential cocycle happens to be given by a globally defined 3-form
$C$, then the induced action functional
$\exp(2 \pi i \int_{X} [X, (-)^{\cup_{\mathrm{conn}}^3}])$ sends this 
to the element in $U(1)$ given by the familiar expression
$$
  \exp(2 \pi i \int_{X} [X, (-)^{\cup_{\mathrm{conn}}^3}])
    : 
	C
	\mapsto
  \exp(2 \pi i \int_{X} C \wedge dC \wedge d C)
   \,.
$$
More precisely this model receives quantum corrections from an 11-dimensional
Green-Schwarz mechanism. In \cite{FiorenzaSatiSchreiberI,FiorenzaSatiSchreiberII} 
we have discussed in detail the relevant corrections to the above extended
cubic cup-product action functional on the moduli stack of flux-quantized
$C$-field configurations.

\vspace{3mm}
This paper is meant to be of interest to both mathematicians and theoretical/mathematical physicists. 
It provides some basic constructions and variations on theories that are 
familiar to the former,
and illustrates this with reduction to explicit examples familiar to the latter.
Our aim is to show and illustrate by further classes of interesting examples
how Chern-Weil theory interpreted in 
higher geometry, i.e.,  \emph{$\infty$-Chern-Simons theory}, usefully
unifies a wealth of structures that are of interest both in themselves as well
as for the role they play in quantum field theory and string theory.
A more general
and encompassing 
discussion should appear in \cite{ExtPrequant, FS, InfinityWZW}.

\section{General theory}
\label{GeneralTheory}

In this section we describe the  general formal definition and construction of
higher extended cup-product Chern-Simons theories defined on their full higher
moduli stacks of fields. This provides the conceptual and theoretical basis for the
discussion of the examples below in section \ref{ExamplesAndApplications}.

\subsection{Fields as smooth higher stacks}
\label{FieldsAsSmoothHigherStacks}
We briefly indicate the context of \emph{smooth higher stacks}
(equivalently: \emph{smooth $\infty$-groupoids} or \emph{smooth homotopy types})
in which we place our discussions of differential cohomology and 
extended action functionals. We initiated this approach in 
\cite{SSSIII} (with an unpublished precursor set of notes \cite{nactwist},
presented at \cite{OberwolfachTalk}), 
and so the reader can find more detailed surveys with emphasis 
on different aspects in the series of papers
 \cite{SchreiberSkoda, FSS, FRS, FiorenzaSatiSchreiberI,
FiorenzaSatiSchreiberII,  FiorenzaSatiSchreiberCS,
NSSa, NSSb}. A comprehesive account is in \cite{cohesive};
an introductory lecture series with emphasis on applications to string theory is in
\cite{Lectures}. 
The basic idea has then also been propagated in 
\cite{Hopkins}\footnote{We are grateful to 
Alexander Kahle for pointing out this talk to us at \emph{String-Math 2012}.}, 
together with making explicit this is the context in which 
the seminal article \cite{HopkinsSinger} was eventually meant to be considered.
This section should serve
to fix our notation and terminology for the present purpose, 
and to give the reader unfamiliar
with the details a quick idea of the conceptual background.
\par
A field theory associates to a spacetime $\Sigma$ a whole hierarchy of objects: fields, gauge transformations between fields, gauge transformations of gauge transformations, etc. In mathematical terms, these data define an $\infty$-categories, 
where the objects are the fields, the 1-morphisms are the gauge transformations, the 2-morphisms are the gauge of gauge transformations, and so on. The $\infty$-category of fields and gauge transformations is of a quite special kind: all $k$-morphisms are invertible for $k\geq 1$. One refers to this by saying that  the $\infty$-category of fields is an $\infty$-groupoid. From a combinatorial point of view, these can be seen as particular simplicial sets, known as \emph{Kan complexes}. Letting the spacetime $\Sigma$ vary, the above can be expressed by saying that the
field content of a field theory is an assignment
\[
\mathbf{Fields}:\mathrm{Smooth\ Manifolds}\to \infty\text{-groupoids}.
\]
Fields (and their gauge transformations) can be restricted to sub-regions of spacetime and, more generally, they can 
be pulled back along smooth maps between different spacetimes: the assignment 
$\mathbf{Fields}$ is a \emph{simplicial presheaf} on the site of smooth manifolds taking values in Kan complexes.
\par
The fact that fields and gauge transformations can be restricted to sub-regions of spacetime (in particular to open regions) allows one to speak of \emph{local data} for fields. The familiar fact that a field theory can be completely described in terms of local data then becomes a distinguished feature of the the assignment 
$\mathbf{Fields}$; it is not just a simplicial presheaf, but it is a \emph{simplicial sheaf} (or, equivalently, an \emph{$\infty$-sheaf} or \emph{$\infty$-stack}). Moreover, since every construction which can be expressed in terms of local charts, is a construction only involving \emph{Cartesian spaces}, i.e., 
those smooth manifolds diffeomorphic to $\mathbb{R}^n$ for some $n$, we see that fields are equivalently $\infty$-stacks on the site of Cartesian spaces.
\par
Actually, the correct formalization of the intuitive notion of an $\infty$-stack on the site of smooth manifolds as something which is equivalent to its local data requires a bit of work, which we now briefly recall. Details can be found in \cite{cohesive}. 
We write $[ \mathrm{CartSp}^{\mathrm{op}}, \mathrm{sSet}]$ for the
category whose objects are simplicial presheaves over Cartesian spaces, and whose morphisms are
natural transformations between them. Since Kan complexes are precisely the 
fibrant objects in the standard model category structure on simplicial sets, we write 
$[ \mathrm{CartSp}^{\mathrm{op}}, \mathrm{sSet}_{\mathrm{fib}}]$ for the subcategory of 
presheaves taking values in Kan complexes.
We say a morphism $f : X \to Y$ of Kan-complex valued presheaves is a 
\emph{local homotopy equivalence} if it is \emph{stalkwise} a homotopy
equivalence of Kan complexes, hence if for every manifold $U$ and every
point $x \in U$ there is a neighbourhood $x \in  U_x \subset U$ such that 
$f(U_x) : X(U_x) \to Y(U_x)$ is a homotopy equivalence of Kan complexes. 
We then write 
$$
  \mathbf{H} := \mathrm{Sh}_\infty( \mathrm{CartSp})\cong  \mathrm{Sh}_\infty( \mathrm{SmoothMfd}) 
$$
for the simplicial category which is universal with respect to the property that
local homotopy equivalences in $[ \mathrm{CartSp}^{\mathrm{op}}, \mathrm{sSet}_{\mathrm{fib}}]$
become actual homotopy equivalences. For $X$ and $A$ any two simplicial presheaves, we 
write $\mathbf{H}(X,A)$ for the resulting $\infty$-groupoid of morphisms between them, and call it the \emph{cocycle $\infty$-groupoid} of $A$-cocycles on $X$.
This construction is called the \emph{simplicial localization} of the category of simplicial 
presheaves at the local homotopy equivalences, and
$\mathbf{H}$ is called the \emph{$\infty$-topos}
of \emph{$\infty$-stacks} over $ \mathrm{CartSp}$. This is the context for higher geometry modeled
on $ \mathrm{CartSp}$. The existence of good open covers of smooth $n$-dimensional manifolds, i.e., of open covers $\mathcal{U}=\{U_\alpha\}$ such that all the $U_\alpha$'s as well as their multiple intersections $U_{\alpha_1\dots\alpha_k}=U_{\alpha_1}\cap\cdots \cap U_{\alpha_k}$ are diffeomorphic to $\mathbb{R}^n$ (when nonempty), ensures that every local construction on smooth manifolds can be equivalently expressed entirely in terms of Cartesian spaces, giving the equivalence $\mathrm{Sh}_\infty( \mathrm{CartSp})\cong  \mathrm{Sh}_\infty( \mathrm{SmoothMfd})$.

\subsection{From local to global: $\infty$-stackification}
\label{Stackification}
Since $\mathbf{H}$ is a localization of $[ \mathrm{CartSp}^{\mathrm{op}}, \mathrm{sSet}]$, every simplicial presheaf $\mathcal{F}ields$ in $[ \mathrm{CartSp}^{\mathrm{op}}, \mathrm{sSet}]$ has a corresponding image $\mathbf{Fields}$ in 
$\mathbf{H}$, called its \emph{$\infty$-stackification}. Unwinding the definitions one finds the following recipe for describing the $\infty$-groupoid $\mathbf{Fields}(\Sigma)$ for a given smooth manifold $\Sigma$. To begin with, since the presheaf $\mathcal{F}ields$ is a functor on $\mathrm{CartSp}^{\mathrm{op}}$ with values in simplicial sets for 
any Cartesian space $U$, we have a simplicial set $\mathcal{F}ields(U)$, with natural restriction maps $\mathcal{F}ields(U)\to \mathcal{F}ields(V)$ along inclusions $V\hookrightarrow U$ of Cartesian spaces. Let us denote the set of its $k$-simplices by 
$\mathcal{F}ields(U)_k$. Next, fix a good open cover $\mathcal{U}$ of $\Sigma$.   A $\mathbf{Fields}$-cocycle on $\Sigma$ is then given by the following local data:
\begin{itemize}
\item 0-simplices $\phi_\alpha\in \mathcal{F}ields(U_\alpha)_0$ for any $\alpha$;
\item 1-simplices $\phi_{\alpha\beta}\in \mathcal{F}ields(U_{\alpha\beta})_1$ for any $\alpha,\beta$, whose boundary 0-simplices are the restrictions of  $\phi_\alpha$ and $\phi_\beta$ to $U_{\alpha\beta}$
\item 2-simplices $\phi_{\alpha\beta\gamma}\in \mathcal{F}ields(U_{\alpha\beta\gamma})_2$ for any $\alpha,\beta,\gamma$, whose boundary 1-simplices are the restrictions of  $\phi_{\alpha\beta}$, $\phi_{\beta\gamma}$ and $\phi_{\gamma\alpha}$ to $U_{\alpha\beta\gamma}$;
\item and so on.
\end{itemize}
The above description only gives the objects (or 0-morphisms) of the $\infty$-groupoid $\mathbf{Fields}(\Sigma)$ of $\mathbf{Fields}$-cocycles on $\Sigma$. The description of morphisms is straightforward: if $\phi$ and $\phi'$ are cocycles, then a 1-morphism $\psi$ between them is the data of
 \begin{itemize}
\item 1-simplices $\psi_\alpha\in \mathcal{F}ields(U_\alpha)_1$ for any $\alpha$, whose boundary 
0-simplices are $\phi_\alpha$ and ${\phi'}_\alpha$, respectively;
\item ``squares'' $\psi_{\alpha\beta}$ (to be thought of as pairs of 2-simplices in $\mathcal{F}ields(U_{\alpha\beta})_2$ with a common edge) whose boundary 1-simplices are as in the following diagram
 \[
 \xymatrix{
 \phi_\alpha\bigr\vert_{U_{\alpha\beta}}\ar[r]^{\psi_\alpha}\ar[d]_{\phi_{\alpha\beta}}&  {\phi'}_\alpha\bigr\vert_{U_{\alpha\beta}}\ar[d]^{{\phi'}_{\alpha\beta}}\\
  \phi_\beta\bigr\vert_{U_{\alpha\beta}}\ar[r]^{\psi_\beta}&  {\phi'}_\beta\bigr\vert_{U_{\alpha\beta}}
  \ar@{=>}_{\psi_{\alpha\beta}}(8,-9);(13,-5)
}
\]
\item and so on.
\end{itemize}
Similarly, one describes $k$-morphisms for any $k\geq 1$.
\medskip

\begin{example}\label{example.manifold}
Let $M$ be a smooth manifold. Then $M$ induces a presheaf of Kan 
complexes on $\mathrm{CartSp}$ mapping a Cartesian space $U$ to the set (i.e., to the 0-groupoid) of smooth functions from $U$ to $M$. We want to describe the associated $\infty$-stack $\mathbf{M}$. To do so, we need to say what 
the $\infty$-groupoid associated to a smooth manifold $\Sigma$ is. 
By the above recipe, once a good open cover $\mathcal{U}$ of $\Sigma$ has been fixed, the objects of this $\infty$-groupoid are given by local maps $\phi_\alpha:U_\alpha\to M$ such that $\phi_\alpha\bigr\vert_{U_{\alpha\beta}}=\phi_\beta\bigr\vert_{U_{\alpha\beta}}$ for any $\alpha,\beta$. In other words, the objects of the $\infty$-groupoid of $\mathbf{M}$-valued-cocycles on $\Sigma$ are just smooth maps form $\Sigma$ to $M$. Also, it is immediate to see from the recipe above that the only morphisms between the objects in $\mathbf{M}(\Sigma)$ are the identities. In other words, $\mathbf{M}$ is nothing but the image of $M$ via the Yoneda lemma:
\[
\mathbf{M}:\Sigma\mapsto C^\infty(\Sigma,M).
\]
By this reasoning, we will identify a smooth manifold $M$ and the $\infty$-stack $\mathbf{M}$ it defines and will denote them both by the symbol $M$. Note that by the Yoneda lemma we have a natural equivalence
\[
\mathbf{Fields}(M)\simeq \mathbf{H}(M,\mathbf{Fields})
\]
for any smooth $\infty$-stack $\mathbf{Fields}$ and any smooth manifold $M$ (where on the right $M$ is identified with the $\infty$-stack it defines).
 \end{example}
 
 \begin{example}
For any nonnegative integer $n$, let $\Omega^n$ be the simplicial presheaf which associates with any Cartesian space $U$ the set $\Omega^n(U;\mathbb{R})$ of degree $n$ smooth differential forms with real coefficients on $U$ (seen as a trivial groupoid). As in example \ref{example.manifold}, the presheaf $\Omega^n$ is already an $\infty$-stack and the $\infty$-groupoid of $\Omega^n$-cocycles on a manifold $\Sigma$ is equivalent to the set $\Omega^n(\Sigma)$ of  of degree $n$ smooth differential forms with real coefficients on $\Sigma$. The same considerations apply to the simplicial presheaf $\Omega^n_{\mathrm{cl}}$ of closed $n$-forms.
\end{example}
 
\begin{example}
Let $G$ be a Lie group $G$. It defines a presheaf of Kan 
complexes which sends a test manifold $U$ to the 1-groupoid $*//C^\infty(U,G)$ with a single object $*$
and the (discrete) group of smooth functions $C^\infty(U,G)$ as morphisms from that
object to itself. One denotes by $\mathbf{B}G \in \mathbf{H}$ the corresponding
smooth stack. For $\Sigma$ a smooth manifold with  a good open cover $\mathcal{U}$, 
an object of $\mathbf{B}G(\Sigma)\simeq  \mathbf{H}(\Sigma, \mathbf{B}G)$ is given by the following set of data:
\begin{itemize}
\item  smooth functions $g_{\alpha\beta}:U_{\alpha\beta}\to G$ for any $\alpha,\beta$, such that  $g_{\alpha\beta}g_{\beta\gamma}g_{\gamma\alpha}=1$ on $U_{\alpha\beta\gamma}$.
\end{itemize}
These are manifestly the data defining a $G$-principal bundle on $\Sigma$. Also, it is immediate to check that morphisms in $\mathbf{B}G(\Sigma)$ are precisely isomorphisms of $G$-principal bundles. This means that
$\mathbf{B}G(\Sigma)$ is equivalent to the groupoid $G \mathrm{Bund}(\Sigma)$ of $G$-principal bundles and 
smooth gauge transformations between these. We will call $\mathbf{B}G$  the \emph{moduli stack of $G$-principal bundles}.
%
More generally, an $\infty$-stack $G$ with group structure (up to higher homotopy: 
a groupal $A_\infty$-structure) determines and is determined by a moduli
$\infty$-stack $\mathbf{B}G$ which \emph{modulates} $G$-principal $\infty$-bundles
in this way \cite{NSSa}. 
The set of connected components of the cocycle $\infty$-groupoid $\mathbf{H}(X, \mathbf{B}G)$ is
$$
  H^1(X, G) := \pi_0 \mathbf{H}(X, \mathbf{B}G)
  \,,
$$
the \emph{degree-1 nonabelian cohomology} of $X$ with coefficients in $G$.
If $\mathbf{B}G$ itself again has a group structure, we may form $\mathbf{B}^2 G$,
and so on. Generally, if an object $A$ is \emph{$n$-times deloopable} this way we write
$$
  H^n(X, A) := \pi_0 \mathbf{H}(X, \mathbf{B}^n A)
$$
for the degree-$n$ cohomology of $X$ with coefficients in $A$. For $G$ a Lie group or $A$ an abelian Lie group, these constructions precisely
reproduce the usual degree-1 nonabelian cohomology with coefficients in $G$
and the usual degree-$n$ cohomology with coefficients in $A$, respectively. A
detailed account of the general situation can be found in \cite{cohesive,NSSb}.

\end{example}

\begin{example}\label{example.bgconn}Let $G$ be a Lie group and let $\mathfrak{g}$ be its Lie algebra. These data define a presheaf of Kan 
complexes which sends a test manifold $U$ to the 1-groupoid   $\Omega^1(U;\mathfrak{g})//C^\infty(U,G)$ which has 
the set $\Omega^1(U;\mathfrak{g})$ of $\mathfrak{g}$-valued 1-forms on $U$ as objects and with the group $C^\infty(U,G)$ acting on $\Omega^1(U;\mathfrak{g})$ via the gauge transformations 
\[
g: A\mapsto g^{-1}Ag+g^{-1}dg
\]
 as set of morphisms. It is then immediate to see, following closely the description in the example above, that the stackification $\mathbf{B}G_{\mathrm{conn}}$ of this simplicial presheaf is the stack mapping a smooth manifold $\Sigma$ to the groupoid of principal $G$-bundles with connections on $\Sigma$. 
\end{example}

\begin{example}As a particular case of example \ref{example.bgconn}, consider $G=U(1)$. Then, under the identification $\Omega^1(U;\mathbb{R})\cong \Omega^1(U;\mathfrak{u}_1)$ given by the multiplication by $2\pi i$, the simplicial presheaf defining the stack $\mathbf{B}U(1)_{\mathrm{conn}}$ of principal $U(1)$-bundles with connection is identified with the simplicial presheaf $\Omega^1//C^\infty(-;U(1))$, with the action
\[
g: A\mapsto A+\tfrac{1}{2\pi i}d\,\mathrm{log} g.
\]
We will come back to this example in section \ref{Sec BD}.
\end{example}

\medskip 
Since localization is a functorial procedure, a morphism of simplicial presheaves $\varphi:\mathcal{A}\to \mathcal{B}$ naturally induces a morphism $\varphi:\mathbf{A}\to \mathbf{B}$ between their $\infty$-stackifications.

\begin{example}
Let $G$ be a Lie group with Lie algebra $\mathfrak{g}$. Then the evident morphism of simplicial presheaves
\[
\Omega^1(-;\mathfrak{g})//C^\infty(-;G) \to *//C^\infty(-;G)
\]
induces the forgetful morphism of stacks $\mathbf{B}G_{\mathrm{conn}}\to \mathbf{B}G$ which forgets the connection.
\end{example}

\begin{example}
We now turn to the curvature corresponding to the above connection. 
The de Rham differential $d:\Omega^1\to \Omega^2$ induces a morphism of simplicial presheaves
\[
d:\Omega^1//C^\infty(-;U(1))\to \Omega^2_{\mathrm{cl}}.
\]
The $\infty$-stackification of this morphism is the morphism
\[
\mathrm{curv}:\mathbf{B}U(1)_{\mathrm{conn}}\to \Omega^2_{\mathrm{cl}}
\]
mapping a $U(1)$-bundle with conenction to its curvature 2-form. We will come back to this example in section \ref{Sec BD}.
\end{example}

\subsection{The stack of field configurations}
\label{FieldConfigurations}
One of the basic properties of $\infty$-toposes such as our
$\mathbf{H} = \mathrm{Sh}_\infty(\mathrm{CartSp})$ is that they are
\emph{cartesian closed}. This means that:
%
for every two objects $X, A \in \mathbf{H}$
 there is another object denoted 
$[X,A] \in \mathbf{H}$ that behaves like the ``space of smooth maps from $X$ to $A$''; 
namely the object $[X,A]$ is such that, for every $Y \in \mathbf{H}$, there is a natural
equivalence of $\infty$-groupoids 
$$
  \mathbf{H}(X \times Y, A) \simeq \mathbf{H}(Y, [X,A])
  \,.
$$
In other word, cocycles with coefficients in $[X,A]$ on $Y$ are 
naturally equivalent to 
$A$-cocycles on the product $X \times Y$.
As remarked, by the Yoneda lemma, each spacetime $\Sigma$ can be viewed 
as a smooth stack, so for any stack of fields $\mathbf{Fields}$ we have a natural notion of the smooth 
stack $[\Sigma,\mathbf{Fields}]$ of \emph{field configurations} on $\Sigma$. Note that, by definition 
and Yoneda lemma, one has a natural equivalence
\[
[*,\mathbf{Fields}]\cong \mathbf{Fields}
\]
and that the stack of field configurations on $\Sigma$ assigns to the 1-point manifold $*$ the $\infty$-groupoid $\mathbf{H}(\Sigma,\mathbf{Fields})$. 

\begin{example}\label{example.i4i8}
  In the discussion of anomaly polynomials in heterotic string theory over a 
  10-dimensional spacetime $\Sigma_{10}$ one encounters degree-12 differential forms
  $I_4 \wedge I_8$, where $I_i$ is a degree $i$ polynomial in characteristic forms. 
  Clearly these cannot live on $\Sigma_{10}$, as every 12-form on $\Sigma_{10}$, 
  given by an element in the hom-$\infty$-groupoid 
  $$
    \xymatrix{
      \mathbf{H}(\Sigma_{10}, \Omega^{12}(-))
	  \ar[rr]_-{\mathrm{Yoneda}}^-\simeq
	  &&
	  \Omega^{12}(\Sigma_{10})
	}
  $$
  is trivial.
  Instead, these differential forms are elements in the internal hom $[\Sigma_{10}, \Omega^{12}(-)]$, which means
  that for every choice of smooth parameter space $U$ there is a smooth 12-form
  on $\Sigma_{10} \times U$, such that this system of forms transforms naturally in $U$. See \cite{SSSII, SSSIII}
  for more detail on the context.
\end{example}
\par
To get a more explicit description of the stack of field configurations on a fixed spacetime manifold $\Sigma$, recall the discussion in section \ref{Stackification}: for $M$ a smooth manifold, the $\infty$-groupoid of $[\Sigma,\mathbf{Fields}]$-valued cocycles on $M$ is 
completely described in terms of local data associated with 
any good open cover $\mathcal{U}_{\Sigma \times M}$ of $\Sigma \times M$. If $\mathcal{U}_\Sigma$ and $\mathcal{U}_M$ are good open covers of $\Sigma$ and $M$, respectively, then $\mathcal{U}_\Sigma\times \mathcal{U}_M$ is a good open cover of $\Sigma \times M$. This means that the stack of field configurations can be seen as the $\infty$-stackification of the simplicial presheaf
\[
 U \mapsto\{\text{local $\mathbf{Fields}$-data on $\mathcal{U}_\Sigma \times U$}\}
	\,,
\]
where $\mathcal{U}_\Sigma$ is a good open cover of $\Sigma$.
\begin{example}\label{example.holonomy}
As an illustrative example let us spell out the simplicial presheaf $\mathrm{pre}\text{-}[\Sigma,\mathbf{B}U(1)_{\mathrm{conn}}]$ inducing the stack $[\Sigma,\mathbf{B}U(1)_{\mathrm{conn}}]$ of principal $U(1)$-bundles with connections on a fixed spacetime $\Sigma$, and construct a natural morphism from this presheaf to the presheaf of $U(1)$-valued functions, the holonomy of the connection. Let $\mathcal{U}_\Sigma=\{U_\alpha\}$ be a fixed good open cover of $\Sigma$. Then to a Cartesian space $U$ is associated the following groupoid:
\begin{itemize}
\item  objects in $\mathrm{pre}\text{-}[\Sigma,\mathbf{B}U(1)_{\mathrm{conn}}](U)$ are collections $\{A_\alpha, g_{\alpha\beta}\}$ of 1-forms, $A_\alpha\in\Omega^1(U_\alpha\times U;\mathbb{R})$ and of smooth functions $g_{\alpha\beta}:U_{\alpha\beta}\times U\to U(1)$ such that
\[
\tfrac{1}{2\pi i}d\,\mathrm{log}g_{\alpha\beta}=A_\beta\bigr\vert_{U_{\alpha\beta}\times U}-A_\alpha\bigr\vert_{U_{\alpha\beta}\times U}
\]
on $U_{\alpha\beta}\times U$ and $g_{\alpha\beta}g_{\beta\gamma}g_{\gamma\alpha}=1$ on $U_{\alpha\beta\gamma}\times U$
\item morphisms between $\{A_\alpha, g_{\alpha\beta}\}$ and $\{{A'}_\alpha, {g'}_{\alpha\beta}\}$ in $\mathrm{pre}\text{-}[\Sigma,\mathbf{B}U(1)_{\mathrm{conn}}](U)$ are collections $\{\psi_\alpha\}$ of smooth functions $\psi_{\alpha}:U_{\alpha}\times U\to U(1)$ such that
\[
\tfrac{1}{2\pi i}d\,\mathrm{log}\psi_{\alpha}={A'}_\alpha-A_\alpha
\]
and
\[
{g'}_{\alpha\beta}\psi_\alpha\bigr\vert_{U_{\alpha\beta}\times U}=\psi_\beta\bigr\vert_{U_{\alpha\beta}\times U}g_{\alpha\beta}
\]
on $U_{\alpha\beta\times U}$.
\end{itemize}
\par
Picking a partition of unity $\{\rho_\alpha\}$ subordinate to the good open $\mathcal{U}_\Sigma$ we can consider the set of 1-forms $\rho_\alpha A_\alpha$, which are compactly supported in the $U_\alpha$ direction. In particular, if $\Sigma=S^1$ then each $U_\alpha$ is 1-dimensional and we can integrate $\rho_\alpha A_\alpha$ along the $U_\alpha$-direction to get a smooth function on $U$. Moreover, since $S^1$ is compact, we can choose the good open cover $\mathcal{U}_\Sigma$ to be finite and add up all these contributions to get a smooth function
\[
\sum_\alpha \int_{U_\alpha} \rho_\alpha A_\alpha
\]
on $U$. One can think of this construction as assigning an element in the set $\Omega^0(U;\mathbb{R})$ to an object $\{A_\alpha, g_{\alpha\beta}\}$ in the groupoid $\mathrm{pre}\text{-}[S^1,\mathbf{B}U(1)_{\mathrm{conn}}](U)$:
\[
\{A_\alpha, g_{\alpha\beta}\}\mapsto \sum_\alpha \int_{U_\alpha} \rho_\alpha A_\alpha.
\]
However, one sees that there is something unnatural in this assignment since half of the data in $\{A_\alpha, g_{\alpha\beta}\}$ are forgotten, and indeed the above map is not a morphism of groupoids
\[
\int_{S^1,{\rho_\alpha}}: \mathrm{pre}\text{-}[S^1,\mathbf{B}U(1)_{\mathrm{conn}}](U)\to \Omega^0(U).
\]
That is, if $\{\psi_\alpha\}: \{A_\alpha, g_{\alpha\beta}\}\to \{{A'}_\alpha, {g'}_{\alpha\beta}\}$ is a morphism, then in general we have
\[
\int_{S^1,{\rho_\alpha}}\{{A}_\alpha, {g}_{\alpha\beta}\}\neq \int_{S^1,{\rho_\alpha}}\{{A'}_\alpha, {g'}_{\alpha\beta}\}.
\]
A natural way to cure this problem, taking into account also the $\{g_{\alpha\beta}\}$ part of the datum, is to notice that the cocycle equation $g_{\alpha\beta}g_{\beta\gamma}g_{\gamma\alpha}=1$ tells us that the $\{g_{\alpha\beta}\}$ define a principal $U(1)$-bundle on $S^1\times U$ and that, since $S^1\times U$ is homotopy equivalent to $S^1$ and the classifying space $BU(1)$ is simply connected, every principal $U(1)$-bundle on $S^1\times U$ can be trivialized. That is, there exist smooth functions $k_\alpha:U_\alpha\times U\to U(1)$ such that
\[
g_{\alpha\beta}=k_\beta\bigr\vert_{U_{\alpha\beta\times U}} \, k_\alpha^{-1}\bigr\vert_{U_{\alpha\beta\times U}}.
\]
Let $\omega_\alpha=A_\alpha-\frac{1}{2\pi}d\, \mathrm{log}k_\alpha$. On $U_{\alpha\beta}\times U$ we have
\[
\omega_\beta\bigr\vert_{U_{\alpha\beta}\times U}-\omega_\alpha\bigr\vert_{U_{\alpha\beta}\times U}=A_\beta\bigr\vert_{U_{\alpha\beta}\times U}-A_\alpha\bigr\vert_{U_{\alpha\beta}\times U}-\tfrac{1}{2\pi i}d\,\mathrm{log}g_{\alpha\beta}=0,
\]
i.e., the $\omega_\alpha$ are the local data for a globally defined 1-form $\omega$ on $S^1\times U$.
This way we are led to correct our naive assignment by adding the counterterm
\[
-\tfrac{1}{2\pi i}\sum_\alpha \int_{U_\alpha} \rho_\alpha d\,\mathrm{log}k_\alpha,
\]
and consequently to
consider the assignment
\[
\{A_\alpha, g_{\alpha\beta}\}\mapsto\sum_\alpha \int_{U_\alpha} \rho_\alpha (A_\alpha- \tfrac{1}{2\pi i}d\,\mathrm{log}k_\alpha)=\sum_\alpha \int_{S^1} \rho_\alpha \omega=\int_{S^1}\omega.
\]
Since $\omega$ is independent of the partition of unity $\{\rho_\alpha\}$, this map only depends on the choice of the trivialization $\{k_\alpha\}$ of the cocycle $\{g_{\alpha\beta}\}$. However, if $\{\tilde{k}_\alpha\}$ is another trivialization and we set $f_\alpha=k_\alpha\tilde{k}^{-1}_\alpha$, we see that 
\[
f_\alpha\bigr\vert_{U_{\alpha\beta}\times U}\,{f_\beta}^{-1}\bigr\vert_{U_{\alpha\beta}\times U}=1,
\]
and so the $f_\alpha$ define a global smooth function $f:S^1\times U\to U(1)$. In particular, the local 1-forms $d\,\mathrm{log}f_\alpha$ glue together into the global 1-form $d\,\mathrm{log}f$. Therefore, if
$\tilde\omega$ is the 1-form defined by the trivialization $\{\tilde{k}_\alpha\}$ , one has
\[
\int_{S^1}\tilde\omega-\omega=\tfrac{1}{2\pi i}\sum_\alpha\int_{S^1}\rho_\alpha(d\,\mathrm{log}k_\alpha-d\,\mathrm{log}\tilde{k}_\alpha)=\int_{S^1}d\, \mathrm{log}f.
\]
This last integral is an integer, namely it is the winding number of $f:S^1\times U\to U(1)$, i.e., equivalently, the element in $H^1(S^1\times U;\mathbb{Z})\cong H^1(S^1;\mathbb{Z})\cong\mathbb{Z}$ corresponding to the homotopy class of $f:S^1\times U\to U(1)\cong K(\mathbb{Z},1)$. Therefore, we see that we have a well defined map
\[
\exp\bigl(2\pi i\int_{S^1}\,\bigr):\{A_\alpha, g_{\alpha\beta}\}\to C^\infty(U;U(1)),
\]
independent of all the choices. Now, if $\{\psi_\alpha\}: \{A_\alpha, g_{\alpha\beta}\}\to \{{A'}_\alpha, {g'}_{\alpha\beta}\}$ is a morphism and $\{k_\alpha\}$ is a trivialization of $\{g_{\alpha\beta}\}$, then $k'_\alpha=\psi_\alpha k_\alpha$ is a trivialization of $\{g'_{\alpha\beta}\}$ and with this choice of trivialization the 1-form $\omega'$ associated to $\{{A'}_\alpha, {g'}_{\alpha\beta}\}$ is locally given by
\[
\omega'_\alpha=
A'_\alpha- \tfrac{1}{2\pi i}d\,\mathrm{log}k'_\alpha=A_\alpha- \tfrac{1}{2\pi i}d\,\mathrm{log}k_\alpha\omega=\omega_\alpha\,.
\]
Hence we see that both $\{A_\alpha, g_{\alpha\beta}\}$ and $\{{A'}_\alpha, {g'}_{\alpha\beta}\}$ have the same image in $C^\infty(U;U(1))$. Since the construction is natural in $U$, we have defined a morphism of simplicial presheaves
\[
\exp\bigl(2\pi i\int_{S^1}\,\bigr):\mathrm{pre}\text{-}[S^1,\mathbf{B}U(1)_{\mathrm{conn}}]\to C^\infty(-;U(1))\,.
\]
By $\infty$-stackification this induces a morphism of $\infty$-stacks
\[
\mathrm{hol}_{S^1}:[S^1,\mathbf{B}U(1)_{\mathrm{conn}}]\to U(1)\,,
\] 
the \emph{holonomy} along $S^1$. Here on the right hand side $U(1)$ is identified with the smooth $\infty$-stack it represents; see example \ref{example.manifold}. More generally, since every compact closed oriented 1-dimensional manifold $\Sigma_1$ is a disjoint union of finitely many copies of $S^1$, one can use the abelian group structure on the target $U(1)$ to assemble the contributions from the various connected components into a global holonomy morphism
\[
\mathrm{hol}_{\Sigma_1}:[\Sigma_1,\mathbf{B}U(1)_{\mathrm{conn}}]\to U(1)\,,
\] 
We will come back to this example in the context of fiber integration in Deligne cohomology in section \ref{FibInt}.
\end{example}

\begin{remark}
The object $[X,A]$ is  known in category theory as the \emph{internal hom} object,
but in applications to physics and to stacks it is often better known as the
``families version'' of $A$-cocycles on $X$: for each smooth parameter space 
$U \in \mathrm{SmthMfd}$, the elements of $[X,A](U)$ are 
``$U$-parameterized families of $A$-cocycles on $X$'', namely $A$-cocycles on $X \times U$.
%
By the universal property characterizing it, the construction of the internal hom object is functorial, i.e., if $f:A\to B$ is a morphism of $\infty$-stacks, then one has a natural morphism
$[X,f]:[X,A]\to [X,B]$
 of $\infty$-stacks, for any $X\in\mathbf{H}$. In particular, taking $X$ to be a spacetime $\Sigma$, one finds a natural morphism 
 \[
 [X,f]:[X,A]\to [X,B]
 \]
from the $\infty$-stack of $A$-field configurations on $\Sigma$ to the $\infty$-stack of $B$-field configurations on $\Sigma$. \end{remark}

%
\begin{example}
  Below in section \ref{CupOfTwo} we discuss how the anomaly forms from example \ref{example.i4i8} appear from
  morphisms of higher moduli stacks 
    $$
    \nabla : \mathbf{SuGra} \to \mathbf{B}^{11}U(1)_{\mathrm{conn}}\;, 
  $$
  for $\mathbf{SuGra}$ the higher stack of supergravity field
  configurations, by sending
  the families of moduli of field configurations on spacetime $\Sigma_{10}$ to their
  anomaly form:
  $$
	\xymatrix{
      [\Sigma_{10}, \mathbf{SuGra}]
	  \ar[rr]^-{[\Sigma_{10}, \nabla]}
	  &&
	  [\Sigma_{10}, \mathbf{B}^{11} U(1)_{\mathrm{conn}}]
	  \ar[rr]^-{[\Sigma_{10}, \mathrm{curv}]}
	  &&
	 [\Sigma_{10}, \Omega^{12}(-)]
	 }
	 \,.
  $$
%
\end{example}

\begin{remark}\label{rem.yoneda}
By the Yoneda lemma one has a natural equivalence
\[
[X\times Y,Z]\cong [X,[Y,Z]],
\]
for any three smooth $\infty$-stacks $X$, $Y$ and $Z$. Also, via the equivalence
\[
\mathbf{H}([X,Y],[X,Y])\cong \mathbf{H}([X,Y]\times X, Y)\;,
\]
the identity of $[X,Y]$ corresponds to the evaluation map $[X,Y]\times X\to Y$.
\end{remark}

\subsection{The Dold-Kan correspondence}
\label{Dold-Kan}

A useful tool for producing $\infty$-stacks $A$ with \emph{abelian} 
$\infty$-group structure 
is the \emph{Dold-Kan correspondence}, which we briefly recall here
(see for instance section III.2 of \cite{GoerssJardine} for a review).
First, at an algebraic level, we have the classical Dold-Kan correspondence
$$
  \xymatrix{
    \mathrm{Ch}_{\bullet\geq 0} \ar[r]^-\Gamma_-\simeq &  \mathrm{sAb},
  }
$$
which
establishes
an equivalence of categories between chain complexes concentrated in non-negative degrees and simplicial abelian groups. Given the chain complex
\[
A_\bullet=\cdots \xrightarrow{\partial}A_3\xrightarrow{\partial}A_2\xrightarrow{\partial}A_1\xrightarrow{\partial}A_0\,,
\]
the simplicial abelian group $\Gamma(A_\bullet)$ is defined as follows:
\begin{itemize}
\item the abelian group of 0-simplices of $\Gamma(A_\bullet)$ is the abelian group $A_0$;
\item the abelian group of $n$-simplices of  $\Gamma(A_\bullet)$ is the abelian group whose elements are standard $n$-simplices decorated by an element $x$ in $A_n$  such that $\partial x$ equals the (oriented) sum of the decorations on the boundary $(n-1)$-simplices.
\end{itemize}
For instance, a 2-simplex in $\Gamma(A_\bullet)$ is 
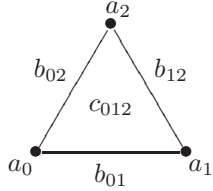
\[
\begin{xy}
,(0,0)*{\bullet};(-10,-17)*{\bullet}**\dir{-}
,(0,0)*{\bullet};(10,-17)*{\bullet}**\dir{-}
,(-10,-17)*{\bullet};(10,-17)*{\bullet}**\dir{-}
,(-12,-19)*{a_0}
,(12,-19)*{a_1}
,(,2)*{a_2}
,(-8,-6)*{b_{02}}
,(8,-6)*{b_{12}}
,(0,-20)*{b_{01}}
,(0,-11)*{c_{012}}
\end{xy}
\]
where
\begin{itemize}
\item $a_i\in A_0$;
\item $b_{ij}\in A_1$ and $\partial b_{ij}=a_j-a_i$;
\item $c_{012}\in A_2$ and $\partial c_{012}=b_{12}-b_{02}+b_{01}$.
\end{itemize}
Then there is the forgetful functor
$$
  F : \mathrm{sAb} \to \mathrm{sSet}_{\mathrm{fib}} \hookrightarrow \mathrm{sSet}
$$
which forgets the group structure on a simplicial abelian group and just remembers
the underlying simplicial set, which in turn is guaranteed to be a Kan complex. 
Denoting by $ \mathrm{DK}$ the composition of $\Gamma$ and $F$ we obtain the \emph{Dold-Kan correspondence}:
\[
 \mathrm{DK}: 
  \xymatrix{
   \mathrm{Ch}_{\bullet\geq 0}
   \ar[r]^-{\simeq}
   &
   \mathrm{sAb}
   \ar[r]^-{F}
   &
   \mathrm{sSet}
   }
   \,.
\]
By construction, this is such that 
the elements in degree $k$ of a chain complex label the extension of
$k$-cells in the corresponding simplicial set; and the chain homology group in degree
$k$ is naturally identified with the simplicial homotopy group in the same degree
$$
  \pi_k( \mathrm{DK}(A_\bullet) ) \simeq H_k(A_\bullet)  
  \,.
$$
All this prolongs directly to presheaves of chain complexes and presheaves of abelian groups on Cartesian spaces,
and we will use the same symbols and write
$$
  \mathrm{DK}
   :
  \xymatrix{
    [ \mathrm{CartSp}^{\mathrm{op}}, \mathrm{Ch}_{\bullet\geq 0}]
	\ar[r]^-\Gamma
	&
	[ \mathrm{CartSp}^{\mathrm{op}}, \mathrm{sAb}]
	\ar[r]^-{F}
	&
	[ \mathrm{CartSp}^{\mathrm{op}}, \mathrm{sSet}]
  }
  \,.
$$

\begin{example}
If $A$ is an abelian Lie group, then for any nonnegative integer $n$ one can consider the presheaf of chain complexes
\[
C^\infty(-;A)[n]:= \bigl[\cdots\to0\to C^\infty(-,A)\to0\to\cdots\to0\bigr],
\]
with $C^\infty(-,A)$ placed in degree $n$. By applying the Dold-Kan map to this presheaf one gets a simplicial presheaf whose stackification is the $n$-stack $\mathbf{B}^nA$ of principal $A$-$n$-bundles. For $n=0$ this is the sheaf of smooth functions with values in $A$; for $n=1$ this is the usual stack of principal $A$-bundles; for $n=2$ this is the 2-stack of principal $A$-bundle gerbes.  
\end{example}
\begin{remark}\label{rem.hyper}
If $X$ is a smooth manifold, then the set of connected components $\pi_0 \mathbf{H}(X, \mathbf{B}^n A)$ is naturally identified with the traditional $n$-th \emph{sheaf cohomology} group
$H^n(X,\mathcal{A})$
 of $X$ with coefficients in $\mathcal{A}:=C^\infty(-;A)$. Note that, since $\mathbf{H}(X, \mathbf{B}^n A)$ is actually defined for any smooth stack $X$, not just a smooth manifold, one gets this way a natural notion of sheaf cohomology groups for stacks.\par
More generally, if $\mathcal{A}_\bullet \in [ \mathrm{CartSp}^{\mathrm{op}}, \mathrm{Ch}_{\bullet \geq 0}]$ is an arbitrary presheaf
of chain complexes, and $\mathbf{A}$ is the smooth stack obtained by the 
stackification of $\mathrm{DK}(\mathcal{A}_\bullet)$, then 
$$
  \mathbb{H}^0(X, \mathcal{A}_\bullet) := \pi_0 \mathbf{H}(X, \mathbf{A})
$$
is what is traditionally is called the $0$-th \emph{sheaf hypercohomology} group of 
$X$ with coefficients in $\mathcal{A}_\bullet$.

\end{remark}

\begin{example}\label{example.deligne1} The length 1 Deligne complex is the presheaf of chain complexes
\[
\cdots\to0\to C^\infty(-;U(1))\xrightarrow{\frac{1}{2\pi}d\mathrm{log}} \Omega^1,
\]
with $\Omega^1$ in degree zero. It 
presents, via Dold-Kan and stackification, the stack $\mathbf{B}U(1)_{\mathrm{conn}}$ of principal $U(1)$-bundles with connection. In particular, the group $H^1_{\mathrm{diff}}(X,U(1)):=\pi_0\mathbf{H}(X,\mathbf{B}U(1)_{\mathrm{conn}})$ is naturally identified with the degree zero hypercohomology of the length 1 Deligne complex. Moreover, by the quasi-isomorphism of presheaves of chain complexes
\[
\xymatrix{
\cdots\ar[r]\ar[d]&0\ar[r]\ar[d]&\mathbb{Z}\ \ar@{^{(}->}[r]\ar[d]&C^\infty(-;\mathbb{R})\ar[rr]^-{d}\ar[d]_{\exp(2\pi i-)}&&\Omega^1\ar[d]\\
\cdots\ar[r]&0\ar[r]&0\ar[r]&C^\infty(-;U(1))\ar[rr]^-{{\frac{1}{2\pi}d\mathrm{log}} }&&\Omega^1
}
\]
one recovers this way the classical result that isomorphism classes of principal $U(1)$-bundles with connection on $X$ are classified by $H^1_{\mathrm{diff}}(X;\mathbb{Z})$, the first ordinary differential cohomology group of $X$.
\end{example}

\begin{example} The description of the stack $[\Sigma,\mathbf{B}U(1)_{\mathrm{conn}}]$ given in example \ref{example.holonomy} can be rephrased in terms of the Dold-Kan correspondence by saying that $[\Sigma,\mathbf{B}U(1)_{\mathrm{conn}}]$ is the stack presented via Dold-Kan and stackification by the presheaf of chain complexes given by the total complex of the Deligne length 1 complex on $\mathcal{U}_\Sigma\times U$. One sees from this perspective that the holonomy morphism $\mathrm{hol}_{\Sigma_1}$ is the image, still via the Dold-Kan correspondence, of the fiber integration map in Deligne cohomology. We will come back to this point of view in Section \ref{FibInt}.
\end{example}

A crucial property of the Dold-Kan map is its compatibility with the tensor product. More precisely, we have the following. 
\begin{proposition}\label{cup-products}
Let $\mathcal{A}_\bullet, \mathcal{B}_\bullet$ and $\mathcal{C}_\bullet$ be presheaves of chain complexes concentrated in non-negative degrees, and let $\cup: \mathcal{A}_\bullet\otimes \mathcal{B}_\bullet\to \mathcal{C}_\bullet$ be a morphism  of presheaves of chain complexes. Then the Dold-Kan map induces a natural morphism of simplicial preseheaves $\cup_{\mathrm{DK}}: \mathrm{DK}(\mathcal{A}_\bullet)\times \mathrm{DK}(\mathcal{B}_\bullet)\to \mathrm{DK}(\mathcal{C}_\bullet)$.
\end{proposition}
\begin{proof}
Both the categories $\mathrm{Ch}_{\bullet\geq 0}$ and $\mathrm{sAb}$ are monoidal categories under the respective standard tensor products. Namely, 
on $\mathrm{Ch}_{\bullet\geq 0}$ this is given by direct sums of tensor products of abelian groups with fixed total degree 
and on $\mathrm{sAb}$ by the degreewise tensor product of abelian groups. 
Moreover, the functor $\Gamma$ is lax monoidal with respect to these structures, 
i.e., for any $V, W \in \mathrm{Ch}_{\bullet\geq 0}$ we have natural 
weak equivalences
$$
  \gamma_{V,W} : \Gamma(V) \otimes \Gamma(W) \to \Gamma(V \otimes W)
  \,.
$$
 The forgetful functor $F$ is the right adjoint to the functor forming degreewise the free abelian group on a set, therefore it preserves products and hence  there are natural isomorphisms
$$
   F(V \times W) \xrightarrow{\simeq} F(V) \times F(W)
   \,,
$$
for all $V,W \in \mathrm{sAb}$.
Finally, by the definition of tensor product, there are universal natural 
quotient maps $V, W \in \mathrm{sAb}$
$$
  p_{V,W} : V \times W \to V \otimes W
  \,.
$$
The morphism $\cup_{\mathrm{DK}}$ is then defined as the composition indicated in the following diagram:
$$
  \xymatrix{
    \mathrm{DK}(\mathcal{A}_\bullet) \times \mathrm{DK}(\mathcal{B}_\bullet)
    	 \ar@{=}[d]
	 \ar[rrrr]^{\mathbf{\cup}_{\mathrm{DK}}}
	 &&&&
	 \mathrm{DK}(\mathcal{C}_\bullet)
	 \ar@{=}[dd]
	 \\
	 F(\Gamma(\mathcal{A}_\bullet)) \times F(\Gamma(\mathcal{B}_\bullet))
	 \ar[d]^\simeq 
	 \\
     F(\Gamma(\mathcal{A}_\bullet) \times \Gamma(\mathcal{B}_\bullet))
	 \ar[r]^{F(p)}
	 &
     F(\Gamma(\mathcal{A}_\bullet) \otimes \Gamma(\mathcal{B}_\bullet))
	 \ar[r]^-{\hspace{-2mm}F(\gamma)}
	 &
     F(\Gamma(\mathcal{A}_\bullet\otimes \mathcal{B}_\bullet))
	 \ar[rr]^-{F(\Gamma(-\cup-))}
	 &&
	 F(\Gamma(\mathcal{C}_\bullet))\;.
  }
$$
\end{proof}

%
%


\subsection{Deligne cohomology and $n$-stacks of higher $U(1)$-bundles with connection}
\label{Sec BD}

Ordinary degree-2 integral cohomology $H^2(X, \mathbb{Z})$ 
on a smooth manifold $X$ classifies smooth circle bundles on $X$. 
Ordinary \emph{differential cohomology} $H^2_{\mathrm{diff}}(X, \mathbb{Z})$
classifies smooth circle bundles with connection. Namely, the
groupoid $\mathbf{H}(X,\mathbf{B}U(1)_{\mathrm{conn}})$ whose objects are circle
bundles with connection on $X$, and whose morphisms are smooth gauge transformations
on $X$ is such that 
\[
H^2_{\mathrm{diff}}(X, \mathbb{Z}) = H^1_{\mathrm{conn}}(X, U(1)):=\pi_0 \mathbf{H}(X,\mathbf{B}U(1)_{\mathrm{conn}}),
\]
 see example \ref{example.deligne1}.
Generalized to arbitrary degree, one obtains $n$-groupoids 
$\mathbf{H}(X,\mathbf{B}^nU(1)_{\mathrm{conn}})$ whose objects are $U(1)$-$n$-bundles with connection, whose morphisms
are smooth gauge transformations, whose 2-morphisms are gauge-of-gauge transformations,
and so on.

\medskip
A famous model for these $n$-groupoids using chain complexes is due to Deligne
and Beilinson, and it is known as the \emph{Deligne complex}. We briefly
review it, together with its cup product operation. 
In the context of differential geometry
the use of Deligne cohomology was amplified notably by Brylinski \cite{Brylinski}. 

\medskip

\begin{definition}
Write
$
  \mathbb{Z}[n+1]^\infty_D
$
for the presheaf of chain complexes on Cartesian spaces given by
$$
 \mathbb{Z}[n+1]^\infty_D
 :=
 \left[
 \xymatrix{ 
     \mathbb{Z} ~
   \ar@{^{(}->}[r]
       &
     \Omega^0
       \ar[r]^-{d} 
            &
     \Omega^1
       \ar[r]^-{d} 
            &
    ~ \cdots~
     \ar[r]^-{d}
     &
     \Omega^{n}
     }
	\right]\;,
$$
with the constant presheaf of integers $\mathbb{Z}$ in degree $(n+1)$, the inclusion morphism into the sheaf $\Omega^0=C^\infty(-;\mathbb{R})$ of smooth real functions (in degree $n$) and with all the other
differentials being the de Rham differentials on the sheaves of differential forms.
This is the \emph{Deligne complex} in degree $(n+1)$. The sheaf hypercohomology
with coefficients in $\mathbb{Z}[n+1]^\infty_D$ is accordingly the 
\emph{Deligne cohomology}. The cohomology group
\[
H^{n+1}_{\mathrm{diff}}(X;\mathbb{Z}):=\mathbb{H}^0(X;\mathbb{Z}[n+1]^\infty_D)
\]
is called the $n$-th ordinary differential cohomology group of $X$.\label{DeligneComplex}
%
\par
We write
$
  \mathbf{B}^{n} U(1)_{\mathrm{conn}}
$
for the smooth $n$-stack presented by the Deligne complex $\mathbb{Z}[n+1]^\infty_D$ via the Dold-Kan map and stackification, and call it the $n$-stack of $U(1)$-$n$-bundles (or circle $n$-bundles) with connection.\end{definition}
\begin{remark}
It follows by the discussion in Remark \ref{rem.hyper} that $H^n_{\mathrm{diff}}(X;\mathbb{Z})$ classifies $U(1)$-$n$-bundles with connection. Also note that, due to the evident quasi-isomorphism of presheaves of complexes
\[
\xymatrix{
\cdots\ar[r]\ar[d]&0\ar[r]\ar[d]&\mathbb{Z}\ \ar@{^{(}->}[r]\ar[d]&C^\infty(-;\mathbb{R})\ar[rr]^-{d}\ar[d]_{\exp(2\pi i-)}&&\Omega^1\ar[d]\ar[r]^{d}&\cdots\ar[r]&\Omega^{n-1}\ar[d]\ar[r]^{d}&\Omega^n\ar[d]\\
\cdots\ar[r]&0\ar[r]&0\ar[r]&C^\infty(-;U(1))\ar[rr]^-{{\frac{1}{2\pi}d\mathrm{log}} }&&\Omega^1\ar[r]^{d}&\cdots\ar[r]&\Omega^{n-1}\ar[r]^{d}&\Omega^n\;,
}
\]
the $n$-stack $\mathbf{B}^nU(1)_{\mathrm{conn}}$ is equivalently presented via Dold-Kan and stackification by the presheaf of chain complexes
$$
 \left[
 \xymatrix{ 
     \cdots\ar[r]&0\ar[r]&C^\infty(-;U(1))\ar[rr]^-{{\frac{1}{2\pi}d\mathrm{log}} }&&\Omega^1\ar[r]^{d}&\cdots\ar[r]&\Omega^{n-1}\ar[r]^{d}&\Omega^n
	}
	\right]\;.
$$
It is precisely this presentation that is the one
 making the interpretation of $\mathbf{B}^nU(1)_{\mathrm{conn}}$ as the $n$-stack of $U(1)$-$n$-bundles with connection manifest.
\end{remark}
\begin{remark}
The obvious morphism
\[
\xymatrix{
\cdots\ar[r]&0\ar[r]\ar[d]&C^\infty(-;U(1))\ar[rr]^-{{\frac{1}{2\pi}d\mathrm{log}} }\ar[d]&&\Omega^1\ar[r]^{d}\ar[d]&\cdots\ar[r]&\Omega^{n-1}\ar[r]^{d}\ar[d]&\Omega^n\ar[d]\\
\cdots\ar[r]&0\ar[r]&C^\infty(-;U(1))\ar[rr]&&0\ar[r]&\cdots\ar[r]&0\ar[r]&0
}
\]
induces the ``forget the connection'' morphism $\mathbf{B}^nU(1)_{\mathrm{conn}}\to \mathbf{B}^nU(1)$ to the stack of principal $U(1)$-$n$-bundles. As in the $n=1$ case, by the exponential exact sequence $0\to \mathbb{Z}\to \mathbb{R}\to U(1)\to 1$ one sees that $H^{n+1}(X;\mathbb{Z})\cong   \pi_0\mathbf{H}(X,\mathbf{B}U(1))$, 
so that principal $U(1)$-$n$-bundles on $X$ are classified by degree $n+1$ integral cohomology.
\end{remark}

\begin{remark}
The morphism of presheaves of chain complexes
\[
\xymatrix{
\cdots\ar[r]\ar[d]&0\ar[r]\ar[d]&\mathbb{Z}\ \ar@{^{(}->}[r]\ar[d]&C^\infty(-;\mathbb{R})\ar[r]^-{d}\ar[d]&\Omega^1\ar[d]\ar[r]^{d}&\cdots\ar[r]&\Omega^{n-1}\ar[d]\ar[r]^{d}&\Omega^n\ar[d]\\
\cdots\ar[r]&0\ar[r]&0\ar[r]&0\ar[r]&0\ar[r]&\cdots\ar[r]&0\ar[r]^{d}&\Omega^{n+1}_{\mathrm{cl}}
}
\]
induces the \emph{curvature} morphism
\[
\mathrm{curv}: \mathbf{B}^nU(1)_{\mathrm{conn}}\to \Omega^{n+1}_{\mathrm{cl}}
\] 
from the stack of $U(1)$-$n$-bundles with connection to the sheaf of degree $n+1$ closed form. For $n=1$ this is nothing but the familiar morphism mapping a $U(1)$-bundle with connection to its curvature 2-form.
\end{remark}

\subsection{Total spaces of higher $U(1)$-bundles}
\label{CircleBundleModuli}

%
%

The (higher) stacks perspective offers a very neat point of view on the construction of the total space of a principal $U(1)$-$n$-bundle classified by a given cohomology class $c$ in $H^{n+1}(X;\mathbb{Z})$. Namely, due to the natural isomorphism $H^{n+1}(X;\mathbb{Z})\cong   \pi_0\mathbf{H}(X,\mathbf{B}U(1))$ one can pick a morphism $\mathbf{c} : X \to \mathbf{B}^n U(1)$ representing the homotopy class $c$ and realize the total space $P$ as the homotopy fiber of $\mathbf{c}$ i.e., as the
%
object universally fitting into a square 
$$
  \raisebox{20pt}{
  \xymatrix{
    P \ar[r]_{\ }="s" \ar[d]^>{\ }="t" & {*} \ar[d]
 	\\
	X \ar[r]_-{\mathbf{c}} & \mathbf{B}^n U(1)~~.
	\ar@{=>}^\simeq "s"; "t"
  }
  }
$$
\label{TotalSpaceOfPrequantumBundle}

This is the (higher) stack version of the usual universal property of classifying spaces for topological groups; a detailed discussion can be found in  in \cite{NSSa}, where the above statement appears as a special case of the first main theorem. Note that the space $X$ need not be a smooth manifold; it can be an arbitrary smooth stack. 
\begin{example}We unravel the above definition of the total space $P$ in the $n=1$ case, to show how it precisely reproduces the construction of the total space of a principal $U(1)$-bundle on a smooth manifold $X$ from its $U(1)$-cocycle data. The first step consists in replacing $X$ with the \v{C}ech nerve $\check{C}(\mathcal{U})$ of a good open cover $\mathcal{U}$ of $X$ (this operation is a fibrant replacement from the point of view of smooth stacks), and in looking at the cocycle data of a $U(1)$-bundle on $X$ as a simplicial morphism $\{g_{\alpha\beta}\}:\check{C}(\mathcal{U})\to \mathbf{B}U(1)$. Now we just take the homotopy fiber of this simplicial map, obtaining the following data:
\begin{itemize}
\item the collection of products $U_\alpha\times U(1)$;
\item the collection of the gluing data $\tilde{g}_{\alpha\beta}:U_{\alpha\beta}\times U(1)\to U_{\alpha\beta}\times U(1)$ given by $\tilde{g}_{\alpha\beta}: (x,a)\mapsto (x,g_{\alpha\beta}\cdot a)$. 
\end{itemize}
These are precisely the data describing the total space of the $U(1)$ bundle associated with the cocycle $\{g_{\alpha\beta}\}$.
\end{example}

\begin{definition}
We denote by $\Omega^{1 \leq \bullet \leq n}$ the $(n-1)$-stack obtained via Dold-Kan and stackification from the presheaf of chain complexes on Cartesian spaces
 \[
 \dots\to0\to\Omega^1\xrightarrow{\,\,d\,\,}\Omega^2\xrightarrow{\,\,d\,\,}\cdots\xrightarrow{\,\,d\,\,}\Omega^n.
 \]
 \end{definition}
\begin{remark}
The short exact sequence of complexes of sheaves
\[
\xymatrix{
\cdots\ar[r]&0\ar[r]\ar[d]&0\ar[rr]\ar[d]&&\Omega^1\ar[r]^{d}\ar[d]&\cdots\ar[r]&\Omega^{n-1}\ar[r]^{d}\ar[d]&\Omega^n\ar[d]\\
\cdots\ar[r]&0\ar[r]\ar[d]&C^\infty(-;U(1))\ar[rr]^-{{\frac{1}{2\pi}d\mathrm{log}} }\ar[d]&&\Omega^1\ar[r]^{d}\ar[d]&\cdots\ar[r]&\Omega^{n-1}\ar[r]^{d}\ar[d]&\Omega^n\ar[d]\\
\cdots\ar[r]&0\ar[r]&C^\infty(-;U(1))\ar[rr]&&0\ar[r]&\cdots\ar[r]&0\ar[r]§ &0
}
\]
exhibits $\Omega^{1 \leq \bullet \leq n}$ as the homotopy fiber of the forgetful morphism $\mathbf{B}^nU(1)_{\mathrm{conn}}\to \mathbf{B}^nU(1)$, i.e., we have a natural homotopy pullback diagram\footnote{
To see that this homotopy commutative square is indeed a homotopy pullback, notice that: {\it i)} the Dold-Kan map $\mathrm{DK}$ is right Quillen for the 
global projective model structure on simplicial presheaves, {\it ii)} homotopy pullbacks in the local model structure may be computed in the global model structure ($\infty$-stackification is left exact), and {\it iii)} the pre-image under $\mathrm{DK}$ of the forgetful morphism is manifestly a fibration. Hence,
we may compute the homotopy pullback of the forgetful morphism as the ordinary pullback
of presheaves of chain complexes, under $\mathrm{DK}$, and, since these are computed as objectwise and degreewise pullbacks
of abelian groups, this manifestly
yields the fiber $\Omega^{1 \leq \bullet \leq n}$ as indicated. See \cite{cohesive} for a comprehensive account on this argument.}
\[
 \raisebox{20pt}{
  \xymatrix{
    \Omega^{1 \leq \bullet \leq n}(-)  \ar[r] \ar[d] & {*} \ar[d]
 	\\
	 \mathbf{B}^n U(1)_{\mathrm{conn}}  \ar[r]^-{\text{forget}} & \mathbf{B}^n U(1)~.
  }
  }
\]
 This gives a natural identification of $\Omega^{1 \leq \bullet \leq n}$
with the $(n-1)$-stack of $U(1)$-$n$-bundles with connections whose underlying principal $U(1)$-$n$-bundle is trivial.
\end{remark}

\medskip

\noindent As a matter of terminology, we will say that a morphism $\mathbf{c}:X\to \mathbf{B}^nU(1)$ \emph{modulates} a $U(1)$-$n$-bundle (and similarly for $U(1)$-$n$-bundles with connection).
\medskip

\begin{corollary}
For a morphism 
$\nabla : X \to \mathbf{B}^n U(1)_{\mathrm{conn}}$ modulating a $U(1)$-$n$-bundle
with connection, the total space of the underlying principal $U(1)$-$n$-bundle is
equivalently realized as the homotopy pullback
$$
  \raisebox{20pt}{
  \xymatrix{
    P \ar[r]_{\ }="s" \ar[d]^>{\ }="t" & \Omega^{1 \leq \bullet \leq n} \ar[d]
 	\\
	X \ar[r]_-\nabla & \mathbf{B}^n U(1)_{\mathrm{conn}}~~,
	\ar@{=>}^\simeq "s"; "t"
  }
  }
$$
where the right vertical morphism is the obvious inclusion.
\label{TotalSpaceOfnConnection}
\end{corollary}
\proof
As we have remarked, the total space of the underlying principal $U(1)$-$n$-bundle is
homotopy pullback
$$
  \raisebox{20pt}{
  \xymatrix{
    P \ar[r]_{\ }="s" \ar[d]^>{\ }="t" & {*} \ar[d]
 	\\
	X \ar[r]_-{\chi(\nabla)} & \mathbf{B}^n U(1)~~,
	\ar@{=>}^\simeq "s"; "t"
  }
  }
$$
where the bottom morphism is the composite
$$
  \chi(\nabla) : 
   \xymatrix{
     X \ar[r]^-\nabla 
	 & \mathbf{B}^n U(1)_{\mathrm{conn}}
      \ar[r]^-{\text{forget}} & \mathbf{B}^n U(1)
  }\,.
$$
By the pasting law for homotopy pullbacks the pullback along such a composite
map may be computed by iteratively pulling back along the two components,
hence by forming the following pasting composite of homotopy pullback squares:
$$
  \raisebox{20pt}{
  \xymatrix{
    P \ar[r] \ar[d] & \Omega^{1 \leq \bullet \leq n}(-)  \ar[r] \ar[d] & {*} \ar[d]
 	\\
	X \ar[r]^-{\nabla} & \mathbf{B}^n U(1)_{\mathrm{conn}}  \ar[r]^-{\text{forget}} & \mathbf{B}^n U(1)~.
  }
  }
$$
\endofproof

\begin{remark}
  In the context of higher geometry, the total space object $P$ may have a deeper meaning
  than in ordinary geometry. For instance if $X = \mathbf{B}G_{\mathrm{conn}}$ is the stack of $G$-connections  for some
  Lie  group $G$, or more generally the higher stack of $G$-$\infty$-connections for some smooth higher group $G$, and $\nabla:\mathbf{B}G_{\mathrm{conn}}\to  \mathbf{B}^n U(1)_{\mathrm{conn}}$ is some universal differential characteristic class, then the total space $P$ is itself also a higher
  stack: namely it is the higher stack of  $G$-gauge fields equipped with 
  a trivialization of their underlying topological class $\chi(\nabla)$.
  Noteworthy examples of this phenomenon
  are discussed below in sections \ref{3dCSSpin}, \ref{7dCSString} and
  \ref{DiffT}.
\end{remark}

\subsection{The Beilinson-Deligne cup product as a morphism of stacks}
\label{Sec BD2}

The Beilinson-Deligne cup product is an explicit presentation of the cup product in ordinary differential cohomology for the case that the latter is modeled by the \v{C}ech-Deligne cohomology. 

\begin{definition}
The \emph{Beilinson-Deligne cup product} is the morphism of sheaves of chain complexes 
$$
\cup_{\mathrm{BD}}: \mathbb{Z}[p+1]^\infty_D\otimes  \mathbb{Z}[q+1]^\infty_D ~\longrightarrow~ \mathbb{Z}[p+q+2]^\infty_D,
$$
given on homogeneous elements $\alpha$, $\beta$ as follows:

$$
  \alpha \cup_{\mathrm{BD}} \beta :=
  \left.
    \begin{cases}
    \alpha \beta  &\text{ if } \mathrm{deg}(\alpha) = p+1,
    \\
    \alpha \wedge d\beta &\text{ if }  \mathrm{deg}(\alpha) \leq p 
    \text{ and } \mathrm{deg}(\beta) = 0, \\
    0  &\text{ otherwise}.
    \end{cases}
  \right.
$$
\end{definition}
A survey of this can be found in \cite{Brylinski} (around Prop. 1.5.8 there).
\begin{remark}
 When restricted to the diagonal in the case that $p= q$, 
 this means that the cup product sends a $p$-form $\alpha$ to the $(2p+1)$-form 
 $\alpha \wedge d \alpha$. This is of course the local Lagrangian for 
 cup product Chern-Simons theory of $p$-forms. We discuss this case in detail in
 section \ref{4k+3}.
\end{remark}
The Beilinson-Deligne cup product is associative and commutative up to homotopy, so it induces an associative and commutative cup product on ordinary differential cohomology,
\[
\cup_{\mathrm{BD}}:H^{p+1}_{\mathrm{diff}}(X;\mathbb{Z})\otimes H^{q+1}_{\mathrm{diff}}(X;\mathbb{Z})\to H^{p+q+2}_{\mathrm{diff}}(X;\mathbb{Z}),
\]
covering the usual cup product in integral cohomology. 
\par
We may now refine the differential cup product to $\infty$-stacks.
With the presentation of $\infty$-stacks by the Dold-Kan correspondence
discussed above this is now immediate, but important.
\begin{definition}
For $p,q \in \mathbb{N}$, we denote by
$$
  \mathbf{\cup}_{\mathrm{conn}} : 
    \mathbf{B}^{p}U(1)_{\mathrm{conn}} \times
   \mathbf{B}^{q}U(1)_{\mathrm{conn}}
   \to
   \mathbf{B}^{p+q+1} U(1)_{\mathrm{conn}}
$$
the morphism of (higher) stacks associated to the Beilinson-Deligne cup product $\cup_{\mathrm{BD}}: \mathbb{Z}[p+1]^\infty_D\otimes  \mathbb{Z}[q+1]^\infty_D ~\longrightarrow~ \mathbb{Z}[p+q+2]^\infty_D$ by Proposition \ref{cup-products} and stackification.
 \label{CupOnStacks}
 \label{ExtendedDifferentialCup}
\end{definition}

\begin{example}
For $p=q=1$ we obtain a natural morphism
\[
\mathbf{\cup}_{\mathrm{conn}}:\mathbf{B}U(1)_{\mathrm{conn}}\times \mathbf{B}U(1)_{\mathrm{conn}}\to \mathbf{B}^3U(1)_{\mathrm{conn}}.
\]
In particular for a smooth manifold $X$ this means  that cocylce data $(A_\alpha,g_{\alpha\beta})$ and $(A'_\alpha,g'_{\alpha\beta})$ for two $U(1)$-principal bundles with connection on $X$ can be used to define a $U(1)$-$3$-bundle with connection on $X$. To do so, one first lifts the transition functions of the bundles to a smooth function
\[
\tilde{g}_{\alpha\beta}, \tilde{g}'_{\alpha\beta}: U_{\alpha\beta}\to \mathbb{R}.
\]
The $\tilde{g}_{\alpha\beta}$ will not in general be a cocycle with values in $\mathbb{R}$, but their failure to be a cocycle will be well behaved: there will exist an integer $n_{\alpha\beta\gamma}$ such that
\[
\tilde{g}_{\alpha\beta}+\tilde{g}_{\beta\gamma}+\tilde{g}_{\gamma\alpha}=2\pi i n_{\alpha\beta\gamma},
\]
and analogously for $\tilde{g}'_{\alpha\beta}$. The triples $(n_{\alpha\beta\gamma},\tilde{g}_{\alpha\beta},A_\alpha)$ and $(n'_{\alpha\beta\gamma},\tilde{g}'_{\alpha\beta},A'_\alpha)$ are degree zero cocycles in the total complex of $X$ with coefficients in 
$\mathbb{Z}[2]^\infty_D$ relative to the good open cover $\mathcal{U}$, 
and so their Beilinson-Deligne cup product will produce a degree zero cocycle $(m_{\alpha\beta\gamma\delta\epsilon},l_{\alpha\beta\gamma\delta},h_{\alpha\beta\gamma},k_{\alpha\beta},B_\alpha)$ with coefficients in $\mathbb{Z}[4]^\infty_D$.
Explicitly, these data are as follows:
\begin{itemize}
\item $m_{\alpha\beta\gamma\delta\epsilon}=n_{\alpha\beta\gamma}n'_{\gamma\delta\epsilon}$ is an integer, for any quintuple intersection $U_{\alpha\beta\gamma\delta\epsilon}$;
\item $l_{\alpha\beta\gamma\delta}=n_{\alpha\beta\gamma}\tilde{g}'_{\gamma\delta}$ is a smooth $\mathbb{R}$-valued function on $U_{\alpha\beta\gamma\delta}$, for any quadruple intersection;
\item $h_{\alpha\beta\gamma}=n_{\alpha\beta\gamma}A'_\gamma$ is a smooth 1-form on $U_{\alpha\beta\gamma}$ for any triple intersection;
\item $k_{\alpha\beta}=\tilde{g}_{\alpha\beta}\wedge dA'_\beta$ is a smooth 2-form on $U_{\alpha\beta}$ for any triple intersection;
\item $B_\alpha=A_\alpha\wedge d A'_\alpha$ a smooth 3-form on $U_{\alpha}$ for any $U_\alpha$ in the cover.
\end{itemize}
If we set 
\[
\lambda_{\alpha\beta\gamma\delta}=\exp(2\pi i \, l_{\alpha\beta\gamma\delta}):U_{\alpha\beta\gamma\delta}\to U(1), 
\]
then $\{\lambda_{\alpha\beta\gamma\delta}\}$ is a 3-cocycle with values in $U(1)$ and so defines a principal $U(1)$-3-bundle on $X$. The data $(h_{\alpha\beta\gamma},k_{\alpha\beta},B_\alpha)$ are then the differential forms data for a connection on this 3-bundle. Note how the local liftings $\tilde{g}_{\alpha\beta}$ as well as the integers $n_{\alpha\beta\gamma}$ measuring the failure of the liftings to be a cocycle, appear in the formulas for the 1-form and 2-form data of the connection. 
\par
By composing the morphism $\mathbf{\cup}_{\mathrm{conn}}$ with the diagonal embedding of $\mathbf{B}U(1)_{\mathrm{conn}}$ into $\mathbf{B}U(1)_{\mathrm{conn}}\times \mathbf{B}U(1)_{\mathrm{conn}}$ we get
a natural morphism of stacks
\[
\mathbf{B}U(1)_{\mathrm{conn}}\to \mathbf{B}^3U(1)_{\mathrm{conn}}
\] 
associating a $U(1)$-3-bundle with connection to a $U(1)$-bundle with connection on $X$. In particular, we see from the explicit formulas above that, if $\{A_\alpha\}$ are the local 1-form data for a $U(1)$-connection, then the 3-forms $\{A_\alpha\wedge dA_\alpha\}$ are the local 3-form data for the local 3-form component of a 3-connection on a $U(1)$-3-bundle.
\par
Note that if we start with a connectionon on a \emph{trivial} $U(1)$-bundle on a fixed spacetime $\Sigma$ , then we can choose cocycle data $(g_{\alpha\beta},A_\alpha)$ for it with $g_{\alpha\beta}\equiv 1$ and $A_\alpha$ the local data for a global 1-form $A$ on $\Sigma$. Choosing the trivial lifts $\tilde{g}_{\alpha\beta}\equiv 0$ we then see that the associated 3-cocycle is $(1,0,0,A_\alpha\wedge dA_\alpha)$, i.e. also the associated $U(1)$-3-bundle is trivial and the associated 3-connection reduces to the globally defined 3-form $A\wedge dA$ on $\Sigma$.
\end{example}

\begin{remark}
Since
the Beilinson-Deligne cup product is associative up to homotopy, it induces a well defined morphism
$$
\mathbf{B}^{n_1}U(1)_{\mathrm{conn}}\times \mathbf{B}^{n_2}U(1)_{\mathrm{conn}}\times \cdots \times \mathbf{B}^{n_{k+1}}U(1)_{\mathrm{conn}}\to \mathbf{B}^{n_1+\cdots+n_{k+1}+k}U(1)_{\mathrm{conn}}.
$$
For instance, if $n_1=\cdots=n_{k+1}=3$, we find a morphism
$$
\left(\mathbf{B}^{3}U(1)_{\mathrm{conn}}\right)^{k+1}\to \mathbf{B}^{4k+3}U(1)_{\mathrm{conn}}.
$$
Composing this with the diagonal embedding $\mathbf{B}^{3}U(1)_{\mathrm{conn}}\to \left(\mathbf{B}^{3}U(1)_{\mathrm{conn}}\right)^{k+1}$ we get a natural morphism
\[
\mathbf{B}^{3}U(1)_{\mathrm{conn}} \to \mathbf{B}^{4k+3}U(1)_{\mathrm{conn}},
\]
which associates a $(4k+3)$-$U(1)$-bundle with connection to a $U(1)$-3-bundle with connection. From the explicit expression of the Beilinson-Deligne cup product we see that, if $\{C_\alpha\}$ are the local 3-form data for the 3-connection on the $U(1)$-3-bundle, 
then the $4k+3$-form local data for the corresponding connection on the associated $U(1)$-$(4k+3)$-bundle are
\(
C_\alpha\wedge \underbrace{dC_\alpha\wedge \cdots \wedge{dC_\alpha}}_{k\text{ times}}.
\label{Eq ktimes}
\)
We will illustrate the above constructions with various (classes of) examples arising from 
string theory and M-theory in Section \ref{ExamplesAndApplications}.
\end{remark}

\subsection{Fiber integration in Deligne cohomology as a morphism of stacks}
\label{FibInt}

Let $\Sigma_k$ be a $k$-dimensional closed (i.e., compact and without boundary) oriented manifold of dimension $k\leq n$, and let $M$ be an arbitrary smooth manifold. Gomi and Terashima show in \cite{gomi-terashima1} that the fiber integration map in ordinary cohomology
\[
\int_{\Sigma_k}:H^{n+1}(\Sigma_k\times M;\mathbb{Z})\to H^{n-k+1}(M;\mathbb{Z})
\]
has a lift to differential cohomology
\[
\int_{\Sigma_k}:H^{n+1}_{\mathrm{diff}}(\Sigma_k\times M;\mathbb{Z})\to H^{n-k+1}_{\mathrm{diff}}(M;\mathbb{Z})
\]
which is the degree zero part of the fiber integration map in Deligne (hyper-)cohomology:
\[
\int_{\Sigma_k}:\mathbb{H}^{\bullet}(\Sigma_k\times M;\mathbb{Z}[n+1]^\infty_D)\to \mathbb{H}^{\bullet}_{\mathrm{diff}}(M;\mathbb{Z}[n-k+1]^\infty_D).
\]
Moreover, this fiber integration map in hypercohomology is induced by a morphism between the total chain complexes
\[
\int_{\Sigma_k}:\mathrm{Tot}^\bullet(\mathcal{U}_{\Sigma_k}\times \mathcal{U}_M;\mathbb{Z}[n+1]^\infty_D)\to \mathrm{Tot}^\bullet(\mathcal{U}_M;\mathbb{Z}[n-k+1]^\infty_D),
\]
where $\mathcal{U}_\Sigma$ and $\mathcal{U}_M$ are good open covers of $\Sigma$ and $M$, respectively. In particular, if $U$ is a Cartesian space, Gomi and Terashima construction gives a morphism of chain complexes
\[
\int_{\Sigma_k}:\mathrm{Tot}^\bullet(\mathcal{U}_{\Sigma_k}\times U;\mathbb{Z}[n+1]^\infty_D)\to \mathrm{Tot}^\bullet(U;\mathbb{Z}[n-k+1]^\infty_D)
\]
natural in $U$, and so a morphism of presheaves of chain complexes on Cartesian spaces. Applying Dold-Kan and stackification this gives a morphism
\[
  \mathrm{hol}_\Sigma:=  \exp(2 \pi i \int_{\Sigma_k} )
	:
	[\Sigma_k, \mathbf{B}^n U(1)_{\mathrm{conn}}]
	\to
	\mathbf{B}^{n-k}U(1)_{\mathrm{conn}}
\]
  from the moduli $n$-stack of $U(1)$-$n$-bundles with connection on $\Sigma_k$
  to the $(n-k)$-stack of $U(1)$-$(n-k)$-bundles with connection. For $n=1$ this is the morphism described in example \ref{example.holonomy}. By analogy with the $n=1$ case, the morphism $\mathrm{hol}_\Sigma$ is called the \emph{$k$-dimensional holonomy} (or $k$-dimensional parallel transport) along $\Sigma$, see, e.g.,  \cite{gomi-terashima2}. Notice how the exponentiation $\exp(2\pi i -)$ has appeared as an effect of going from the Deligne complexes $\mathbb{Z}[p+1]^\infty_D$ to the quasi-isomorphic complexes of sheaves
  \[
   \left[
 \xymatrix{ 
     \cdots\ar[r]&0\ar[r]&C^\infty(-;U(1))\ar[rr]^-{{\frac{1}{2\pi}d\mathrm{log}} }&&\Omega^1\ar[r]^{d}&\cdots\ar[r]&\Omega^{p-1}\ar[r]^{d}&\Omega^p
	}
	\right]\;.
  \]

\begin{remark}
For $k=n$ the $n$-dimensional holonomy is a morphism
\[
\mathrm{hol}_{\Sigma_n}= \exp(2 \pi i \int_{\Sigma_n} (-))
	:
	[\Sigma_n, \mathbf{B}^n U(1)_{\mathrm{conn}}]
	\to
	\mathbf{U}(1),
\]
where $\mathbf{U}(1)$ is the 0-stack (i.e., the sheaf) of $U(1)$-valued smooth functions. This is the  \emph{$n$-volume holonomy} of a $U(1)$-$n$-connection over the
  ``$n$-dimensional Wilson volume'' $\Sigma_n$.
\end{remark}

\begin{remark}
If the manifold $\Sigma_k$ is a product, $\Sigma_k=\Sigma_{k_1}\times\Sigma_{k_2}$, then fiber integration along $\Sigma_k$ can be computed in two steps: first by fiber integrating along $\Sigma_{k_1}$ and then integrating the result along $\Sigma_{k_2}$. n terms of the holonomy morphisms, this Fubini formula translates into the natural homotopy commutative diagram
\[
\xymatrix{
[\Sigma_{k_1},[\Sigma_{k_2}, \mathbf{B}^n U(1)_{\mathrm{conn}}]\ar@{=}[d]^\wr\ar[rr]^{[\Sigma_{k_1},\mathrm{hol}_{\Sigma_2}]}&& [\Sigma_{k_1},\mathbf{B}^{n-k_2} U(1)_{\mathrm{conn}}]\ar[d]^{\mathrm{hol}_{\Sigma_1}}\\
[\Sigma_{k_1}\times \Sigma_{k_2}, \mathbf{B}^n U(1)_{\mathrm{conn}}]\ar[rr]^{\mathrm{hol}_{\Sigma_{k_1}\times \Sigma_2}}&&\mathbf{B}^{n-k_1-k_2} U(1)_{\mathrm{conn}} }\,,
\]
where the left vertical equivalence is the one described in remark \ref{rem.yoneda}.
\end{remark}

\subsection{Extended higher Chern-Simons actions}
\label{extendedCS}

The general setting of $n$-dimensional Chern-Simons theory is that of a smooth (higher) stack of fields endowed with a $U(1)$-$n$-bundle with connection:
\[
\nabla:\mathbf{Fields}\to \mathbf{B}^nU(1)_{\mathrm{conn}},
\]
the \emph{prequantum $n$-bundle of the theory}.
By functoriality of the internal hom, this gives natural morphisms
\[
[\Sigma_k,\mathbf{Fields}]\xrightarrow{[\Sigma_k,\nabla]}[\Sigma_k,\mathbf{B}^nU(1)_{\mathrm{conn}}]
\]
for any closed oriented $k$-dimensional manifold $\Sigma_k$. If $k\leq n$ we can further compose with fiber integration along $\Sigma_k$ to 
%
get a prequantum bundle on the moduli stack of fields configuration on $\Sigma_k$.
\begin{definition}
Let
$
  \nabla
  :
  \mathbf{Fields}
  \to
  \mathbf{B}^n U(1)_{\mathrm{conn}}
$
be a differential characteristic map. Then for $\Sigma_k$
a closed smooth manifold of dimension $k \leq n$, we call
$$
  \exp(2 \pi i \int_{\Sigma_k} [\Sigma_k, \nabla] )
  :
  \xymatrix{
    [\Sigma_k, \mathbf{Fields}]
	\ar[rr]^-{[\Sigma_k, \nabla]}
	&&
	[\Sigma_k,\mathbf{B}^n U(1)_{\mathrm{conn}}]
	\ar[rr]^{\mathrm{hol}_{\Sigma_k}}
	&&
	\mathbf{B}^{n-k}U(1)_{\mathrm{conn}}
  }
$$
the \emph{off-shell prequantum $(n-k)$-bundle of the 
$\infty$-Chern-Simons theory} defined by $\nabla$. 
For $n = k$ we have a \emph{circle 0-bundle} 
$$
  \exp(2 \pi i \int_{\Sigma_n} [\Sigma_n, \nabla] )
  :
  \xymatrix{
    [\Sigma_n, \mathbf{Fields}]
	\ar[rr]^-{[\Sigma_n, \nabla]}
	&&
	[\Sigma_n,\mathbf{B}^n U(1)_{\mathrm{conn}}]
	\ar[rr]^-{\mathrm{hol}_{\Sigma_n}}
	&&
	\mathbf{U}(1)  }\;,
$$
which we call the  \emph{(stacky) action functional} of the theory.
\label{FiberIntegration}
\end{definition}
A detailed discussion of higher prequantum bundles as such is in 
\cite{hgp}.
\begin{remark}[Gauge invariance and smoothness]
 \label{gauge invariance and smoothness}
 The fact that the stacky action functional 
 \[
 [\Sigma_n,\mathbf{Fields}]\to \mathbf{U}(1)
 \]
 goes from the smooth (higher) stack of filed configurations on $\Sigma_n$ to the the sheaf $\mathbf{U}(1)$ of smooth $U(1)$-valued functions (which is a smooth 0-stack) tells us that the action functional is smooth. Moreover, since the target is a 0-stack (and so tehre are no nontrivial gauge transformations on the target), the action functional necessarily maps every gauge transformation in $[\Sigma_n,\mathbf{Fields}]$ to the identity. This means that once formulated in the language of higher stacks as in definition \ref{FiberIntegration}, the action functional is automatically smooth and gauge invariant. Evaluating the stacky action functional on the one-point manifold, we get a morphism of $\infty$-groupoids
 \[
 \mathbf{H}(\Sigma_n,\mathbf{Fields})\to U(1)_{\mathrm{disc}},
 \]
 where $U(1)_{\mathrm{disc}}$ is the discrete groupoid corresponding to $U(1)$ seen as a set. Passing to the sets of connected components, one finds the \emph{classical action functional} of the theory,
 \[
\exp(i S): \mathrm{Fields}(\Sigma_n)/\mathrm{gauge}\to U(1).
 \]

\end{remark}

The crucial example to have in mind is the following.
\begin{example}[Tradtional Chern-Simons theory]
Let $G$ be a compact simple and simply connected Lie group, and let 
\[
 \mathbf{c}_{\mathrm{conn}}:\mathbf{B}G_{\mathrm{conn}}\to \mathbf{B}^3U(1)_{\mathrm{conn}}
\]
be a differential refinement of a cohomology class $H^4(BG;\mathbb{Z})$ to a morphism of stacks from the stack of principal $G$-bundles with connection to the 3-stack of $U(1)$-3-bundles with connection (see \cite{brylinski-mclaughlin, FSS} for a detailed construction of these differential refinements). Then the induced 
classical action functional
\[
\{\text{Principal $G$-bundles with connection on $\Sigma_3$}\}/\mathrm{gauge}\to U(1)
\]
is the standard Chern-Simons action \cite{FreedCS}. Also the prequantum 1- and 2-bundles associated with $\mathbf{c}_{\mathrm{conn}}$ have a classical interpretation. Namely, 
$$
 \xymatrix{
    [S^1, \mathbf{B}G_{\mathrm{conn}}]
	\ar[rr]^-{[S^1, \mathbf{c}_{\mathrm{conn}}]}
	&&
	[S^1,\mathbf{B}^3 U(1)_{\mathrm{conn}}]
	\ar[rr]^-{\mathrm{hol}_{S^1}}
	&&
	\mathbf{B}^{2}U(1)_{\mathrm{conn}}
  }
$$
is the $U(1)$-bundle gerbe with connection inducing the Wess-Zumino-Witten gerbe on $G$ and 
 $$
  \xymatrix{
    [\Sigma_2, \mathbf{B}G_{\mathrm{conn}}]
	\ar[rr]^-{[\Sigma_2, \mathbf{c}_{\mathrm{conn}}]}
	&&
[\Sigma_2,\mathbf{B}^3 U(1)_{\mathrm{conn}}]	\ar[rr]^-{\mathrm{hol}_{\Sigma_2}}
	&&
	\mathbf{B}^{1}U(1)_{\mathrm{conn}}
  }
$$
is the $U(1)$-bundle with connection whose curvature form induces the canonical symplectic structure on the moduli space of flat $G$-connections on a Riemann surface $\Sigma_2$. See \cite{FiorenzaSatiSchreiberCS} for a careful discussion of these aspects of classical Chern-Simons theory from the point of view of smooth higher stacks.
%
\end{example}

\noindent Notice how the stacky realization subsumes several fundamental aspects 
of 
Chern-Simons theory:

\begin{enumerate}
  \item \emph{Gauge invariance and smoothness of the action functional}. This is remark \ref{gauge invariance and smoothness} above.
  
 \vskip 3mm
  
  \item \emph{Inclusion of instanton sectors (nontrivial $G$-bundles)}. Ordinary 3-dimensional Chern-Simons theory is
often discussed for the special case when the gauge group $G$ is connected 
and simply connected. This
yields a drastic simplification compared to the general case. Namely a 
compact 1-connected Lie group is automatically 2-connected, and so its classifying
space $BG$ is 3-connected. Hence 
every continuous map $\Sigma_3 \to B G$ out of a 3-manifold is homotopic to the trivial map.
This implies that every $G$-principal bundle over $\Sigma_3$ is trivializable. 
As a result, the moduli stack $[\Sigma_3, \mathbf{B}G_{\mathrm{conn}}]$
 of $G$-gauge fields on $\Sigma_3$, is actually equivalent 
to the stack of $\mathfrak{g}$-valued 1-forms on $\Sigma_3$ and
gauge transformations between these, which is indeed the familiar configurations 
space for 3-dimensional $G$-Chern-Simons theory.
\par
One should compare this to the case of 4-dimensional $G$-gauge theory on a 
4-dimensional manifold $\Sigma_4$,
such as $G$-Yang-Mills theory. In this case $G$-principal bundles may be nontrivial,
but are classified enirely by the second Chern class (or first Pontrjagin class) 
$[c_2] \in H^4(\Sigma_4, \pi_3(G))$.
In Yang-Mills theory with $G = SU(n)$, this class is known as the 
\emph{instanton number} of the gauge field.
\par
The simplest case where non-trivial classes occur already in dimension 3 
is the non-simply connected
gauge group $G = U(1)$, discussed in section \ref{3dU1CS} below. 
Here the moduli stack of fields 
$[\Sigma_3, \mathbf{B}U(1)_{\mathrm{conn}}]$ contains configurations which are not
given by globally defined 1-forms, but by connections on non-trivial circle bundles. 
By analogy with the case
of $SU(n)$-Yang-Mills theory, we will loosely refer to such field configurations 
as instanton field congurations,
too. In this case it is the first Chern class $[c_1] \in H^2(X,\mathbb{Z})$ 
that measures the non-triviality of the bundle.
If the first Chern-class of a $U(1)$-gauge field configurations happens to vanish, 
then the gauge field is again given by just a 1-form $A \in \Omega^1(\Sigma_3)$, 
the familiar gauge potential of electromagnetism. The value of the
3d Chern-Simons action functional on such a zero-instanton configuration 
is then simply the familiar expression 
$$
  \exp(i S(A)) = \exp(2 \pi i \int_{\Sigma_3} A \wedge d A)
  \,,
$$
where on the right we have the ordinary integration of the 3-form 
$A \wedge d A $ over $\Sigma_3$.

\vspace{3mm}
In the general case, however, i.e., when the configuration 
in $[\Sigma_3, \mathbf{B}U(1)_{\mathrm{conn}}]$ has non-trivial first Chern class,
the expression for the value of the action functional on this configuration is 
more complicated. If we pick a
good open cover $\{U_i \to \Sigma_3\}$, then locally on 
each patch $U_i$ the gauge field is given by
a 1-form $A_i$ and there is a contribution to the action functional coming from the integral over $U_\alpha$ 
of the local 3-form $A_\alpha \wedge d A_\alpha$ (suitably cut by the use of a partition of unit, see example \ref{example.holonomy}). But there are further terms to be included to get the correct action functional. It is precisely this what the fiber integration construction in definition \ref{FiberIntegration} achieves.
 
 \vspace{3mm}
  \item \emph{Level quantization.}   Traditionally, Chern-Simons theory in 3-dimensions
  with a gauge group $G$  which is connected and simply connected group 
  comes in a family parameterized by a \emph{level} $k \in \mathbb{Z}$.
  This level is secretly the cohomology class of the differential characteristic
  map 
  $
    \mathbf{c}_{\mathrm{conn}} : \mathbf{B}G_{\mathrm{conn}} \to \mathbf{B}^3 U(1)_{\mathrm{conn}}
  $
in 
  $
     H^4(B G, \mathbb{Z}) \simeq \mathbb{Z}$.
  So the traditional level is a cohomological shadow of the differential characteristic map 
  that we interpret as the off-shell prequantum $n$-bundle in full codimension $n$
  (down on the point). Notice that for a general smooth $\infty$-group $G$ the cohomology
  group $H^{n+1}(B G , \mathbb{Z})$ need not be equivalent to $\mathbb{Z}$ and so 
  in general the level need not be an integer.
  For for every smooth $\infty$-group $G$, and given 
  a morphism of moduli stacks 
  $\mathbf{c}_{\mathrm{conn}} : \mathbf{B}G_{\mathrm{conn}} \to \mathbf{B}^n U(1)_{\mathrm{conn}}$,
  also every integral multiple $k \mathbf{c}_{\mathrm{conn}}$ gives an
  $n$-dimensional Chern-Simons theory, ``at $k$-fold level''. The converse is in 
  general hard to establish: one is asking whether a given $\mathbf{c}_{\mathrm{conn}}$ can be divided
  by a fixed integer. 
  For instance for 3-dimensional Chern-Simons theory division by 2 may be possible
  for a Spin-structure. For 7-dimensional Chern-Simons theory division by 6 may be possible
  in the presence of a String-structure \cite{FiorenzaSatiSchreiberI}.
   
 \vspace{3mm}
  \item \emph{Definition on non-bounding manifolds and 
    relation to topological Yang-Mills on bounding manifolds.} Ordinary level $k$ 3-dimensional Chern-Simons theory is often defined on 
  bounding 3-manifolds $\Sigma_3$ by the formula
  $$
    \exp(i S(\nabla)) = \exp(2 \pi i k \int_{\Sigma_4} 
	\langle F_{\widehat \nabla} \wedge F_{\widehat \nabla}\rangle)
	\,,
  $$
  where $\Sigma_4$ is any 4-manifold with $\Sigma_3 = \partial \Sigma_4$ and where
  $\widehat \nabla$ is any extension of the gauge field configuration from 
  $\Sigma_3$ to $\Sigma_4$. Similar expressions exist for higher dimensional Chern-Simons theories.
  If one takes these expressions to be the actual definition of Chern-Simons
  action functional, then one needs extra discussion for which manifolds 
  (with desired structure) are bounding, hence which vanish in the respective
  cobordism ring, and, more seriously, one needs to exclude from the discussion those 
  manifolds which are not bounding.
  For example, in type IIB string theory one encounters the cobordism group 
$\Omega_{11}^{\rm Spin}(K(\Z, 6))$ \cite{Witten96}, which is 
proven to vanish in \cite{KS2}, meaning that all the desired manifolds happen
to be bounding.
\par
We emphasize that the formula for the action functinal given in definition \ref{FiberIntegration} applies generally,
 whether or not a manifold
 is bounding. Moreover, it is guaranteed that 
 \emph{if} $\Sigma_n$ happens to be bounding after all, then the action
 functional is equivalently given by integrating a higher curvature invariant over
 a bounding $(n+1)$-dimensional manifold.
 At the level of differential cohomology classes 
 of a differential manifold $X$, this is 
a well-known property 
 which is an explicit axiom in the equivalent
formulation by Cheeger-Simons differential characters:
 a Cheeger-Simons differential character of degree $(n+1)$ 
 is by definition a group homomorphism from (piecewise smooth) $n$-cycles in $X$ to $U(1)$ such 
 that whenever an $n$-cycle happens to be represented by $f_*\partial\Sigma_{n+1}$ for some smooth map $f:\Sigma_{n+1}\to X$, the value in $U(1)$
 is given by the exponentiated integral of $f^*\omega$ over $\Sigma_{n+1}$, for a fixed $(n+1)$-form $\omega$ over $X$  (a review and further pointers are given in \cite{HopkinsSinger}).
 With reference to such differential characters, Chern-Simons action 
 functions have been formulated for instance in \cite{Witten96,Witten98}.
 The sheaf hypercohomology classes of the Deligne complex that we are concerned with  
 here are well known to be equivalent
 to these differential characters, and {\v C}ech-Deligne cohomology has the 
 advantage that, with the results from
  \cite{gomi-terashima1} which are at the base of  
 definition \ref{FiberIntegration} above, it yields explict formulas for the action functional
 on non-bounding manifolds in terms of local differential form data.

%
%
%
%
%

  

\end{enumerate}

\subsection{Higher abelian Chern-Simons theories with background charge and quadratic refinement} 
\label{CSWithBackgroundCharge}

Let $\Sigma_{4n+3}$ be an $4n+3$-dimensional closed oriented manifold. Then the Beilinson-Deligne cup product and  holonomy along $\Sigma$ define a morphism of stacks
\begin{align*}
[\Sigma_{4n+3},\mathbf{B}^{2n+1}U(1)_{\mathrm{conn}}]&\times [\Sigma_{4n+3},\mathbf{B}^{2n+1}U(1)_{\mathrm{conn}}]\cong
[\Sigma_{4n+3},\mathbf{B}^{2n+1}U(1)_{\mathrm{conn}}\times \mathbf{B}^{2n+1}U(1)_{\mathrm{conn}}]\xrightarrow{\cup_{\mathrm{conn}}}\\
&\xrightarrow{\cup_{\mathrm{conn}}}
[\Sigma_{4n+3},\mathbf{B}^{4n+3}U(1)_{\mathrm{conn}}]\xrightarrow{\mathrm{hol}_{\Sigma_{4n+3}}} \mathbf{U}(1)
\end{align*}
inducing the
\emph{intersection pairing}
\begin{align*}
 \langle -,-\rangle: H^{2n+2}_{\mathrm{diff}}(\Sigma_{4n+3};\mathbb{Z})\otimes H^{2n+2}_{\mathrm{diff}}(\Sigma_{4n+3};\mathbb{Z})\to  U(1)\\
 ([\hat{a}],[\hat{b}])\mapsto \exp(2\pi i\int _{\Sigma_{4n+3}}\hat{a}\cup_{\mathbf{BD}} \hat{b})
\end{align*}
in Deligne cohomology.
\par
For $n = 0$ this yields the action functional of 
ordinary 3d Chern-Simons theory with gauge group $U(1)$. Notice that 
its expression in terms of the differential cup product on differential
cocycles in this case is what takes care of the fact that there are 
in general non-trivializable $U(1)$-principal bundles on a 3-dimensional manifold.
This is contrary to the case of a simply connected gauge group, 
for which over a 3-manifold every bundle is trivilizable, so that 
every pricipal connection is given by a globally defined differential 1-form.
The above cup product formula takes care of the subtlety that this is not
in general the case for $U(1)$-principal connections on a 3-manifold.

One step up the hierarchy, for $n =1$, the above action functional defines a 7-dimensional
Chern-Simons theory of principal 3-form connections. This theory appears
as the bosonic abelian sector of the Chern-Simons term of 11-dimensional 
supergravity, when the latter is compactified on a 4-sphere. 
This is important, since $\mathrm{AdS}_7/\mathrm{CFT}_6$-duality asserts
that the holographic dual of this 7-dimensional theory controls the 
self-dual 2-form theory on the worldvolume of the M5-brane
\cite{Witten96}.  
This we turn to below in \ref{7dSugra}. However, in order for this
holographic relation to work precisely, this 7-dimensional cup product 
action functional needs to be ``quadratically refined'' by imposing
a ``flux quantization'' condition
\cite{WittenFluxQuantization}. This was formalized 
on the level of differential cohomology classes in 
\cite{HopkinsSinger} via \emph{integral Wu classes}
and further refined to the level of smooth moduli stacks, as
discussed here, in our previous
article in \cite{FiorenzaSatiSchreiberII}, via a smooth stacky refinement
of integral Wu classes.
Since this is an important subtlety, we now explain 
this into the present context.	

Generally, an interesting question is to decide whether the intersection pairing admits a \emph{quadratic refinement}, i.e., if there exist a  function
$$
  q : H_{\mathrm{diff}}^{2n+2}(\Sigma_{4n+3};\mathbb{Z}) \to U(1)
$$
 for which
the intersection pairing is obtained via the polarization formula
$$
  \langle \hat a, \hat b\rangle 
    = 
  q([\hat a] + [\hat b])\,
  q([\hat a])^{-1}\,
  q([\hat b])^{-1}\,
  q(0)
  \,.
$$
The naive guess is, obviously,
\[
 q([\hat{a}])= \langle [\hat{a}],[\hat{a}]\rangle^{1/2}=\exp(\pi i\int _{\Sigma_{4n+3}}\hat{a}\cup_{\mathbf{BD}} \hat{a})\;,
 \]
but this is actually not a well defined expression due to the sign ambiguity in the square root: as $\hat{a}$ varies in the Deligne cohomology class $[\hat{a}]$, the integral $\int _{\Sigma_{4n+3}}\hat{a}\cup_{\mathbf{BD}} \hat{a}$ is only well defined mod $\mathbb{Z}$ and so the exponential which should define $q([\hat{a}])$ is only well defined up to a sign. The problem is that, since the differential classes in $H^{2n+2}_{\mathrm{diff}}(\Sigma_{4n+3};\mathbb{Z})$ refine 
\emph{integral} cohomology, we cannot in general simply divide by 2 and
pass from $\exp(2\pi i \int_{\Sigma_{4n+3}} \hat a \cup \hat a)$ to 
$\exp( \pi i  \int_{\Sigma_{4n+3}}  \hat a \cup \hat a)$: the integrand in the
latter expression just does not make any sense in general in differential cohomology.
Namely, if one tried to write it out in the ``obvious'' local formulas one would
find that it is a functional on fields which is not gauge invariant.
This phenomenon is not surprising. Something analogous happens with Chern-Simons theory
with simply-connected gauge group $G$, where the theory is consistent only at 
integer \emph{levels}; here the ``level'' is nothing but the 
underlying integral class $[a] \cup [a]$, where $[a]$ is the integral cohomology class underlying the Deligne cohomology class $[\hat{a}]$, and the theory is consistent only at ``levels which are divisible by 2''. Let us briefly explain this, by pushing forward the analogy with traditional Chern-Simons theory. Since, by dimensional reasons, $H^{4n+4}(\Sigma_{4n+3};\mathbb{Z})=0$, the underlying topological $U(1)$-$(4n+3)$-bundle of an  $U(1)$-$(4n+3)$-bundle with connection in the equivalence class $\hat{a}$ is trivializable, and one can compute the integral $\int _{\Sigma_{4n+3}}\hat{a}\cup_{\mathbf{BD}} \hat{a}$ in terms of a chosen trivialization (see Example \ref{example.holonomy} for a toy version of this computation). The ambiguity in the value of the integral is then given by the effect of a change of trivialization. If we assume that $\Sigma_{4n+3}$ is a boundary of some oriented $4n$-dimensional manifold, then the values for the integral for two different trivializations are given by the integer
\[
\int_{\Sigma_{4n+4}}[a]\cup [a]\;,
\]
where $\Sigma_{4n+4}$ is some closed oriented $4n+4$-dimensional manifolds containing $\Sigma_{4n+3}$ as an hypersurface, and where $[a]$ is the integral cohomology class underlying an extension of $[\hat{a}]$ to $\Sigma_{4n+4}$. It is precisely this integral over $\Sigma_{4n+4}$ that needs to be an even integer in order to avoid the sign ambiguity in the square root of $\langle[\hat{a}],[\hat{a}]\rangle$, and so what we want is  $[a]\cup [a]$ to be divisible by 2 in $H^{4n+4}(\Sigma_{4n+4};\mathbb{Z})$.
\par
Although the divisibility of  $[a]\cup [a]$ is an apparently  insurmountable obstruction to having a quadratic refinement of the cup-product intersection pairing, there is  actually a systematic way to obtain a square root of the quadratic form $Q([\hat{a}]):=\langle [\hat{a}],[\hat{a}]\rangle$ by \emph{shifting it}. Here we think of 
the analogy with a
quadratic form 
$$
  Q : x \mapsto x^2
$$ 
on the real numbers (a parabola in the plane).
Replacing this by 
$$
  Q^{\lambda} : x \mapsto x^2 - \lambda x
$$ 
for some real number
$\lambda$ means keeping the shape of the form, but shifting its minimum from 0 to
$\frac{1}{2}\lambda$. If we think of this as the potential term for a scalar field
$x$ then its ground state is now at $x = \frac{1}{2}\lambda$. We may say that there is
a \emph{background field} or \emph{background charge} that pushes
the field out of its free equilibrium.
\par
In order to apply this reasoning to the action quadratic form 
$Q([\hat{a}])$ on differential cocycles in a way that leads to a well-defined square root, we need a 
differential class $[\hat \lambda] \in H^{2n+2}_{\mathrm{diff}}(\Sigma_{4n+3})$ such that 
for every $[\hat a] \in H^{2n+2}_{\mathrm{diff}}(\Sigma_{4n+3})$ the class
$$
  [a] \cup [a] - [a] \cup [\lambda]
$$
in $H_{\mathrm{diff}}^{4n+4}(\Sigma_{4n+4})$ is even, i.e., it is divisible by 2 (where, as above, $[a]$ and $[\lambda]$ denote the underlying topological classes of extension of $[\hat{a}]$ and $[\hat{\lambda}]$ to $\Sigma_{4n+4}$). Because then the integral
\[
\int_{\Sigma_{4n+3}} \hat a \cup \hat a - \hat a \cup \hat \lambda
\]
will be well-defined mod $2\mathbb{Z}$ and so the \emph{shifted
action functional}
$$
q^{[\hat{\lambda}]}:= [\hat a] \mapsto \exp\left(\pi
  i \int_{\Sigma_{4n+3}} \hat a \cup_{\mathrm{BD}} \hat a - \hat a \cup_{\mathrm{BD}} \hat \lambda\right)
  \,,
$$
will be well defined. 
One directly sees that this shifted action is indeed a
quadratic refinement of the intersection pairing:
$$
q^{[\hat{\lambda}]}([\hat a] + [\hat b])
q^{[\hat{\lambda}]}([\hat a])^{-1}
  q^{[\hat{\lambda}]}([\hat b])^{-1}
  q^{[\hat{\lambda}]}(0)
  = 
  \exp(2\pi i \int_{\Sigma_{4n+3}} \hat a \cup \hat b)
  \,.
$$
The divisibility by 2 property we are requiring on $[a] \cup [a] - [a] \cup [\lambda]$  means that this class vanishes under the reduction mod 2 map
$$
  (-)_{\mathbb{Z}_2} : H^{4n+4}(\Sigma_{4n+4}, \mathbb{Z}) \to H^{4n+4}(\Sigma_{4n+4}, \mathbb{Z}_2)
$$
from integral cohomology to to $\mathbb{Z}_2$-cohomology. When one passes to $\mathbb{Z}_2$-cohomology, it is well known that, for any nonegative integer $k$, on every oriented manifold $X$ there exist a unique class $\nu_{2k+2} \in H^{2k+2}(X,\mathbb{Z}_2)$, namely, the \emph{Wu class},  such that 
$$
   [a] \cup [a] 
     - 
   [a] \cup \nu_{2k+2} 
   = 
   0
$$
 in $H^{4k+4}(X, \mathbb{Z}_2)$ for any cohomology class $[a] $ in $H^{2k+2}(X,\mathbb{Z}_2)$. Moreover, if $X$ is a $\mathrm{Spin}$-manifold, then every
Wu class $\nu_{2k+2}$  with $2k+2\equiv 0 \mod 4$ odd  can be lift to integral cohomology, so in these cases there actually exist an integral cohomology class $\lambda$ such that 
$$
   [a] \cup [a] 
     - 
   [a] \cup [\lambda] 
   \equiv 
   0 \mod 2
$$
in $H^{4k+4}(X, \mathbb{Z})$ for any cohomology class $[a] $ in $H^{2k+2}(X,\mathbb{Z})$. This is almost what we were looking for: 
the last step at the level cohomology is to lift this class $[\lambda]$ from integral to differential cohomology. To do this one needs a further piece of information. Namely, that the lift of Wu classes to integral cohomology 
is given by 
polynomials in the Pontrjagin classes of $X$
(see section E.1 of \cite{HopkinsSinger}). For instance
the degree 4 Wu class is refined by the first fractional Pontrjagin class
$\frac{1}{2}p_1$
$$
  (\tfrac{1}{2}p_1)_{\mathbb{Z}_2} = \nu_4
  \,.
$$
This was observed in \cite{Witten96} (see around eq. (3.3) there). 
So the problem of quadratic refinement of the intersection pairing in Deligne cohomology in dimension $4n+3$ is solved by taking differential refinements 
of the Pontrjagin classes \cite{brylinski-mclaughlin}. Finally one goes back
to an action functional with purely quadratic terms by introducing
the shifted field
$$
 \hat G := \hat a - \tfrac{1}{2}\hat \lambda
$$
in terms of which th above action functional takes again the quadratic 
but globally shifted form
$\exp(\pi i\int_X \hat G \cup_{\mathrm{BD}} \hat G 
- \tfrac{1}{2}\hat \lambda \cup_{\mathrm{BD}} \tfrac{1}{2}\hat \lambda  )$. 

But so far all of this is on the level of differential cohomology
classes only. In the context of \emph{extended} higher cup product 
theories that are the topic here, one needs to further lift this quadratic refinement
from differential cohomology to differential cocycles, hence from Deligne cohomology
to the moduli stacks of higher connections.  This first of all means to
refine the differential Pontrjagin classes to morphisms of smooth stacks. 
For the first and second Pontrjagin class this was accomplished in \cite{FSS},
yielding maps of the form
\[
\tfrac{1}{2}\hat{\mathbf{p}}_1:\mathbf{B}\mathrm{Spin}_{\mathrm{conn}}\to \mathbf{B}^3U(1)_{\mathrm{conn}}
\]
and
\[
\tfrac{1}{6}\hat{\mathbf{p}}_2:\mathbf{B}\mathrm{String}_{\mathrm{conn}}\to \mathbf{B}^7U(1)_{\mathrm{conn}}\;,
\]
where $\mathbf{B}\mathrm{String}_{\mathrm{conn}}$ is the 2-stack of String 2-connections.

Now to refine the above ``flux quantization condition'' 
$\hat G = \hat a - (\tfrac{1}{2}\hat \lambda)^2$ from cohomology to cocycles means
to replace the \emph{equation} on cohomology classes with the corresponding 
\emph{homotopy fiber product} of such maps of moduli stacks. This leads to 
a mixing of the $\mathrm{Spin}$-connection with the ($2n +1$)-form gauge field to a 
higher gauge field.  For the 7d cup product theory
this was discussed in detail in \cite{FiorenzaSatiSchreiberII}.
We will come back to the simpler case of the quadratic refinement
via mixed terms below in example \ref{CScharge}.

\section{Examples and applications}
\label{ExamplesAndApplications}

Here we list and discuss examples of higher extended cup-product Chern-Simons theories
constructed by the general procedure introduced above in section \ref{GeneralTheory}.
Some of the examples below are known from constructions
in string theory and M-theory, while others have maybe not been considered before. Even in the known cases, only their action functionals
(codimension 0) and their prequantum line bundles (codimension 1) are usually investigated.
Our discussion provides the refinement of the action functional
\begin{enumerate}
  \item to the full higher moduli stacks of fields;
  \item to arbitrary codimension.
\end{enumerate}
The titles of the following subsections follow the pattern
\begin{center}
  XYZ Chern-Simons theory $\;$ and $\;$ ABC theory
\end{center}
where ``ABC theory'' is an incarnation of the extended ``XYZ Chern-Simons theory''  in 
higher codimension.

\medskip
Before we proceed to section \ref{UnaryExamples},
the following list gives an overview of the various types of examples that we consider,
and how they conceptually relate to each other as specializations and/or 
combinations of other classes of 
examples. 

\vspace{5mm}
\noindent{\bf List of classes of examples.}
\begin{enumerate}
  \item
   {\bf Fully general $\infty$-Chern-Simons theory.}
  In full generality, an ``$\infty$-Chern-Simons theory'' is specified
  by a smooth gauge $\infty$-group $G$ and a differential characteristic map
  of moduli stacks
  $$
    \mathbf{c}_{\mathrm{conn}} : \mathbf{B}G_{\mathrm{conn}} \to 
	\mathbf{B}^n U(1)_{\mathrm{conn}}
	\,.
  $$
  This is such that, for $\Sigma_k$ a $k$-dimensional smooth manifold, the object
  $[\Sigma_k, \mathbf{B}G_{\mathrm{conn}}]$ 
  discussed in section \ref{FibInt}
  is the moduli stack of $G$-gauge fields on 
  $\Sigma_k$, and $[\Sigma_k, \mathbf{B}^n U(1)_{\mathrm{conn}}]$ is the 
  moduli stack of $U(1)$-$n$-bundles with connection on $\Sigma_k$. Then if $0 \leq k \leq n$ and $\Sigma_k$
  is closed and oriented, we obtain a morphism
  $$
    \exp(2 \pi i \int_{\Sigma_k} [\Sigma_k, \mathbf{c}_{\mathrm{conn}}])
	: 
	[\Sigma_k, \mathbf{B}G_{\mathrm{conn}}] \to \mathbf{B}^{n-k}U(1)_{\mathrm{conn}}
  $$
  as in section \ref{FibInt}, which gives the \emph{off-shell prequantum $(n-k)$-bundle}
  of an $n$-dimensional Chern-Simons theory. In particular, for $k = n$ this is the 
  \emph{action functional} of the higher extended Chern-Simons theory specified by 
  $\mathbf{c}_{\mathrm{conn}}$.
  \item
    {\bf Inhomogeneous $U(1)$ cup-product theories.}
    In this general context, the cup product 
	$$
	  \cup_{\mathrm{conn}} : \mathbf{B}^p U(1)_{\mathrm{conn}} \times
	  \mathbf{B}^q U(1)_{\mathrm{conn}} \to  \mathbf{B}^{p+q+1} U(1)_{\mathrm{conn}}
	$$
	from section \ref{Sec BD} for $p,q \geq 1$, is itself a differential characteristic map, 
	since we may regard it as defining an $\infty$-Chern-Simons theory with gauge 
	$\infty$-group the product  
	$(\mathbf{B}^{p-1} U(1))\times (\mathbf{B}^{q-1} U(1))$, hence by reading
	$\cup_{\mathrm{conn}}$ as
	$$
	  \cup_{\mathrm{conn}} 
	    : 
	  \mathbf{B}\left(
	     \mathbf{B}^{p-1} U(1)\times \mathbf{B}^{q-1} U(1)\right)_{\mathrm{conn}}
	    \to  
	  \mathbf{B}^{p+q+1} U(1)_{\mathrm{conn}}
	  \,.
	$$
   A class of examples of this form that does appear in the physics literature is the
   electric-magnetic coupling term in higher abelian gauge theory. This is a \emph{section}
   of the prequantum 1-bundle of this Chern-Simons theory, and the class of that bundle
   is the electric-magnetic quantum anomaly.
   
   Two variants of this theory are important.
   \begin{enumerate}
    \item
	  {\bf $U(1)$ cup-square theories.}
     In the case that
	  $p = q$ we may restrict to the diagonal of the cup pairing, hence taking the two $p$-form
	fields to be two copies of one single field. Formally this means that we are considering the 
	differential characteristic map which is the composite
	$$
	    (-)^{\cup^2_{\mathrm{conn}}}
		:
	  \xymatrix{
		\mathbf{B}^p U(1)_{\mathrm{conn}}
		\ar[r]^-{\Delta}
		&
		\mathbf{B}^p U(1)_{\mathrm{conn}} \times
		\mathbf{B}^p U(1)_{\mathrm{conn}}
		\ar[rr]^-{\cup_{\mathrm{conn}}}
		&&
		\mathbf{B}^{2p+1} U(1)_{\mathrm{conn}}
	  }
	  \,.
	$$
	For $p = 1$, this yields (the higher codimension-extended version of)
	traditional 3-dimensional $U(1)$-Chern-Simons theory. For $p = 2k+1$
	it yields the $(4k+3)$-dimensional $U(1)$-Chern-Simons theory which 
	is the holographic dual of self-dual $2k$-form theory in dimension $4k+2$.
	\item
	  {\bf Cup product of two nonabelian theories.}
    Given two possibly nonabelian gauge $\infty$-groups $G_1$ and $G_2$
	equipped with two differential characteristic maps
	$(\mathbf{c}_1)_{\mathrm{conn}}$ and $(\mathbf{c}_1)_{\mathrm{conn}}$,
	we may form the ``cup product of two nonabelian Chern-Simons theories''
	$$
	\hspace{-0.7cm}
	  (\mathbf{c}_1)_{\mathrm{conn}}
	  \cup_{\mathrm{conn}}
	  (\mathbf{c}_2)_{\mathrm{conn}}
	  :
	  \xymatrix{
	    \mathbf{B} (G_1 \times G_2)_{\mathrm{conn}}
		\ar[rr]^-{ ((\mathbf{c}_1)_{\mathrm{conn}}, (\mathbf{c}_2)_{\mathrm{conn}}) }~~
		&&
		~~\mathbf{B}^p U(1)_{\mathrm{conn}}
		\times
		\mathbf{B}^q U(1)_{\mathrm{conn}}
		\ar[r]^-{\hspace{-2mm}\cup_{\mathrm{conn}}}
		&
		\mathbf{B}^{p+q+1}U(1)_{\mathrm{conn}}
	  }
	  \,.
	$$
	This appears for instance in the electric-magnetic anomaly of the heterotic string.
	\end{enumerate}
  \item
    {\bf Cup-square of one non-abelian theory.}
    The two variants above may be combined to yield the cup product of a 
	non-abelian Chern-Simons theory with itself.
  \item
    {\bf Multiple-factor cup-product theories.}
    Finally, all of this can be considered with three cup factors
	(``cubic theories'') or more cup-factors, instead
	of just two of them as in the ``quadratic theories'' described above. 
	Examples of cubic Chern-Simons theories appear, for instance, in
	11-dimensional supergravity.
\end{enumerate}

\subsection{Unary examples}
\label{UnaryExamples}

Before discussing genuine cup-product higher Chern-Simons 
theories we consider here some indecomposable theories --
\emph{unary} cup-product theories, if one wishes -- 
that serve as building blocks for the cup product theories.

\subsubsection{Higher differential Dixmier-Douady class and higher dimensional $U(1)$-holonomy}
\label{VolumeHolonomy}

The \emph{degenerate} or rather \emph{tautological} 
case of extended $\infty$-Chern-Simons theories nevertheless
deserves special attention, since it appears universally in all other examples:
it is the case where the extended action functional is the \emph{identity} morphism
$$
  (\mathbf{DD}_n)_{\mathrm{conn}} 
    : 
   \xymatrix{
	\mathbf{B}^n U(1)_{\mathrm{conn}} 
     \ar[r]^{\mathrm{id}}
	 &
	 \mathbf{B}^n U(1)_{\mathrm{conn}}
    }\;,
$$
for some $n \in \mathbb{N}$.
Trivial as this may seem, this is the differential refinement of what 
is called the (higher) \emph{universal Dixmier-Douady class} 
the higher universal first Chern class -- of circle $n$-bundles
/ bundle $(n-1)$-gerbes, which on the topological classifying space $B^n U(1)$
is the weak homotopy equivalence
$$
  \mathrm{DD}_n
  :
  \xymatrix{
    B^n U(1)
	\ar[r]^{\hspace{-3mm}\simeq}
	&
	K(\mathbb{Z}, n+1)
  }
  \,.
$$
Therefore, we are entitled to consider
$(\mathbf{DD}_n)_{\mathrm{conn}}$ as the extended action functional
of an $n$-dimensional $\infty$-Chern-Simons theory. Over an $n$-dimensional
manifold $\Sigma_n$ the moduli $n$-stack of field configurations is 
that of circle $n$-bundles with connection on $\Sigma_n$.
In generalization to how a circle 1-bundle with connection
has a \emph{holonomy} over closed 1-dimensional oriented manifolds, we note
 that
 a circle $n$-connection
has a \emph{$n$-volume holonomy} over an $n$-dimensional closed oriented manifold $\Sigma_n$.
This is the ordinary (codimension-0) action functional associated to $(\mathbf{DD}_n)_{\mathrm{conn}}$
regarded as an extended action functional:
$$
  \mathrm{hol} := \exp(2 \pi i \int_{\Sigma_n} [\Sigma_n, (\mathbf{DD}_n)_{\mathrm{conn}}])
  :
  [\Sigma_n, \mathbf{B}^n U(1)_{\mathrm{conn}}]
  \to 
  U(1)
  \,.
$$
This formulation makes it manifest that, for $G$ any smooth $\infty$-group and
$\mathbf{c}_{\mathrm{conn}} : \mathbf{B}G_{\mathrm{conn}} \to \mathbf{B}^n U(1)_{\mathrm{conn}}$
any extended $\infty$-Chern-Simons action functional in codimension $n$, the 
induced action functional is indeed the $n$-volume holonomy of a family of 
``Chern-Simons circle $n$-connections'', i.e., that we have
$$
  \exp(2 \pi i \int_{\Sigma_n} [\Sigma_n, \mathbf{c}_{\mathrm{conn}}])
  \simeq
  \mathrm{hol}_{\mathbf{c}_{\mathrm{conn}}}
  \,.
$$
This is most familiar in the case where the moduli $\infty$-stack $\mathbf{B}G_{\mathrm{conn}}$
is replaced with an ordinary smooth oriented manifold $X$ 
(of any dimension and not necessarily
compact). In this case $\mathbf{c}_{\mathrm{conn}} : X \to \mathbf{B}^n U(1)_{\mathrm{conn}}$
modulates a circle $n$-bundle with connection $\nabla$ on this smooth manifold. 
Now regarding
this as an extended Chern-Simons action function in codimension $n$ means to 
\begin{enumerate}
  \item
    take the moduli stack of fields over a given closed oriented
	manifold $\Sigma_n$ to be $[\Sigma_n,X]$, which is simply the
	space of maps between these manifolds, equipped with its natural 
	(``diffeological'') smooth structure (for instance the smooth loop space
	$L X$ when $n = 1$ and $\Sigma_n = S^1$);
  \item 
	take the value of the action functional on a field configuration
	$\phi : \Sigma_n \to X$ to be the $n$-volume holonomy
	of $\nabla$
	$$
	  \mathrm{hol}_\nabla(\phi)
	  =
	  \exp(2 \pi i \int_{\Sigma_n} [\Sigma_n, \mathbf{c}_{\mathrm{conn}}] )
	  :
	  \xymatrix{
	    [\Sigma_n, X]
		\ar[rr]^-{[\Sigma_n,\mathbf{c}_{\mathrm{conn}}]}
		&&
		[\Sigma_n, \mathbf{B}^n U(1)_{\mathrm{conn}}]
		\ar[rrr]^-{\exp(2 \pi \int_{\Sigma_n}(-))}
		&&&
		U(1)
	  }
	  \,.
	$$
\end{enumerate}
Using 
the constructions in \ref{FibInt}
to unwind this in terms of local differential
form data, this reproduces the familiar formulas for (higher) $U(1)$-holonomy.

\subsubsection{Ordinary $3d$ $\mathrm{Spin}$-Chern-Simons theory and String-2-connections}
\label{3dCSSpin}

For $G$ any connected and simply connected compact simple Lie group
we have $H^4(B G, \mathbb{Z}) \simeq \mathbb{Z}$. In the case that
$G = \mathrm{Spin}$ is the spin group (in dimension $\geq 3$), the
generator (unique up to sign) of this group is called the
\emph{first fractional Pontrjagin class}, represented by a map
$$
  \tfrac{1}{2}p_1 : B \mathrm{Spin} \to B^3 U(1) \simeq K(\mathbb{Z}, 4)
  \,.
$$
In \cite{cohesive} it is shown that this has a unique (up to equivalence)
smooth refinement to a morphism of higher smooth moduli stacks of the form
$$
  \tfrac{1}{2}\mathbf{p}_1 : \mathbf{B}\mathrm{Spin} \to \mathbf{B}^3 U(1)
  \,.
$$
Moreover, in \cite{FSS} we construct the further differential refinement
$$
  \tfrac{1}{2}(\mathbf{p}_1)_{\mathrm{conn}}
  :
  \mathbf{B}\mathrm{Spin}_{\mathrm{conn}}
  \to 
  \mathbf{B}^3 U(1)_{\mathrm{conn}}
$$
from the moduli stack of $\mathrm{Spin}$-principal bundles with connection
to the smooth moduli 3-stack of smooth circle 3-bundles (bundle 2-gerbes) with connection.
Regarding this as an extended action functional for an $\infty$-Chern-Simons theory,
it is not hard to see that the corresponding action functional 
$$
  \exp(2 \pi i \int_{\Sigma_3})
  :
  [\Sigma_3, \mathbf{B}\mathrm{Spin}_{\mathrm{conn}}]
  \to
  U(1)
$$
is that 
of ordinary 3d $\mathrm{Spin}$-Chern-Simons theory, 
as discussed in the Introduction, section \ref{Introduction}
(in terms of cohomology classes this was first highlighted in \cite{CJMSW}).

\medskip
In addition to the comments on ordinary Chern-Simons theory regarded
as an extended prequantized theory already made in the Introduction, 
we here observe the following. The total space 2-stack is of the prequantum circle 3-bundle 
of this theory, regarded as an 0-1-2-3 extended prequantum Chern-Simons theory, is,
by Corollary \ref{TotalSpaceOfnConnection}, the homotopy pullback  of the form
$$
  \raisebox{20pt}{
  \xymatrix{
     \mathbf{B}\mathrm{String}_{\mathrm{conn}'}
	 \ar[r]
	 \ar[d]
	 &
	 \Omega^{1 \leq \bullet \leq 3}
	 \ar[r]
	 \ar[d]
	 &
	 {*}
	 \ar[d]
	 \\
	 \mathbf{B}\mathrm{Spin}_{\mathrm{conn}}
	 \ar[r]_-{\tfrac{1}{2}(\mathbf{p}_1)_{\mathrm{conn}}}
	 &
	 \mathbf{B}^3 U(1)_{\mathrm{conn}}
	 \ar[r]_{\text{forget}}
	 &
	 \mathbf{B}^3 U(1)~.
  }
  }
$$
Comparison with \cite{FiorenzaSatiSchreiberII} shows that this total space is
the moduli 2-stack $\mathbf{B}\mathrm{String}_{\mathrm{conn}'}$ of (twisted)
$\mathrm{String}$-principal 2-connections, as indicated (see the appendix to 
\cite{FiorenzaSatiSchreiberI} for a discussion of how these are 
nonabelian 2-form connections). If one further restricts along
the inclusion 
$\Omega^3 \hookrightarrow \Omega^{1 \leq \bullet \leq 3}$,
then these restrict
to the structures discussed in \cite{Waldorf}. Further 
restriction along the inclusion $\{0\}\hookrightarrow \Omega^3$ leads to 
the moduli 2-stack $\mathbf{B}\mathrm{String}_{\mathrm{conn}}$ of 
\emph{strict} $\mathrm{String}$-principal 2-connections.
If, on the other hand, one replaces
the twist by a fixed 3-form with the twist by the differential second Chern-class
of an $E_8 \times E_8$-principal bundle
$$
  \mathbf{a} : \mathbf{B}(E_8 \times E_8) \to \mathbf{B}^3 U(1)
$$
then one obtains the moduli 2-stack of $\mathrm{String}^{\mathbf{a}}$-connections
that control the anomaly-free field content, including the twisted B-field,
of the heterotic Green-Schwarz mechanism as discussed in \cite{SSSIII}.

\begin{remark}
Note that there are other effectively unary theories which fall under our formulation;
notably, those whose action functional takes the form $\int \Omega \wedge CS$,
where $CS$ is a Chern-Simons term, not necessarily three-dimensional,
and $\Omega$ is an auxiliary form on the underlying manifold, independent of the
Chern-Simons term. Since $\Omega$ is a fixed form it does not enter into the dynamics
and so the whole system is governed by the Chern-Simons term.
Examples include 
\begin{enumerate}
\item K\"ahler-Chern-Simons theories (see \cite{NS, LMNS, IKU, Li})
where $\Omega$ is a K\"ahler form,

\item holomorphic Chern-Simons theories (see \cite{Holom, FT}) where $\Omega$ is 
a middle form on a Calabi-Yau manifold, 

\item theories where $\Omega$ is a form on special holonomy manifolds (see \cite{BLN, BKS}), 
as well as 

\item theories
that lift M-theory via terms of the form $\int_{M^{27}} \Omega_{16} \wedge CS_{11}$ \cite{OP2},
where $CS_{11}$ is the Chern-Simons term in M-theory \eqref{CGG}, and $\Omega_{16}$ is
a composite form on the octonionic projective plane.
\end{enumerate}
\end{remark}

\subsubsection{$7d$ $\mathrm{String}$-Chern-Simons theory and Fivebrane 6-connections}
\label{7dCSString}

The construction of the total space of the fully extended prequantum $n$-bundle
for the $\mathrm{Spin}$ group in section \ref{3dCSSpin} above is just the first step in a whole tower of
\emph{higher Spin structure} and (extended) \emph{higher Spin-Chern-Simons theories}
that are obtained by a smooth and differential refinement of the
\emph{Whitehead tower} of $B O$. This is the tower of homotopy types on the
left vertical axis of the following diagram:
$$
\hspace{1.5cm}
  \xymatrix{
    \vdots
    \\
    \vspace{-3mm}
    B \mathrm{Fivebrane}
	\ar[d]
	\ar[r]
	&
	\cdots
	\ar[r]
	&
	{*}
	\ar[d]
    \\
    B \mathrm{String}
	\ar[d]
	\ar@{-}[r]
	&
	\cdots
	\ar[r]^<<<<<{\tfrac{1}{6}p_2}
	&
	B^8 \mathbb{Z}
	\ar[r]
	\ar[r]
	\ar@{-}[d]
	&
	{*}
	\ar[d]
	&&&&&&&
    \\
    B \mathrm{Spin}
	\ar[d]
	\ar@{-}[r]
	&
	\cdots
	\ar[rr]^<<<<<{\tfrac{1}{2}p_1}
	&
	\ar@{-}[d]
	&
	B^4 \mathbb{Z}
	\ar[r]
	\ar@{-}[d]
	&
	{*}
	\ar[d]
	&&&&&&&&&&&
    \\
    B SO
	\ar[d]
	\ar@{-}[r]
	&
	\cdots
	\ar[rrr]^<<<<<{w_2}
	&
	\ar[d]&
	\ar[d]
	&
	B^2 \mathbb{Z}_2
	\ar[r]
	\ar[d]
	&
	{*}
	\ar[d]
	&&&&&&&&
    \\
    B O
	\ar[r]
	\ar[d]
	\ar@/_1pc/[rrrrr]_<<<<<<<<<<<<<{w_1}
	&
	\cdots
	\ar[r]
	&
	\tau_{\leq 8} B O
	\ar[r]
	&
	\tau_{\leq 4} B O
	\ar[r]
	&
	\tau_{\leq 2} B O
	\ar[r]
	&
	\tau_{\leq 1} B O \simeq B \mathbb{Z}_2
	&&
	\\
	B \mathrm{GL}\;.
  }
  \,.
  $$
Here the bottom horizontal tower is the Postnikov tower of $B O$
and all rectangles are homotopy pullbacks (see section 4 of \cite{cohesive} for more details). 

\medskip
For $X$ a smooth manifold, there is a canonically given map $X \to B \mathrm{GL}$,
which classifies the tangent bundle $T X$. The lifts of this classifying map
through the above Whitehead tower correspond to structures on $X$ as indicated
in the following diagram:
$$
  
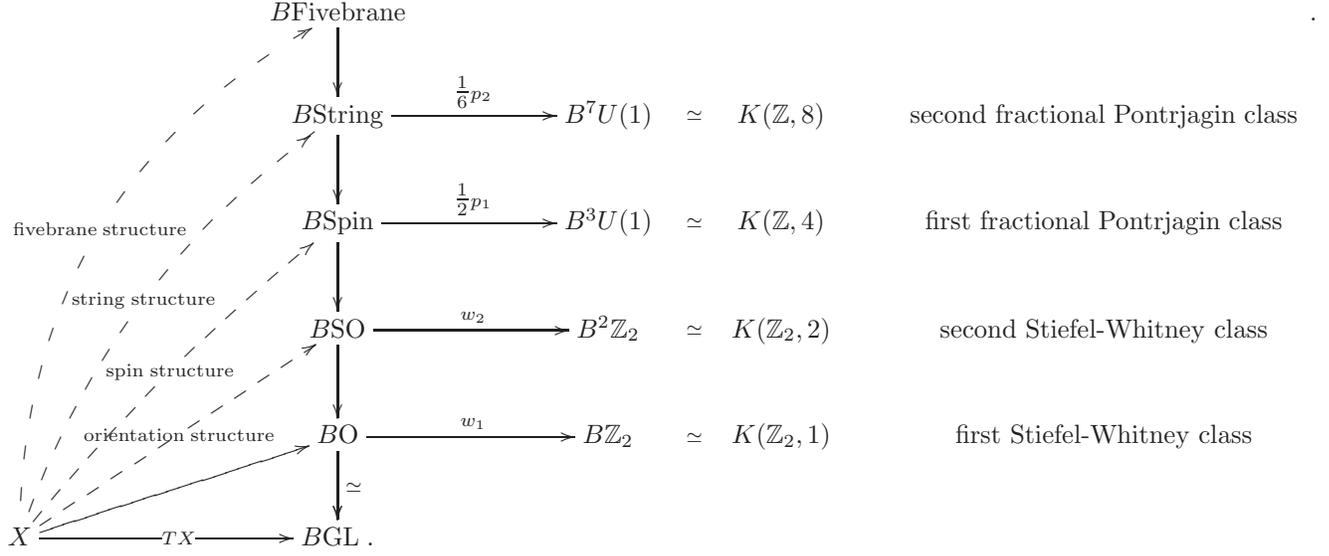
\xymatrix{
    &&& B \mathrm{Fivebrane}  \ar[d]
    \\  
    &&& B \mathrm{String} \ar[d]\ar[rr]^{\tfrac{1}{6}p_2} && B^7 U(1) \ar@{}[r]|\simeq & K(\mathbb{Z},8)    
	& 
    \mbox{second fractional Pontrjagin class}    
    \\
    &&& B \mathrm{Spin} \ar[d]\ar[rr]^{\tfrac{1}{2}p_1} && B^3 U(1) \ar@{}[r]|\simeq & K(\mathbb{Z},4)
	& 
    \mbox{first fractional Pontrjagin class}    
    \\
    &&& B \mathrm{SO} \ar[d]\ar[rr]^{w_2} && B^2 \mathbb{Z}_2 \ar@{}[r]|\simeq & K(\mathbb{Z}_2,2)
	&
	\mbox{second Stiefel-Whitney class}
    \\
    &&& B \mathrm{O} \ar[d]\ar[rr]^{w_1} \ar[d]^{\simeq} && B \mathbb{Z}_2 \ar@{}[r]|\simeq & K(\mathbb{Z}_2,1)
	&
	\mbox{first Stiefel-Whitney class}
    \\
    X 
	  \ar[rrr]|{T X } 
	  \ar@{->}[urrr]|{}
	  \ar@{-->}[uurrr]|{\mathrm{orientation}\;\mathrm{structure}}
	  \ar@/^.4pc/@{-->}[uuurrr]|{\mathrm{spin}\;\mathrm{structure}}
	  \ar@/^1.4pc/@{-->}[uuuurrr]|{\mathrm{string}\;\mathrm{structure}}
	  \ar@/^2.9pc/@{-->}[uuuuurrr]|{\mathrm{fivebrane}\;\mathrm{structure}}
	&&& B \mathrm{GL}\;.
  }
  \,.
$$

\vspace{2mm}
\noindent Here the horizontal morphisms denote representatives of
universal characteristic classes. These are such that the sub-diagrams on the right of the form
$$
  \xymatrix{
     B \widehat G
	 \ar[d]
	 \\
	 B G \ar[rr]^c && B^n K
  }
$$
are homotopy fiber sequences, i.e., these universal characteristic classes represent the obstructions to lift a given structure to the ``upper level'' in the Postnikov tower. 
Several variations and twists on the above structures are considered in 
\cite{S2, S3, S4}.

\medskip
In \cite{FSS} we gave an explicit construction of the smooth refinement of the
second fractional Pontrjagin class to a morphism of smooth moduli stacks
$$
  \tfrac{1}{6}(\mathbf{p}_2)_{\mathrm{conn}}
  :
  \mathbf{B}\mathrm{String}_{\mathrm{conn}}
  \to 
  \mathbf{B}^7 U(1)_{\mathrm{con}}
$$
from that of String 2-connections to that of circle 7-bundles with connection.
When regarding this as the fully extended action functional of an 
$\infty$-Chern-Simons theory it produces a 7-dimensional theory which in 
\cite{FiorenzaSatiSchreiberI} we argued is part of the holographic dual
of the M5-brane theory, see section \ref{7dSugra} below.
As before, it is of interest to 
compute the total space of the prequantum circle 7-bundle on the moduli 2-stack
of $\mathrm{String}$-connections. 
By Corollary \ref{TotalSpaceOfnConnection} and after comparison with
\cite{SSSIII}, this is the moduli 6-stack of (twisted) 
{$\mathrm{Fivebrane}$-principal 6-connections.
$$
  \xymatrix{
     \mathbf{B}\mathrm{Fivebrane}_{\mathrm{conn}'}
	 \ar[r]
	 \ar[d]
	 &
	 \Omega^{1 \leq \bullet \leq 7}
	 \ar[r]
	 \ar[d]
	 &
	 {*}
	 \ar[d]
	 \\
	 \mathbf{B}\mathrm{String}_{\mathrm{conn}}
	 \ar[r]_-{\tfrac{1}{6}(\mathbf{p}_2)_{\mathrm{conn}}}	
	 &
	 \mathbf{B}^7 U(1)_{\mathrm{conn}}
	 \ar[r]_-{\chi}
	 &
	 \mathbf{B}^7 U(1)~.
  }
$$

Another important example is the Whitehead tower of $\mathrm{U}(n)$: 
the $k$-connected cover 
$\mathrm{U}(n)\langle 2k-1\rangle\simeq \mathrm{U}(n)\langle 2k\rangle$ is 
the natural home for the
differential refinement $({\mathbf{c}}_{k+1})_{\mathrm{conn}}:\mathbf{B}\mathrm{U}(n)\langle 2k\rangle_\mathrm{conn}\to \mathbf{B}^{2k+1}\mathrm{U}(1)_\mathrm{conn}$ of the $(k+1)$st Chern class $c_{k+1}\in H^{2k+2}(B\mathrm{U}(n);\mathbb{Z})$. Constructions analogous to those of the orthogonal case 
follow similarly. 

\subsubsection{$(2n+1)d$ Chern-Simons (super)gravity and $\mathrm{WZW}_{2n}$-models}

A remarkable property of 3-dimensional pure gravity, described by Chern-Simons theory,
is that it forms an exactly solvable system \cite{Soluble}. 
The literature contains various proposals of 
higher-dimensional (super) Chern-Simons-type theories, all \emph{unary}
in our sense here, that are argued to be possible candidates for a
theory related to actual (super)gravity \cite{BTZ, TZ}, see \cite{Zanelli}
for a review. 
In codimension 1 these theories are known to be related to 
higher dimensional analogs of the 2d WZW-model in dimension $2n$ 
\cite{BGH, GK}.
In the case of M-theory, with $n=5$, there are candidates that 
propose to describe the theory based on holography and Chern-Simons 
theory \cite{Horava, Nastase, IR}. 

\medskip
These unary nonabelian higher dimensional Chern-Simons theories
are interesting candidates for extended prequantization as considered here, but whether or
in which cases their fully extended prequantizations exist has not been worked out yet.

\subsection{Quadratic examples}

We now consider examples of extended $\infty$-Chern-Simons theories 
that are formed by the differential cup-product of two 
factors.
\subsubsection{3d $U(1)$-theory with two species and differential T-duality}
\label{DiffT}

Consider the extended $\infty$-Chern-Simons action functional given simply by the differential
cup product of def. \ref{ExtendedDifferentialCup} in the first non-trivial degree:
$$
  (-) \mathrm{\cup}_{\mathrm{conn}} (-)
  :
  \xymatrix{
    \mathbf{B}U(1)_{\mathrm{conn}}
    \times
    \mathbf{B}U(1)_{\mathrm{conn}}
	\ar[r]
	&
	\mathbf{B}^3 U(1)
  }
  \,.
$$ 
Its moduli stack of fields 
$[\Sigma_3, \mathbf{B}U(1)_{\mathrm{conn}} \times \mathbf{B}U(1)_{\mathrm{conn}}]$
consists of pairs of two different $U(1)$-gauge fields on $\Sigma_3$.
On those field configurations that have trivial underlying integral classes and are hence
given by globally defined 1-forms $A_1, A_2$, the action functional in dimension 3
takes these to
$$
  \exp(2 \pi i \int_{\Sigma_3}[\Sigma_3, (-)\cup_{\mathrm{conn}}(-)])
  :
  (A_1, A_2)
  \mapsto
  \exp(2 \pi i \int_{\Sigma_3} A_1 \wedge d A_2)
  =
  \exp(2 \pi i \int_{\Sigma_3} A_2 \wedge d A_1)
  \,.
$$
The ``diagonal of this theory'', namely the extended action functional obtained by
precomposition with the diagonal map
$\Delta : \mathbf{B}U(1)_{\mathrm{conn}} \to \mathbf{B}U(1)_{\mathrm{conn}}\times \mathbf{B}U(1)_{\mathrm{conn}}$ is the ordinary 3d $U(1)$-Chern-Simons theory of a single gauge field
species discussed below in section \ref{3dU1CS}.

\medskip
By Corollary \ref{TotalSpaceOfnConnection} 
the total space object $P$ of the prequantum circle 3-bundle of 
the above extended action functional is the 
homotopy pullback
$$
  \raisebox{20pt}{
  \xymatrix{
    P \ar[rr] \ar[d] && \Omega^{1 \leq \bullet \leq 3}(-) \ar[r] \ar[d] & {*} \ar[d]
	\\
	\mathbf{B}U(1)_{\mathrm{conn}}
	\times
	\mathbf{B}U(1)_{\mathrm{conn}}
    \ar[rr]^-{\cup_{\mathrm{conn}}}
	&&
	\mathbf{B}^3 U(1)_{\mathrm{conn}}
	\ar[r]
	&
	\mathbf{B}^3 U(1)~.
  }
  }
$$
By the universal property of the homotopy pullback this means that $P$ is the
moduli 2-stack for pairs $(\nabla_1,\nabla_2)$ of circle bundles with connection -- 
hence pairs of \emph{1-torus bundles} with connection -- equipped with a 
smooth trivialization
of the cup product 
$$
  \mathbf{c}_1(\nabla_1) \cup \mathbf{c}(\nabla_2)
  = \chi(\nabla_1) \cup \chi(\nabla_2)
$$
of their Chern classes. This is the
structure called a \emph{differential T-duality pair} 
in Def. 2.1 of \cite{KahleValentino}, expressing the necessary differential geometric
structure for an action of T-duality between two torus fibrations 
on the \emph{differential K-theory} of the underlying spaces, hence on the 
charge-quantized RR-fields\footnote{Here
we do not discuss the smooth stacky formulation of 
RR-fields, hence of differential K-cocycles, since this is 
a topic well beyond the scope of the present article.
To obtain such one has to provide a smooth stacky delooping
$\mathbf{B}_\oplus \mathbf{B}U$ of the moduli stack of 
stable unitary bundles with respect to their direct sum, such that the differential
refinement of the resulting abelian geometric cohomology theory as in 
\cite{cohesive} is a stacky refinement of differential K-theory.} in type II string theory.

\subsubsection{Ordinary 3d $U(1)$-Chern-Simons theory and generalized $B_n$-geometry}
  \label{3dU1CS}

As remarked above, ordinary 3-dimensional $U(1)$-Chern-Simons theory 
on a closed oriented manifold $\Sigma_3$ 
contains field configurations which are given by globally defined 1-forms
$A \in \Omega^1(\Sigma_3)$ and on which the action functional is given by
the familiar expression
$$
  \exp(i S(A)) = \exp(2 \pi i k \int_{\Sigma_3} A \wedge d A)
  \,.
$$
More generally, though, a field configuration of the theory is a connection 
$\nabla$ on 
a $U(1)$-principal bundle $P \to \Sigma_3$ and this simple formula
is modified, from being the exponential of the ordinary integral of 
the wedge product of two differential forms, to the fiber integration
in differential cohomology 
(section \ref{FibInt}) 
of the differential cup-product 
(Def. \ref{CupOnStacks}):
$$
  \exp(i S(\nabla)) = \exp(2 \pi i k \int_{\Sigma_3} \nabla \cup_{\mathrm{conn}} \nabla)
  \,.
$$
This defines the action functional on the set 
$H^1_{\mathrm{conn}}(\Sigma_3, U(1))$ of equivalence classes of $U(1)$-principal bundles
with connection
$$
  \exp(i S(-)) : H^1_{\mathrm{conn}}(\Sigma_3) \to U(1)
  \,.
$$
That the action functional is gauge invariant means that it extends from
a function on gauge equivalence classes to a functor on the groupoid 
$\mathbf{H}^1_{\mathrm{conn}}(\Sigma_3, U(1))$, whose objects are actual
$U(1)$-principal connections, and whose morphsims are smooth gauge transformations
between these:
$$
  \exp(i S(-)) : \mathbf{H}^1_{\mathrm{conn}}(\Sigma_3) \to U(1)
  \,.
$$
Finally, that the action functional depends \emph{smoothly} on the connections
means that it extends further to the moduli stack of fields to a morphism of stacks
$$
  \exp(i S(-)) : [\Sigma_3, \mathbf{B} U(1)_{\mathrm{conn}}] \to U(1)
  \,.
$$

\medskip
The fully extended prequantum circle 3-bundle of this extended 3d Chern-Simons theory
is that of the two-species theory in section \ref{DiffT}, restricted along the 
diagonal $\Delta : \mathbf{B}U(1)_{\mathrm{conn}} 
\to \mathbf{B}U(1)_{\mathrm{conn}}\times \mathbf{B}U(1)_{\mathrm{conn}}$.
This is the homotopy fiber of the smooth cup square in these degrees.

\medskip
According to \cite{Hitchin} 
aspects of the differential geometry of 
the homotopy fiber of a differential refinement of this 
cup square are captured by the ``generalized geometry of $B_n$-type''
that was suggested in section 2.4 of \cite{Baraglia}.
In view of the relation of the same structure to differential T-duality discussed above
in section \ref{DiffT} one is led to expect that ``generalized geometric of $B_n$-type''
captures aspects of the differential cohomology
on fiber products of torus bundles that exhibit auto T-duality on differential K-theory. 
Indeed, such a relation is pointed out in \cite{Bouwknegt}\footnote{
Thanks, once more, to Alexander Kahle, for discussion of this point, at \emph{String-Math 2012}.}. 

\subsubsection{$(4k+3)d$ $U(1)$-Chern-Simons theory and self-dual $(2k+1)$-form field theory}
\label{4k+3}
\label{7dSugra}

The differential cup square in general degree
$$
  (-)^{\cup_{\mathrm{conn}}^2}
  :
  \mathbf{B}^{2k+1}U(1)_{\mathrm{conn}}
  \to
  \mathbf{B}^{4k+3}U(1)_{\mathrm{conn}}
$$
for any $k \in \mathbb{N}$ reduces in codimension 0 and on cohomology classes
to the action functional
$$
  \exp(2 \pi i \int_{\Sigma_{4k+3}} [\Sigma_{4k+3}, (-)^{\cup_{\mathrm{conn}}^2}])
  :
  H^{2k+2}_{\mathrm{conn}}(\Sigma_{4k+3})
  \to 
  U(1)
$$
on differential cohomology that exhibits $(4k+3)$-dimensional $U(1)$-Chern-Simons theory,
as it
is generally considered. See,  for instance,  \cite{FP,HopkinsSinger}.
For $k = 0$ this is the 3-dimensional system from section \ref{3dU1CS}.
For a general $k$, its space of quantum states in codimension 1 gives the conformal
blocks of self-dual $(2k+1)$-form gauge theory on $\Sigma_{4k+2}$ 
-- this is the higher Chern-Simons \emph{holography}, as discussed
generally in \cite{BM, BM2} and for the case of $k=1$ famously in \cite{Witten96, Witten98}.

\medskip
We briefly recall a few instances of  self-dual theories of this kind, for $k$ ranging from $0$ to $3$.
  In some cases of self-dual theories, notably 
  for the 6d-theory on the Fivebrane, one wants to further divide the 
  differential cup class after adding a background shift, 
  hence perform a quadratic refinement of the cup square in the holographic
  dual Chern-Simons type theory.
  This was originally argued for in \cite{Witten96} and then formalized 
  in terms of differential cohomology classes in \cite{HopkinsSinger}
  via \emph{integral Wu structures}. In \cite{FiorenzaSatiSchreiberII}
  we showed how this is further refined to differential cocycles on moduli stacks
  of fields, as used here, by stacky \emph{differential Wu structures}. 
In section \ref{CSWithBackgroundCharge} we provided the conceptual and 
theoretical foundation for such a refinement. We highlight this explicitly by 
example \ref{CScharge} in six dimensions at the end of this section. The other cases
in other dimensions would work in a similar way, with obvious changes.

\paragraph{$k=0$: the self-dual scalar field in 3 dimensions.}
The action for the scalar field $\phi$ in two dimensions 
is $d \phi \wedge *d \phi$. The partition function of this field 
can be described via 3-dimensional Chern-Simons theory, taking the form
\(
i\int_{Y^3}\mathrm{CS}_1(A) \cup d \mathrm{CS}_1(A)=i\int_{Y^3}\mathrm{tr}(A) \wedge F_A,
\)
where the curvature 2-form $F_A$ is a representative for the first Chern of a complex line bundle.

\paragraph{$k=1$: the $6d$ self-dual theory on the M5-brane.}
The action functional of classical 11-dimensional supergravity 
contains a \emph{cubic} abelian Chern-Simons term, recalled below
in section \ref{11dSupergravity}. After compactification on a 
four-sphere $S^4$ this becomes an abelian 7-dimensional quadratic Chern-Simons term,
an example of the above system for $k = 1$. In \cite{Witten96,Witten98} 
it is argued that this topological term alone in the full supergravity
action functional determines the conformal blocks of the 
$(0,2)$-superconformal field theory on a single M5-brane under
$\mathrm{AdS}_7/\mathrm{CFT}_6$-duality.
But if the 11-dimensional quantum corrections are taken into account,
the 11-dimensional Chern-Simons term is accompanied by 
further terms which after reduction to 7 dimensions involve a cup
product of a nonabelian 3d Chern-Simons theory with itself, as in section \ref{3dCSSpin}, whose action thus locally reads \cite{Witten98}
\(
 -i \tfrac{N}{4\pi} \int_{Y^7}\mathrm{CS}_3(A) \cup d \mathrm{CS}_3(A)\;= -i \tfrac{N}{4\pi}\int_{Y^7}\left(\langle A,dA\rangle+\tfrac{1}{3}\langle A, [A,A]\rangle\right)\wedge \langle F_A,F_A\rangle\,,
 \) 
as well as an indecomposable 7-dimensional term.
In \cite{FiorenzaSatiSchreiberI,FiorenzaSatiSchreiberII} we argued that if furthermore
the \emph{flux quantization} of the supergravity $C$-field is taken into account, then 
the quantum-corrected 7d Chern-Simons action that is holographically dual
to the M5-brane theory is defined on \emph{String 2-form fields}
as in section \ref{UnaryExamples}.

\paragraph{$k=2$: Ramond-Ramond fields in type IIB string theory.}
Type II RR fields are self-dual. The relation between the RR
partition function and the Chern-Simons theory in eleven dimensions is 
explained in \cite{BM} (see also \cite{S-gerbe}). 
The action is of the form $\int_{Y^{11}}F_5 \wedge dF_5$ and 
the quantization condition of the Ramond-Ramond fields 
implies that these fields are given essentially by the Chern character: 
$F_5= {\rm ch}(E) \sqrt{A(X)}$, where $E$ is the Chan-Paton bundle
\cite{MW}.
The way the Chern character is to be interpreted is by extending by 
a circle to one dimension higher. 
Alternatively, one can view $F_5$ as a ``composite connection" for 
a degree six field strength \cite{Witten96}. 
By identifying $F_5$ with the Chern-Simons 5-form $\mathrm{CS}_5$, one
sees that the Chern-Simons action is indeed of the form 
$\int_{Y^{11}}\mathrm{CS}_5 \wedge d \mathrm{CS}_5$.

\paragraph{$k=3$: Fivebrane structures and 15-dimensional theories.}
One could continue this pattern in the obvious way. For example, 
one could consider $\mathrm{CS}_7(A) \cup d \mathrm{CS}_7(A)$ with $d\mathrm{CS}_7(A)$ an $8$-form representative for the second Pontrjagin class of a String bundle
\cite{SSSI,SSSII}.
With the right normalization constant $\kappa$, one associates this 15-dimensional action to 
the Fivebrane structure \cite{SSSI}. 
A lift of this to sixteen dimensions would have the form $x_8 \cup x_8$, an instance of 
which is studied in \cite{OP2} in the lift of M-theory to higher dimensions. 

The following example spells out the mixed-term version of the quadratic
refinement that we discussed above in 
section \ref{CSWithBackgroundCharge}. 
(This is a simplified version of the quadratic refinement of 
shifted ``flux quantized'' fields. For a full discussion of the latter
see \cite{FiorenzaSatiSchreiberII}.)

\begin{example} 
Let $\Sigma_7$ be a closed 7-dimensional Spin manifold, $G$ a simple and simply connected compact Lie group and $\hat{\mathbf{c}}:\mathbf{B}G_{\mathrm{conn}}\to \mathbf{B}^3U(1)_{\mathrm{conn}}$ a refinement of a degree 4 characteristic class for $BG$ to a morphism of stacks. The tangent bundle of $\Sigma_7$ defines a characteristic map $T\Sigma_7:\Sigma_7\to \mathbf{B}SO(7)$ and the datum of the Spin structure is the datum of a lift of $T\Sigma_7$ to a map $\Sigma_7\to \mathbf{B}\mathrm{Spin}(7)$. Moreover, the Levi-Civita connection gives a further lift $\nabla_{\mathrm{LC}}$ to $\mathbf{B}\mathrm{Spin}(7)_{\mathrm{conn}}$. In terms of moduli stacks of field configurations, this amounts to saying that $[\Sigma_7,\mathbf{B}\mathrm{Spin}(7)_{\mathrm{conn}}]$ is a pointed stack, with distinguished point given by the Levi-Civita connection. We can then consider the sequence of morphism of stacks
\begin{align*}
[\Sigma_7,&\mathbf{B}G_{\mathrm{conn}}]\cong [\Sigma_7,\mathbf{B}G_{\mathrm{conn}}]\times *\xrightarrow{(\mathrm{id},\nabla_{\mathrm{LC}})}[\Sigma_7,\mathbf{B}G_{\mathrm{conn}}]\times [\Sigma_7,\mathbf{B}\mathrm{Spin}(7)_{\mathrm{conn}}]\xrightarrow{[\Sigma_7,(\hat{\mathbf{c}},\frac{1}{2}\hat{\mathbf{p}}_1)]}\\
&\xrightarrow{[\Sigma_7,(\hat{\mathbf{c}},\frac{1}{2}\hat{\mathbf{p}}_1)]}[\Sigma_7,\mathbf{B}^3U(1)_{\mathrm{conn}}\times \mathbf{B}^3U(1)_{\mathrm{conn}}]\xrightarrow{[\Sigma_7,(\hat{x},\hat{y})\mapsto \hat{x}\cup_{\mathrm{conn}}\hat{x}-\hat{x}\cup_{\mathrm{conn}}\hat{y}]}[\Sigma_7,\mathbf{B}^7U(1)_{\mathrm{conn}}]\xrightarrow{\mathrm{hol}_{\Sigma_7}}\\
&\xrightarrow{\mathrm{hol}_{\Sigma_7}}\mathbf{U}(1)\;,
\end{align*}
which is the stacky quadratic refinement of the cup-product intersection pairing of $G$-Chern-Simons theories and which, for a topologically trivial $G$ bundle over $\Sigma_7$ has the form
\[
A\mapsto \exp\left(\pi i\int_{\Sigma_7}\mathrm{CS}_3(A)\wedge \left(\mathrm{tr}(F_A)-\tfrac{1}{2}p_1(T\Sigma_7)\right)\right)\;,
\]
where $A\in \Omega^1(\Sigma_7;\mathfrak{g})$ is the $G$-connection 1-form and $\mathrm{CS}_3(A)$ and $F_A$ are its Chern-Simons 3-form and curvature 2-form, respectively. 
\label{CScharge}
\end{example}

\subsubsection{The cup-product of two extended CS theories and the higher charge anomaly}
\label{CupOfTwo}

We have already discussed the interpretation of  the differential 
cup product from def. \ref{CupOnStacks} as an extended action functional
$$
  (-)\cup_{\mathrm{conn}} (-)
  :
  \mathbf{B}^p U(1)_{\mathrm{conn}}
  \times 
  \mathbf{B}^q U(1)_{\mathrm{conn}}
  \to
  \mathbf{B}^{p+q+1}U(1)_{\mathrm{conn}}
  \,.
$$
By itself this encodes higher Maxwell charge anomalies
in terms of extended Chern-Simons theory. 
We briefly recall what this looks like in 
heterotic string theory.
 See the third 
section of \cite{Lectures} for more exposition in the present context.

\paragraph{Charge anomaly in heterotic string theory.}
The \emph{local anomaly} term (the curvature of the fully extended
action functional on the moduli stack of fields) in this example is a 12-form 
$I_4(F, R) \wedge I_8(F, R)$, where $I_4(F, R)$ and $I_8(F, R)$ are the 
Green-Schwarz anomaly polynomials in degree 4 and 8, respectively, in 
terms of the curvature $R$ of the tangent bundle and the curvature $F$ of the 
gauge bundle (see \cite{Freed}). These terms are given essentially by a difference of 
first Pontrjagin classes and a difference of second Pontrjagin classes, respectively. 
In simplified form (as in \cite{SSSII, SSSIII}) 
\(
I_4(F, R)={\rm ch}_2(F) - \tfrac{1}{2}p_1(R)\;, \qquad \qquad
I_8(F, R) = {\rm ch}_4(F) - \tfrac{1}{48}p_2(R)\;. 
\)
The trivializations  $I_4(F, R)=dH_3$  and $I_8(F, R)=dH_7$ are of the form 
\(
H_3= CS_3(A) - \tfrac{1}{2}CS_3(\nabla)\;, \qquad \qquad
H_7=CS_7(A) - \tfrac{1}{48}CS_7(\nabla)\;,
\)
where $A$ is the gauge connection with curvature $F$ and 
 $\nabla$ is the Spin connection with curvature $R$. 

\medskip
Thus, in eleven dimensions, 
this is a cup product Chern-Simons theory, which can be written as 
the integral of local date of the form
\(
[\mathrm{CS}_3(\nabla)- \tfrac{1}{2}CS_3(\nabla)]\wedge
d[ CS_7(A) - \tfrac{1}{48}CS_7(\nabla)]\;,
\)
or, dually, 
\(
[CS_7(A) - \tfrac{1}{48}CS_7(\nabla)]
\wedge d[\mathrm{CS}_3(\nabla)- \tfrac{1}{2}CS_3(\nabla)]\;.
 \) 
Notice that, since 
 both the gauge bundle and the 
tangent bundle are involved, the Chern-Simons action term is of the mixed type.
In particular, from the mixed terms we see that 
we have a new type of examples of the form 
$CS(\omega_1) \wedge dCS(\omega_2)$ for 
two different connections $\omega_1$ and $\omega_2$.

\subsection{Higher order examples}
\label{Cubic examples}

We have seen so far examples that are the cup products of two copies of the 
same or different Chern-Simons theories. One might wonder whether 
more than two terms can naturally occur.
There are at least two remarkable examples of abelian Chern-Simons theories 
where there are three terms in the action. 

\subsubsection{$5d$ supergravity }

The topological part of pure five-dimensional $N=2$ supergravity resembles 
that of M-theory, except that a connection 1-form $A_1$ replaces the 
C-field. This topological term is locally given by
\(
\int_{Y^5} A_1 \wedge F_2 \wedge F_2\;,
\)
where $F_2=dA_1$ is the curvature of the $U(1)$-connection $A_1$. 
Globally this means that $\hat F_2$ is a differential 2-cocycle 
with curvature $F_2$ and that $A_1$ is locally the corresponding
abelian Chern-Simons term.
In terms of this refinement to differential cohomology the above is
globalized to the the cubic cup term
\(
\int_{Y^5} \hat F_1 \cup \hat F_2 \cup \hat F_2\;, 
\)
i.e. a 3-fold Chern-Simons theory. Thus this falls under our 
formulation and hence admits a refinement to the corresponding 
moduli stacks of supergravity fields. 

\subsubsection{$11d$ supergravity}
\label{11dSupergravity}

The topological aspects of this supergravity theory allows for 
a glimpse at the elusive M-theory. An ingredient which allows for this
is the Chern-Simons term for the C-field given by
\(
\tfrac{1}{6}\int_{Y^{11}} C_3 \wedge G_4 \wedge G_4\;,
\label{CGG}
\)
where $G_4$ is the field strength of the C-field 3-form $C_3$. 
Geometrically, 
this can be seen as the curvature 4-form of a connection $\hat G_4$ on a $U(1)$-2-gerbe.
Therefore, refined to differential cohomology, the above action takes the form 
of a three-term cup-product $\frac{1}{6}\int_{Y^{11}} \hat G_4 \cup \hat G_4 \cup \hat G_4$ of the type \eqref{Eq ktimes} for $k=2$. Note that
the C-field is essentially a Chern-Simons 3-form $CS_3(A)$ for a connection 
1-form $A$ which admits a refinement to moduli 3-stacks (see \cite{FiorenzaSatiSchreiberII}).
The total term  
\eqref{CGG} thus admits a refinement in the sense of  higher 
cup-product Chern-Simons theories.

\subsubsection{Final remarks}
\label{SectionHigher}

\paragraph{Cup product Chern-Simons theories on manifolds with corners}

The Chern-Simons theory considered in Sec. \ref{CupOfTwo}
can be reduced by one dimension, as is the 
case in the above systems. 
One can further reduce dimension by one, by working in the context
of manifolds of corners of codimension 2, as explained in \cite{S-corner}. 
On these codimension 2 corners one then has terms of the form 
$\mathrm{CS}_3 \wedge \mathrm{CS}_7$; see \cite{S-f} for more details
and for topological significance.

\medskip 
One can more generally consider an arbitrary number $k$ of
terms in the cup product. The pattern that 
emerges is a generalization of the heterotic anomaly
cancellation discussed above, where the anomaly takes the form 
of the wedge product of two Chern characters ${\rm ch}_{n_1}$ 
and ${\rm ch}_{n_2}$, to more terms, that is to 
\(
S_Z=\int_{Z^{n_1 + n_2 + \cdots n_k}} 
{\rm ch}_{n_1} \wedge {\rm ch}_{n_2} \wedge \cdots \wedge 
{\rm ch}_{n_k}\;.
\label{many}
\) 
With the local formula ${\rm ch}_{n_i}=dCS_{2n_i +1}$ and 
passing to differential cohomology, we can write each of the factors in 
\eqref{many} in terms of $CS_{2n_i +1}$, for $i=1, \cdots, k$. This involves using 
a type of Stokes formula for various faces in codimension $k$, in the
setting advocated in \cite{S-corner,Coh, S-f}. That is, 
we take $Z^{n_1 + n_2+ \cdots + n_k}$ to admit a codimension-$k$
 corner $X^{n_1 + n_2 + \cdots + n_k -k}$, on which the action schematically 
 takes the 
 form 
 \(
 S_X= \int_{X^{n_1 + n_2 + \cdots + n_k -k} }
 CS_{2n_1 +1} (A_1)
 \wedge 
 CS_{2n_2 +1} (A_2)
 \wedge
 \cdots
 \wedge
 CS_{2n_k +1} (A_k)\;.
 \)
The detailed study of such 
systems is currently under investigation.

\paragraph{The classification of abelian Spin $n$-fold Chern-Simons theories}
Classification of general Chern-Simons theories is a formidable task. 
Three-dimensional abelian Spin Chern-Simons theories with structure group $U(1)^N$ have
been classified by Belov and Moore \cite{BM1}. 
This classification of quantum theories involves three invariants, one of which 
is a quadratic form. 
It is natural to ask what the 
corresponding classification for cup products of such theories would be. 
We do not attempt an answer to this question here, but merely point out that 
that such an extension should involve 
a correspondence with higher forms, that is beyond quadratic forms. More precisely, a $k$-fold
Chern-Simons theory is expected to involve a $k$-ary form.

\vspace{5mm}
\noindent
{\bf \large Acknowledgements}

\vspace{1mm}

\noindent The research of H.S. is supported by NSF Grant PHY-1102218. We thank the referee for important suggestions and constructive criticism which helped us a lot in revising a first version of this article.


\end{document}